\documentclass[11pt]{article}

\usepackage[margin = 1in]{geometry}

\usepackage{amsmath,amssymb,amsthm,color,fixmath}
\usepackage{hyperref}
\hypersetup{colorlinks=true,citecolor=blue, linkcolor=blue, urlcolor=blue}

\usepackage{graphicx}
\usepackage[bbgreekl]{mathbbol}

\newtheorem{lemma}{Lemma}
\newtheorem{theorem}[lemma]{Theorem}

\newtheorem{proposition}[lemma]{Proposition}
\newtheorem{claim}[lemma]{Claim}

\newtheorem{corollary}[lemma]{Corollary}
\newtheorem{remark}{Remark}
\newtheorem{definition}[remark]{Definition}

\newcommand{\eps}{\epsilon}

\newcommand{\dist}{\mathrm{dist}}
\newcommand{\diameter}{\mathrm{diameter}}
\newcommand{\distance}{{\tt dist}}
\newcommand{\Tmix}{T_\mathrm{mix}}
\renewcommand{\P}{\mathcal{P}}
\newcommand{\degree}{\mathrm{deg}}

\newcommand{\fpaus}{\mathsf{FPAUS}}
\newcommand{\ball}{{\tt B}}

\newcommand{\Expectation}{\mathrm{E}}
\newcommand{\Variance}{\mathrm{Var}}
\newcommand{\Exp}[1]{{\Expectation\left[{#1}\right]}}

\newcommand{\ExpCond}[2]{{\Expectation\left[{#1} \mid {#2} \right]}}

\newcommand{\Probability}{\mathrm{Pr}}
\newcommand{\Prob}[1]{{\Probability\left[{#1}\right]}}

\newcommand{\indicator}[1]{{\;\mathbf{1}\!\left(#1\right)}}
\newcommand{\U}{\mathcal U}
\newcommand{\B}{\mathcal B}
\newcommand{\TDeg}{\widehat{d}}
\newcommand{\outBound}{\partial_{\mathrm{out}}}
\newcommand{\inBound}{\partial_{\mathrm{in}}}

\newcommand{\EdgeBlockWeight}{\distance}
\newcommand{\poly}{\mathrm{poly}}

\newcommand{\Dis}{D}
\newcommand{\ImpVrtx}{{\partial {\cal B}}}

\newcommand{\ManyDisEvent}{{\cal E}}
\newcommand{\OutOfBalls}{{\cal B}_1}
\newcommand{\FewColors}{{\cal B}_2}
\newcommand{\HighDegreeSet}{\mathbold{L}}
\newcommand{\parent}[1]{{\tt Parent}(#1)}
\newcommand{\IntrstGraphFam}{{\cal F}}
\newcommand{\AvialColors}{A}

\title{Sampling Random Colorings of Sparse Random Graphs}

\author{
Charilaos Efthymiou\thanks{Goethe University, Frankfurt am Main, Germany. Email:
efthymiou@gmail.com. Research supported by DFG grant EF 103/11.}
\and
Thomas P. Hayes\thanks{University of New Mexico, USA. Email: hayes@cs.unm.edu.}
\and
Daniel \v{S}tefankovi\v{c}\thanks{University of Rochester, USA.
 Email: stefanko@cs.rochester.edu. Research
supported in part by NSF grant CCF-1318374.}
\and 
Eric Vigoda\thanks{Georgia Institute of Technology, USA.  Email: ericvigoda@gmail.com. Research supported in part by NSF grants CCF-1617306 and CCF-1563838.
}
}

\begin{document}

\maketitle

\begin{abstract}
We study the mixing properties of the single-site Markov chain known as the
Glauber dynamics for sampling $k$-colorings of a sparse random graph $G(n,d/n)$
for constant $d$.  The best known rapid mixing results for general graphs 
are in terms of the maximum degree $\Delta$ of the input graph $G$
and hold when $k>11\Delta/6$ for all $G$.  Improved results hold when $k>\alpha\Delta$
for graphs with girth $\geq 5$ and $\Delta$ sufficiently large where $\alpha\approx 1.7632\ldots$
is the root of $\alpha=\exp(1/\alpha)$; further improvements on the constant $\alpha$
hold with stronger girth and maximum degree assumptions.
For sparse random graphs the maximum degree is a function of $n$ and the goal is to
obtain results in terms of the expected degree $d$.  The following
rapid mixing results for $G(n,d/n)$ hold with high probability over the choice of the random graph 
for sufficiently large constant~$d$.  Mossel and Sly (2009)
proved rapid mixing for constant $k$, and Efthymiou (2014) improved this to $k$ linear in~$d$.
The condition was improved to $k>3d$ by Yin and Zhang (2016) using non-MCMC methods.
Here we prove rapid mixing when $k>\alpha d$ where $\alpha\approx 1.7632\ldots$ is
the same constant as above.
Moreover we obtain $O(n^{3})$ mixing time of the Glauber dynamics, while in previous
rapid mixing results the exponent was an increasing function in $d$.
As in previous results for random graphs our proof analyzes an appropriately defined
block dynamics to ``hide'' high-degree vertices.  One new aspect in our improved 
approach is utilizing so-called local uniformity properties for the analysis of block dynamics.
To analyze the ``burn-in'' phase we 
prove a concentration inequality
for the number of disagreements propagating in large blocks.

\end{abstract}

\thispagestyle{empty}

\newpage

\setcounter{page}{1}

\section{Introduction}

Sampling from Gibbs distributions is an important problem in many contexts.
For example, in theoretical computer science sampling algorithms are 
often the key element in approximate counting algorithms, 
in statistical physics Gibbs distributions describe the equilibrium state of large physical systems,
and in statistics they are used for Bayesian inference.
In this paper we focus on random colorings, which are an
example of a spin system, corresponding to the zero-temperature
limit of the anti-ferromagnetic Potts model.  The natural combinatorial structure of colorings
makes it a nice testbed for studying connections to statistical physics phase
transitions and its study has led to many new techniques.

Given a graph $G=(V,E)$ of maximum degree $\Delta$
and a positive integer $k$, can we generate a random $k$-coloring of $G$
in time polynomial in $n=|V|$?  To be precise, let $\Omega=\Omega_G$ denote
the set of proper vertex $k$-colorings of $G$, and let $\pi$ denote the
uniform distribution over $\Omega$.  Our goal is to obtain an $\fpaus$ (fully polynomial-time
approximate uniform sampling scheme)
for sampling from $\pi$: given $\delta>0$ in time $\poly(n,\log(1/\delta))$
generate a coloring $X$ from a distribution $\mu$ which is within variation
distance $\leq\delta$ of the uniform distribution $\pi$.

The Glauber dynamics is a simple and well-studied algorithm for sampling colorings, and more
generally, for spin systems it is of particular interest as a model of how a physical
system approaches equilibrium.  
The dynamics is the following single-site spin update Markov chain $(X_t)$ with 
state space $\Omega$.  We
present here the heat-bath version, but our results are
robust and hold for other versions as well.   
The Markov chain $(X_t)$ has the following transitions $X_t\rightarrow X_{t+1}$:
from $X_t$, choose a random vertex $v$, and a random color $c$ not 
appearing in the current neighborhood of $v$, i.e., from $[k]\setminus X_t(N(v))$.
Update $v$ to the new color by setting $X_{t+1}(v) = c$, and keep the coloring
the same on the rest of the graph $X_{t+1}(w) = X_t(w)$ for all $w\neq v$.

The dynamics is ergodic whenever $k\geq\Delta+2$ where $\Delta$ is the maximum 
degree of the input graph $G$, and hence since it is symmetric 
its unique stationary distribution $\pi$ is uniform over $\Omega$ \cite{Jerrum}.  
We measure the convergence time to the stationary distribution by the 
{\em mixing time},  the minimum number of steps $T$, from the worst 
initial state $X_0$, to ensure that the distribution $X_T$ is within variation
distance $\leq 1/4$ of the uniform distribution $\pi$.
Our aim is to show that the 
mixing time is polynomial in $n$, the size of the underlying graph, in which case
we say that the dynamics is {\em rapidly mixing}.  When the mixing time is 
exponential in $n^{\Omega(1)}$ then we say the dynamics is {\em torpidly mixing}.

The study of Gibbs sampling has yielded many beautiful results, we survey the relevant results
for the colorings problem here.  
The natural conjecture is that whenever $k\geq\Delta+2$ then the Glauber dynamics is
rapidly mixing.
The minimal evidence in favor of the conjecture is that {\em uniqueness},
which is a weak form of decay of correlations, holds on infinite $\Delta$-regular trees \cite{Jonasson}.
On the hardness side, \cite{GSV} showed that the dynamics is torpid mixing 
on random bipartite, $\Delta$-regular graphs for even $k$ when $k<\Delta$; more generally, in
this regime the approximate counting problem is NP-hard (unless NP=RP) on triangle-free graphs of maximum
degree $\Delta$.
On the positive side, the best known result for general graphs is $O(n\log{n})$ mixing time
for $k>2\Delta$ \cite{Jerrum} and $O(n^2)$ for $k>\frac{11}{6}\Delta$~\cite{Vigoda}.

Further improvements were made with various assumptions about the graph such as girth or maximum degree.
Dyer and Frieze \cite{DFUniformity} utilized properties of the stationary distribution, later termed
{\em local uniformity properties}, to prove rapid mixing on graphs
with maximum degree $\Delta=\Omega(\log{n})$ and girth $g=\Omega(\log{\Delta})$ when
$k>(1+\eps)\alpha\Delta$  where $\alpha\approx 1.763...$ 
is the root of  $\alpha = \exp(1/\alpha)$.
The girth and maximum degree assumptions were further improved by 
Dyer et al. \cite{DFHVConstantDegree} to 
girth $g\geq 5$ and $\Delta>\Delta_0$ where $\Delta_0=\Delta_0(\eps)$ is a sufficiently
large constant.
Further improvements on the constant~$\alpha$ were made in \cite{Molloy,LM,DFHVConstantDegree,HV-nonmarkovian} 
with stronger girth and maximum degree assumptions;  however, as we'll outline later these improvements required more sophisticated local uniformity properties which
necessitated the stronger conditions and more complicated arguments.
This same threshold $\alpha\Delta$ appeared in the work of 
Goldberg, Martin and Paterson \cite{GMP} who proved a strong form of decay of correlations
on triangle-free graphs when $k>\alpha\Delta$, which implied rapid mixing for amenable graphs.
We utilize similar local uniformity properties to 
\cite{GMP,DFUniformity,TomsUniformity,DFHVConstantDegree,HV-nonmarkovian}
and naturally the  constant $\alpha$ arises in our work.

An intriguing case to study in this context are sparse random graphs,  namely
Erd\"{o}s-R\'{e}nyi random graphs $G(n,d/n)$ for  
constant $d>1$.  Sampling from Gibbs distributions induced by instances of $G(n,d/n)$, or,  
more generally, instances of so-called random constraint satisfaction problems, 
is at the heart 
of  recent endeavors to  investigate connections between 
phase transition phenomena and the efficiency of algorithms 
\cite{ACO,CoEfth,PNAS,GamSud,Sly}.

Whereas the rapid mixing results for general graphs bound $k$ in terms of the maximum degree
$\Delta$, on the other hand for sparse random graphs $G(n,d/n)$ it is 
natural to bound $k$ by the {\em expected degree} $d$. 
This is a substantial difference since typical instances of $G(n,d/n)$ have maximum degree as 
large as $\Theta(\log{n}/\log\log{n})$, while the 
expected degree  $d$ is constant (i.e., independent of $n$).
To this end, for deriving our results, it is necessary to argue about the 
statistical properties of the underlying graph.

The performance of the Glauber dynamics has been studied in 
statistical physics using sophisticated tools, but mathematically non-rigorous. 
In particular, in \cite{PNAS} it is conjectured that 
rapid mixing holds in the uniqueness region and hence it should hold for $k\geq d+2$.
Moreover, it is conceivable that there is a weak form of a sampler down to the
clustering threshold at $k\approx d/\log{d}$~\cite{ACO}.

The first results in this context were  by Dyer et al. \cite{DFFV} who
proved rapid mixing of an associated block dynamics
when $k=\Omega(\log{\log{{n}}}/\log{\log{\log{n}}})$.
A significant improvement was made by Mossel and Sly \cite{MS2}
who established rapid mixing for a constant number of colors $k$ 
(though $k$ was polynomially related to $d$).
This was further improved in \cite{Efthymiou}  to reach $k$ which is linear in $d$, namely
 $k>\frac{11}{2}d$.
Recently, a non-Markov chain $\fpaus$ was presented for colorings that requires $k>3d+O(1)$ \cite{YZ};
however this did not imply any guarantees on the behavior of the Glauber dynamics.
We note that a significantly weaker form of a sampler was presented for 
the case $k \geq (1+\eps)d$ for all $\eps>0$ \cite{Efthymiou-ESA}; this 
only obtains a weak approximation depending on $n$, whereas an $\fpaus$ allows arbitrary close
approximation.

We further improve rapid mixing results for sparse random graphs.  What is especially notable 
in our results is that the threshold on $k/d$ is now comparable to those on
general graphs for $k/\Delta$.
Our main result is rapid mixing of the Glauber dynamics 
on sparse random graphs when $k>\alpha d$.

\begin{theorem}\label{thm:1.76-main}
Let $\alpha\approx 1.763...$ denote the root of 
$\alpha = \exp(1/\alpha)$.
For all $\eps>0$, there exists $d_0$,
for all $d>d_0$, for $k\geq (\alpha+\eps)d$, with probability $1-o(1)$
over the choice of $G\sim G(n,d/n)$, the mixing time of the 
Glauber dynamics is $O(n^{2+1/(\log d)})$.  
\end{theorem}

From an algorithmic perspective, we have to consider how to  get the initial configuration of the dynamics.
We use   the well-known polynomial time algorithm by Grimmett and  McDiarmid  \cite{grimmett},
which $k$-colors typical instances of $G(n,d/n)$ for any $k>d/\log d$. Note that $\alpha d\gg d/\log d$.

Previous results for the Glauber dynamics on 
sparse random graphs \cite{MS2,Efthymiou} implied polynomial
mixing time but the exponent was an increasing function of $d$; similarly for the 
running time of the sampler presented in \cite{YZ}.  Here we get a fixed
polynomial.  This results from an improved comparison argument which utilizes
a more detailed analysis of the star graph.  

The previous results \cite{DFFV,MS2,Efthymiou} for sparse random graphs (as does our work) 
use arguments about the statistical properties of  the underlying graph, for example, the distribution
of high-degree vertices.  To achieve a bound below $2d$ we also need to
argue about the {\em statistical properties of random colorings} as well;  
that is, what does a typical coloring of $G(n,d/n)$ look like.  
This poses new challenges in the analysis 
of the Glauber dynamics as it requires a meticulous study of its  behavior when it starts from 
a pathological coloring, see further details in Section \ref{sec:BurnInAnalysis}.

 The first step in our analysis is defining an appropriate block dynamics;
the use of the block dynamics was also done in previous results on random graphs 
\cite{DFFV,MS2,Efthymiou}.
The block dynamics partitions the vertex set $V$ into disjoint blocks $V=B_1\cup B_2\cup \dots \cup B_N$.  
In each step we choose a random block and recolor that block (uniformly at random
conditional on the fixed coloring outside the chosen block).
After proving rapid mixing of the block dynamics, rapid mixing of the Glauber dynamics
will follow by a standard comparison argument, see Section \ref{sec:thm:1.76-Glauber}.

The key insight is to use the blocks to ``hide'' high degree vertices deep inside the blocks.
By high degree we mean a vertex of degree $>(1+\delta)d$ for a small constant $\delta$, 
and the remaining vertices are 
classified as low degree.
The blocks are designed so that from a high degree vertex there
is a large buffer of  low degree vertices to the boundary of the block.
In addition, each block is a tree (or unicyclic), and hence it
is straightforward to efficiently generate a random coloring of the chosen block.
Our  block construction builds upon ideas from \cite{Efthymiou} 
which assigns appropriate weights on the paths of $G(n,d/n)$  
to distinguish which vertices can be used at the boundary of the blocks.
For more details regarding the block construction see Section \ref{sec:IntroBlocks}.

Our first progress is to achieve rapid mixing when $k>2d$.  
Even if the maximum degree was $\Delta$ 
it was unclear how to extend Jerrum's \cite{Jerrum} classic $k>2\Delta$ approach to 
directly analyze the block dynamics, as opposed to the Glauber dynamics.  
That is our first contribution: we present a simple weighting scheme
so that path coupling applies to establish rapid mixing when $k>2\Delta$ for 
the block dynamics with ``simple'' blocks, see Section \ref{sec:2d} for more details.
From there it is straightforward to 
extend to random graphs with expected degree $d$ when $k>2d$
(though technically it requires considerable work to deal with the high degree vertices).

To improve the result from $2d$ to $1.763...d$ we utilize the so-called {\em local uniformity} properties,
in particular the lower bound on available colors as in \cite{GMP,DFUniformity,TomsUniformity,DFHVConstantDegree}.  
The idea is that whereas a worst case coloring has $\Delta$ colors in the neighborhood of
a particular $v$ (we're considering the case of a graph with maximum degree $\Delta$
for simplicity) and hence $k-\Delta$ ``available'' colors, after a short burn-in period 
in the coloring $(X_t)$ we are likely to have $k(1-1/k)^\Delta \approx k\exp(-\Delta/k)$ 
available colors for $v$.  Our approach for establishing local uniformity is  similar in spirit to
that in \cite{DFUniformity}.

Our challenge is that while we are burning-in to obtain this local uniformity property,
we need that the initial disagreement does not spread too far.  For this we need a
concentration bound on the spread of disagreements within a block.
To do that we utilize disagreement percolation, which is now a standard tool in
the analysis of Markov chains and statistical physics models.   This is one of the 
key technical contribution of our work, see  Sections \ref{sec:BurnInAnalysis} and 
\ref{sec:DisPerc}, for further discussion.

Concluding, we  remark that our techniques  find application  to  other  models on $G(n,d/n)$. For example  
in Section \ref{sec:HCRapid4Block}, 
we prove a rapid mixing result for the so-called hard-core model with {\em fugacity} $\lambda$.  
Our result  improves  the   previous best bound, in terms of $\lambda$, 
in \cite{Efthymiou} by a factor $2$. \\ 

\noindent
{\bf Outline of paper}
In Section \ref{sec:IntroBlocks} we introduce the blocks dynamics  for which we show rapid mixing.
Then, our main theorem (Theorem \ref{thm:1.76-main}) for the Glauber dynamics
follows from rapid mixing of the block dynamics via a comparison argument.
In Section \ref{sec:2d} we give an overview of how we obtain rapid mixing for $k>2d$ for
the block dynamics by introducing a new metric for the space of configurations.
In  Section \ref{sec:from2dtoalphad} we discuss the improved $k>1.763... d$ bound,
focusing on utilizing the local uniformity properties and 
the analysis of the  burn-in phase. \\

 
 \noindent
%
{\bf Notation}
We will define a block dynamics with a disjoint set of blocks $\B=\{B_1\cup\dots\cup B_N\}$. 
For a block $B\in\B$, denote its outer and inner boundaries as
\begin{eqnarray*}
   \outBound B &:=& \{ y\in V: y\notin B, \mbox{ there exists } z\in B \mbox{ where } (y,z)\in E\},
   \\
   \inBound B  &:=& \{ z\in V: z\in B, \mbox{ there exists } y\notin B \mbox{ where } (y,z)\in E\}.
\end{eqnarray*}
For the collection $\B$ we will look at the union of the outer boundaries, or
equivalently the union of the inner boundaries, namely:
\[
   \ImpVrtx := \textstyle \bigcup_{B\in \B}  \outBound B = \bigcup_{B\in \B}  \inBound B.
 \]
The degree of vertex $v$ is denoted as $\degree(v)$, and its
set of neighbors is denoted by $N(v)$.
Similarly, for a block $B\in\B$, the neighboring blocks are denoted as $N(B)$.

\section{Rapid mixing for Block dynamics}\label{sec:IntroBlocks}

As mentioned earlier, to prove Theorem \ref{thm:1.76-main} we will prove
rapid mixing of a corresponding block dynamics on $G(n,d/n)$ and then  we employ a  standard 
comparison argument \cite{Martinelli}.
That is, we bound the relaxation time for the Glauber dynamics in 
terms of  the relaxation time of the 
block dynamics and the relaxation time of the Glauber dynamics within a single block.
Since the blocks are trees (or unicyclic) our approach 
requires studying the mixing rate  of the Glauber dynamics on highly non-regular  trees 
and we do so in a manner similar to  \cite{LucMol, TVVY}. We provide some, we believe non-trivial,
bounds on the relaxation times of a star-structured block dynamics.
We refer the interested reader to Section \ref{sec:thm:1.76-Glauber} of the appendix
for the  comparison argument.

First we describe how we create the blocks for the dynamics.  
For this we need use a weighting schema similar to \cite{Efthymiou}.
Assume that we are a given a graph $G=(V,E)$ of {\em maximum degree} $\Delta$.
We specify  weights for  the vertices of $G$.  There are two  {\em parameters},
$\epsilon>0$ and $d>0$.  We let   $\TDeg=(1+\epsilon/6)d$ denote the threshold for ``low/high'' 
degree vertices.  For each vertex  $u\in V$  we define its weight $W(u)$ as follows:
\begin{equation}\label{eq:DefVrtxWeight}
W(u)= 
\begin{cases}
\left ( {1+\epsilon/10}\right)^{-1} & \textrm{if } \degree(u)\leq \TDeg
\\
{ d^{15}\ \degree}(u)  & \mbox{otherwise.}
\end{cases}
\end{equation}

The weighting assigns low-degree vertices, namely those with degree $\leq \TDeg$, a weight~$<1$, whereas high-degree vertices have weight~$\gg1$ which is proportional to their degree.
Given the vertex weights in \eqref{eq:DefVrtxWeight}  for each path $\P$ in $G$ we specify weights, too. 
More specifically,  for each path $\P=u_1,\ldots, u_\ell$ in $G$ define its weight $W(\P)$ 
as the product of the vertex weights:
\begin{equation}
W(\P)=\prod^\ell_{i=1}W(u_i). \label{def:PathWeight}
\end{equation}

\noindent
We use the above weighting schema to specify the blocks for our  dynamics. 
Of particular interest  are the vertices $v$ for which {\em all of the paths} 
that emanate from $v$ are of low weight. 
Given some integer $r\geq 0$,  a vertex $v$ is called a ``$r$-{\em breakpoint}'' if the following holds:
\[  \mbox{For every path $\P$ of length at most $r$ that
starts at $v$ it holds that $W(\P)\leq 1$.}
\]

\noindent
The breakpoints are particularly important for our block construction as we use them to specify the boundary 
of the blocks.  Intuitively, choosing large $r$,  for a $r$-breakpoint we have that high degree vertices are far from it.

We say that the graph $G$, of maximum degree at most $\Delta$, admits
a ``{\em sparse block partition}" ${\cal B}={\cal B}(\epsilon,  d, \Delta)$, for some $\epsilon, d>0$, if 
 $\cal B$ has the following properties:
Each block $B\in {\cal B}$  is a tree with at most one extra edge. Each vertex $u$
 which is at the outer boundary of multivertex block $B$, can only have one neighbour
 inside $B$. More importantly, $u$ is at a sufficiently large distance from the high degree vertices in 
 $B$ as well as the cycle in $B$ (if any).
The high degree requirement translates to  $u$ being an $r$-breakpoint for  large $r$.
Finally, $u$ does not belong  to any cycle of length less than $d^2$. 
To be more specific we have the following:

\begin{definition}[Sparse block partition]\label{def:SparseBlockPart}
For $\epsilon>0$, $d>0$ and $\Delta>0$,  consider a graph  $G=(V,E)$  of maximum degree at most $\Delta$.
We say that $G$ admits a ``sparse block partition" ${\cal B}={\cal B}(\epsilon,  d, \Delta)$ 
 if  $V$ can be partitioned  into the set of blocks  $ {\cal B}$ for which  the following is true:
\begin{enumerate}
\item Every $B\in {\cal B}$ is a tree with at most one extra edge.
\item Each vertex $v$ in the outer boundary of a multi-vertex block $B$ has the following  properties:
\begin{enumerate}
\item  $v$ is an $r$-breakpoint for $r> { \max\{ {\tt diam}(B),  \log\log n \}}$,
\item  $v$ has exactly one neighbor inside $B$,
\item  if $B$ contains a cycle $C$, then 
$\textstyle\distance(v,C)\geq \max \left \{2  \log(|C|\  \Delta),\   \frac{\log\log d}{\log d}\left( {|C|+\log \Delta} \right) \right\}$
\end{enumerate}
\item  Each vertex $u\in \outBound B$, for any $B\in {\cal B}$, 
does not belong to any cycle of length $<d^{2}$.
\end{enumerate}
\end{definition}

To give an idea how such a partition looks like, we consider  the case of  $G(n,d/n)$. 
There, the sparse block partition   ``hides"  the large degree vertices, i.e., $>\TDeg$,  deep inside the blocks, 
and similarly the cycles of length $<d^{-2/5}\log n$.
For the high degree requirement we use $r$-breakpoints at the boundary of multivertex  blocks.
Usually $r\leq \log n/\log^4 d$ and  typically $G(n,d/n)$ has a plethora of $r$-breakpoints.
We also we the fact that, typically,  the  short cycles in $G(n,d/n)$ are far apart from each other.
The plethora of $r$-breakpoint in $G(n,d/n)$ allow to surround the short cycles from the appropriate distance.

Our rapid mixing result for block dynamics is about graphs which admit a sparse block parition 
${\cal B}={\cal B}(\epsilon, d, \Delta)$, for appropriate   $\epsilon, d, \Delta$. We consider block dynamics with set of blocks specified by ${\cal B}$. 
The lower bound on $k$  for rapid mixing will  depend on $d$ rather than the maximum 
degree $\Delta$. 
In that respect   the   interesting case is when  $\Delta\gg d$, like the typical instances of 
$G(n,d/n)$.

So as to show rapid mixing for the graphs which admit vertex partition ${\cal B}(\epsilon, d,\Delta)$, we have to  guarantee that the corresponding block dynamics is  {\em ergodic}.

\begin{definition}
For $\epsilon, d, \Delta>0$, let $\IntrstGraphFam=\IntrstGraphFam(\epsilon, d, \Delta)$ be the family of graphs
on $n$ vertices such that for every $G\in \IntrstGraphFam$ the following holds:
\begin{enumerate}
\item $G$ admits a sparse block partition ${\cal B}(\epsilon, d, \Delta)$
\item The corresponding block dynamics is {\em ergodic} for $k\geq \alpha d$
\end{enumerate}
where the quantity $\alpha$ we use above is the solution of the equation $\alpha^{\alpha}=e$, 
i.e.,  $\alpha=1.7632\ldots$
\end{definition}

\begin{theorem}\label{thrm:RapidMixingAvrgDegGraph176}
For all $\epsilon>0$, there exists $C>0$  such that for all sufficiently large  $d>0$
 and any  graph  $G\in \IntrstGraphFam (\epsilon, d, \Delta)$, where $\Delta>0$ can be a function of $n$,
the following is true: For $k=(\alpha+\epsilon)d$, the 
  block dynamics with set of block ${\cal B}$ has mixing time 
\[
 \Tmix\leq C n \log n,
 \]
 where $\alpha$ is the solution of the equation $\alpha^{\alpha}=e$, i.e., $\alpha=1.7632\ldots$.
 Moreover, each step of the dynamics can be implemented in $O(k^3 B_{\rm max} )$ time, where 
 $B_{\rm max}$ is the size of the largest block. 
\end{theorem}

\noindent
The proof of Theorem \ref{thrm:RapidMixingAvrgDegGraph176} appears in Section \ref{sec:RapidMixingColor} of  the appendix.

In light of   Theorem \ref{thrm:RapidMixingAvrgDegGraph176}  we get  rapid mixing  for the block dynamics for $G(n,d/n)$
by considering the following, technical,  result.

\begin{lemma}\label{thrm:RapidMixBlockGnp}
For all $\epsilon>0$ and $\Delta=(3/2) \left( \log n/\log \log n\right)$ and sufficiently large $d>0$
it holds that $\Pr[G(n,d/n)\in \IntrstGraphFam (\epsilon, d, \Delta)]\geq 1-o(1)$.
Moreover, $G(n,d/n)\in \IntrstGraphFam (\epsilon, d, \Delta)$ implies that  $B_{\rm max}\leq n^{1/(\log d)^2}$.
\end{lemma}

\noindent
The proof of Lemma \ref{thrm:RapidMixBlockGnp} appears in Section \ref{sec:thrm:GnpAdmitsPart} of the appendix.

In light of Theorem \ref{thrm:RapidMixingAvrgDegGraph176} and Lemma \ref{thrm:RapidMixBlockGnp},  Theorem \ref{thm:1.76-main}  follows by a 
comparison argument we present in   Section \ref{sec:thm:1.76-Glauber} in the appendix.

\section{Analysis of Block Dynamics for $k>2d$ - Overview}
\label{sec:2d}

The techniques we present in this section are sufficient to show rapid mixing of the corresponding
block dynamics for $k>2d$.  
Later  we utilize {\em local uniformity properties} to get a better bound on $k$.

\subsection{A new metric - Proof overview for $k>2\Delta$}\label{sec:MaxDegRapid}

We will use path coupling and hence we consider two copies of the block dynamics
$(X_t),(Y_t)$ that differ at a single vertex $u^*$.
Let us first consider the analysis for a graph with maximum degree $\Delta$.
Jerrum's analysis of the single-site Glauber dynamics \cite{Jerrum} 
(and Bubley-Dyer's simplification using path coupling \cite{PathCouplingBD})
are well-known for the case $k>2\Delta$.  They show a coupling so that the
expected Hamming distance decreases in expectation.

Our first task is generalizing this analysis of the Glauber dynamics to the block dynamics.  
The difficulty is that when we update a block $B$ that neighbors the disagree vertex $u^*$
 the number of disagreements may grow by the size of $B$.  However
disagreements that are fully contained within a block do not spread.
Consequently, we can replace Hamming distance by a simple metric, and then we
can prove rapid mixing for $k>2\Delta$ for any block dynamics where the blocks are
all trees. 

In particular,
if some vertex $z$ is internal, i.e., it does not have any neighbors outside its block
it gets { weight} 1. If $z$ is not internal, it    is assigned a { weight} which is $n^2$ times its out-degree from 
its block,  i.e., $\degree_{out}(z) = |N(z)\setminus B|$  where $B$ is the block containing $z$.  
Then for a pair  $X_t,Y_t$
their distance is the sum of the weight of the vertices in their symmetric difference, i.e.,
\begin{equation}\label{eq:DefOfHammingWeights}
\EdgeBlockWeight(X_t,Y_t) =\sum_{z\in V\setminus \ImpVrtx} \indicator{z\in X_t\oplus Y_t}+n^{2} \sum_{z\in  \ImpVrtx}  \degree_{out}(z)\indicator{z\in X_t\oplus Y_t}
\end{equation}

\noindent
To get some intuition, note that the vertices which are internal in the blocks have 
``tiny" weight compared to the rest ones. This essentially captures  that 
the disagreements that matter in the path coupling analysis are those which involve vertices 
at the boundary of blocks, while the ``potential" for such a vertex to spread disagreements to neighboring blocks
 depends on its out-degree.

Using the above metric 
we will derive the following rapid mixing result. For expository reasons we, also,  provide the proof here.

\begin{theorem} \label{thm:block-regular}
There exists $C>0$, for all $g\geq 3$, all $G=(V,E)$ with 
girth $\geq g$, maximum degree $\Delta$ and $k > 2\Delta$,
for any partition of the vertices $V$ into disjoint blocks
 $V=B_1\cup B_2 \cup \dots \cup B_N$ where $\diameter(B_i)\leq g/2-3$
 for all $i$, the mixing time of the block dynamics satisfies:
 \[
\textstyle  \Tmix\leq C\Delta n\log{n}.
 \]
\end{theorem}

\begin{proof}
Let $S\subset\Omega\times\Omega$ denote a pair of colorings that differ at a single vertex.
Moreover, partition $S=\cup_{v\in V} S_v$ where $S_v$ contains those
pairs $(X_t,Y_t)$ which differ at $v$. We will define a coupling for all pairs in $S$ where the expected 
distance decreases and then apply path coupling \cite{PathCouplingBD} to derive a coupling for an arbitrary 
pair of states where the distance contracts.

Consider a pair of colorings $(X_t,Y_t)\in S_{u^*}$ the differ at an arbitrary vertex $u^*$. 
In our coupling both chains update the  same block at
each step.  Let $B_t$ denote the block updated for this step $(X_t,Y_t)\rightarrow (X_{t+1},Y_{t+1})$.  Also, let $B^*$ denote the block containing $u^*$.  

We consider two cases for the vertex $u^*$, either: 
(i) $u^*$ is an internal vertex to its block $B^*$, 
i.e., $\degree_{out}(u^*)=0$, or 
(ii) $u^*$ is  on the boundary of its block, i.e., $u^*\in \inBound B^*$.

The easy case is case (i) when $u^*$ is internal.  There are no blocks with disagreements on their boundary, and hence new disagreements cannot form.  
Since the neighborhood of the updated block $B_t$ is the same in
both chains, we can use the identity coupling 
so that $X_{t+1}(B_t) = Y_{t+1}(B_t)$.   The distance cannot increase, and if $B_t=B^*$ then
we have $X_{t+1}=Y_{t+1}$; this occurs with probability $1/N$ where $N$ is the number of blocks.
Therefore, in the case that $u^*\notin\inBound B^*$ we have:
\begin{equation}
  \Exp{ 
\EdgeBlockWeight(X_{t+1}, Y_{t+1})
 \ | \ X_t, Y_t }   \leq  \left( 1 -1/N \right) \EdgeBlockWeight (X_{t}, Y_{t}) . \label{eq:BlockUpdt4Internal}
\end{equation}

Now consider case (ii) where $u^*\in\inBound B^*$. If $u^*\notin\outBound B_t$ then
we can couple $X_{t+1}(B_t)=Y_{t+1}(B_t)$ and hence the distance does not increase.
Moreover if $B_t=B^*$ then we have $X_{t+1}=Y_{t+1}$; thus with probability $1/N$ the
distance decreases by $ -n^2\degree_{out}(u^*)$.
The distance can only increase when $u^*\in\outBound B_t$ and hence our
main task is to bound the expected change in the distance in this scenario.
We will prove the following:
\begin{equation}\label{eq:BlockUpdtCovergentSimple}
\Exp{ \EdgeBlockWeight(X_{t+1}, Y_{t+1})
- \EdgeBlockWeight (X_{t}, Y_{t})
 \ | \ X_t, Y_t, \ B_t, \  u^*\in\outBound B_t}
 \leq n^2\left( 1-{1}/(2\Delta) \right).
\end{equation}

\noindent
All the above imply that having $u^*\in \outBound B^*$ we get that

\begin{eqnarray}
 \Exp{ 
\EdgeBlockWeight(X_{t+1}, Y_{t+1})
 \ | \ X_t, Y_t } &\leq&\EdgeBlockWeight (X_{t}, Y_{t})- \frac{n^2}{N} \degree_{out}(u^*) + \frac{n^2}{N}\sum_{B:\  u^*\in\outBound B} \left(1-{1}/{(2\Delta)}\right) 
\nonumber \\ 
& \leq &  \left ( 1- {1}/{(2N\Delta)} \right) \EdgeBlockWeight (X_{t}, Y_{t}), \label{eq:BlockUpdt4External} 
\end{eqnarray}
where in the first inequality we use the fact that each block is updated with probability $1/N$. The second inequality follows from the observation that $\EdgeBlockWeight (X_{t}, Y_{t})=n^2\degree_{out}(u^*)$, while the number
of sumads  in the first inequality is equal to $\degree_{out}(u^*)$.

In light of \eqref{eq:BlockUpdt4Internal} and \eqref{eq:BlockUpdt4External}, path coupling implies the  following: For two copies of the Glauber dynamics
$(X_t)_{t\geq 0}$, $(Y_t)_{t\geq 0}$ there is a coupling such that for any $T>0$ and any $X_0, Y_0$ we have
\[
\ExpCond{\EdgeBlockWeight(X_T, Y_T)}{X_0, Y_0} \leq \left(1-{1}/({2N\Delta)} \right)^T\EdgeBlockWeight(X_0, Y_0).
\]
Since $\EdgeBlockWeight(X_0, Y_0)\leq 2\Delta n^3$, we have:
$$
\Prob{X_T\neq Y_T} \leq 2\Delta n^3 \exp(-T/(2N\Delta) )\leq \eps,
$$
for $T=20\Delta n\log{n}$, which proves the theorem.

We now prove \eqref{eq:BlockUpdtCovergentSimple}.
The disagreements on the inner boundary of a block are the dominant term in $\distance()$,
hence for a pair of colorings $\sigma, \tau$, let
\begin{displaymath}
\textstyle {\cal R}( \sigma, \tau)=n^2\sum_{z\in \sigma\oplus \tau } \degree_{out}(z).
\end{displaymath}

By simply ``giving away'' all of the vertices in $B_t$ as internal disagreements after the update we
can upper bound the l.h.s. of \eqref{eq:BlockUpdtCovergentSimple} in terms of ${\cal R}()$:
\begin{multline}
 \Exp{ \EdgeBlockWeight(X_{t+1}, Y_{t+1})
- \EdgeBlockWeight (X_{t}, Y_{t})
 \ | \ X_t, Y_t, \ B_t, \  u^*\in\outBound B_t}
 \\
 \nonumber
 \leq 
|B_t|+ \Exp{ {\cal R}( X_{t+1}, Y_{t+1}-{\cal R}( X_{t}, Y_{t})
 \ | \ X_t, Y_t, \ B_t, \  u^*\in\outBound B_t}.
\end{multline}

\noindent
Since $|B_t|\leq n$, \eqref{eq:BlockUpdtCovergentSimple} follows by showing that 
\begin{equation}\label{eq:Target4eq:BlockUpdtCovergentSimple}
\Exp{ {\cal R}( X_{t+1}, Y_{t+1}-{\cal R}( X_{t}, Y_{t})
 \ | \ X_t, Y_t, \ B_t, \  u^*\in\outBound B_t} \leq n^2\left(1-{1}/{(\Delta+1)}\right ).
\end{equation}

\noindent
For  $v\in V$ and $T\subseteq V$, where the induced subgraph on $T$ is a tree and
$\diameter(T)\leq g/2-3$, let
\begin{equation}\label{eq:MainTaskFor2Delta}
 Q_v(T) = \max_{(X_t,Y_t)\in S_v} 
\ExpCond{ {\cal R}(X_{t+1}, Y_{t+1}) - {\cal R}( X_t, Y_t)}{X_t,Y_t \mbox{ and recolor block $T$}}.
\end{equation}
The reader may identify the expectation in  \eqref{eq:Target4eq:BlockUpdtCovergentSimple} as $Q_{u^*}(B_t)$.  
Even though our concern is the blocks of the dynamics,   $Q_v(T)$ is defined for arbitrary $T$. 
Note that if $v\in \outBound T$ and $|N(v)\cap T|\geq 2$ then the diameter assumption for $T$ 
would imply that a cycle of length $<g$ is present in $G$.
Clearly this is not true since $G$ is assumed  to have girth $g$. Therefore, we conclude that if $v\in \outBound T$,  
then it has is exactly one neighbor  in $T$.

We'll prove by induction on $|T|$  that   $Q_v(T)\leq n^2\left(1-{1}/({\Delta+1})\right)$.
When,  $v\notin \outBound T=\emptyset$ we have $Q_v(T)=0$, since there are no disagreements 
on $\outBound T$ and hence we can trivially use the identical coupling for the vertices in $T$.
We proceed with the  case where $v\in \outBound T$.

 Assume that $z\in T$ is adjacent to $v$. Furthermore, assume that the tree is rooted at $z$ and for every 
vertex $y$ let $T_y$ be the subtree which contains $y$ and all its descendants.

The identical coupling is precluded because of the disagreement at $\outBound T$.
The coupling decides the colorings  of a single vertex at a time. It starts  with  $z$
and couples $X_{t+1}(z)$ and $Y_{t+1}(z)$ maximally, subject to the boundary conditions
of $T$.  Then, in a BFS manner it considers the rest of the vertices, starting with the children of $z$. 
For each  $w$ the coupling $X_{t+1}(w)$ and $Y_{t+1}(w)$ is  maximal, subject to the boundary 
conditions of $T$ but also the configuration of the parent of $w$.

Consider  $w\in T$ and let $u$ be its parent (with  $v$ being the parent of $z$).
Given these $w, u$ it is useful to make a few observations: Consider  the coupling
of $X_{t+1}(w)$ and $Y_{t+1}(w)$  given that  $X_{t+1}(u)=Y_{t+1}(u)$. Then, it is direct  that there is 
no disagreement on the boundary of the
subtree $T_w$ and hence we can use the  identical coupling for  $X_{t+1}(w)$ and $Y_{t+1}(w)$,
and in fact,  we can have identical coupling for all of the vertices in $T_w$.
In the other case of disagreement at $u$, note that
\begin{equation}\label{eq:introProbDis}
\textstyle \Pr[X_{t+1}(w)\neq Y_{t+1}(w)\ |\ X_{t+1}(u) \neq Y_{t+1}(u)]  \leq  {1}/({k-\Delta}).
\end{equation} 
since the only disagreement at the boundary of $T_w$ is at $u$ and the probability of disagreement at $w$ is upper bounded by 
the probability of the most likely color for $X_{t+1}(w)$ and $Y_{t+1}(w)$ which is $1/(k-\Delta)$. Since there are
at least $k-\Delta$ available colors for $w$.

Now we proceed with the induction.  The base case is  $T=\{z\}$, then, using \eqref{eq:introProbDis} we have
 $$
 Q_v(T)\ \leq \    n^2 \Delta \Pr[X_{t+1}(z)\neq Y_{t+1}(z) ] \ \leq\  \frac{n^2 \Delta}{k-\Delta}\leq  n^2\left( 1-\frac{1}{\Delta+1}\right),
\ \ \mbox{  for $k>2\Delta$},
$$
where the first inequality follows because the contribution of $z$ to the distance  is  $\leq n^2\Delta$.
This proves the base of induction. 
To continue,  we note that  the following inductive relation holds
\[
Q_v(T) \leq  \textstyle \Pr[X_{t+1}(z)\neq Y_{t+1}(z) ]\left( n^2 \ \degree_{out}(z) + \sum_{y\in N(z)\cap T} Q_z(T_y) \right).
\]
The above follows by noting  $Q_v(T)$ is equal to the expected contribution from $z\in N(u^*)\cap T$ plus 
the  expected contribution  from  each subtree $T_y$.
We multiply the contribution of all $T_y$ with the probability of the event $X_{t+1}(z)\neq Y_{t+1}(z)$
because,  each subtree starts contributing once we have $X_{t+1}(z)\neq Y_{t+1}(z)$. 

The induction hypothesis implies that for any $y$ we have $Q_z(T_y)<n^2$. 
We get that 
\begin{eqnarray*}
Q_v(T)  
&\leq & \textstyle \Pr[X_{t+1}(z)\neq Y_{t+1}(z) ] \left(  n^2\ \degree_{out}(z) + n^2(\Delta- \degree_{out}(z))\right)\\
&\leq &  \frac{n^2\Delta }{k-\Delta}   \hspace{5.7cm} \mbox{ [by  \eqref{eq:introProbDis}]}.
\\ &\leq &  n^2 \left ( 1-{1}/({\Delta+1})\right)  \hspace{3.7cm} \mbox{ [since $k\geq 2\Delta+1$]}.
\end{eqnarray*}

\noindent
The above  bound implies that \eqref{eq:Target4eq:BlockUpdtCovergentSimple} holds, 
since we can  identify the expectation in  \eqref{eq:Target4eq:BlockUpdtCovergentSimple} as $Q_{u^*}(B_t)$.

The theorem follows.
\end{proof}

\subsection{Proof overview for random graphs $G(n,d/n)$ and $k\geq (2+\epsilon)d$}
We extend the above approach to random graphs when 
$k>(2+\epsilon)d$ where $d$ is the expected degree instead of the maximum degree $\Delta$.  
Morally, this  amounts to having blocks whose   behavior, in terms of generating new disagreements, 
is not too different than that of a tree of maximum degree $\TDeg:=(1+\epsilon/6)d$.
Our goal is to prove a result similar to  \eqref{eq:BlockUpdtCovergentSimple},
i.e.,  the expected increase from updating a block which is next to a single disagreement is less than $n^2$.
If we have that,  then  the proof of rapid mixing follows   the same line of  arguments as  that we have in Theorem \ref{thm:block-regular}.

We use blocks from sparse block partition (Definition \ref{def:SparseBlockPart})
The blocks here  are tree-like with at most one extra edge.  There is a buffer of low degree vertices along the inner boundary 
of a block.  (Recall low degree means degree $\leq \TDeg$.)
Note that even though high degree vertices have tiny weight under our 
distance $\dist()$, they can still
have dramatic consequences since their degree may be a function of $n$ while 
$k$ and $d$ are constants, and when a disagreement reaches a high degree vertex it 
then has the potential to propagate along a huge number of paths to the boundary of the block.

The blocks are designed so that  high degree vertices and any possible cycle
 are ``deep'' inside their respective blocks: specifically, for a vertex $v$  
of degree $L>\TDeg$, every path from $v$ 
to the boundary of its block consists of $\Omega(\log{L})$ low degree vertices
(in an appropriate amortized sense).   Using these low degree vertices 
the probability of propagating a disagreement along this
path of low-degree vertices offsets the potentially huge effect of a high degree vertex disagreeing.
Similarly, we work for the cycle inside the block.

More concretely, we get a handle on the expected increase of distance when we update the block $B$ which 
has a disagreement at   $u^*\in \outBound B$ by arguing about the  {\em probability of propagation}  inside the block.
For a vertex $v\in B$ we let { probability of propagation} be the probability  of having a path of disagreeing vertices from $u^*$ to $v$,
 given that  all the vertices in path  but  $v$ are  disagreeing.  We get the desirable bound on 
 the expected increase by showing that for every  {\em low degree}  $v\in B$ which is within small 
 distance from $u^*$  (i.e., $\log^2 d$)  the  probability of propagation  is less than 
 $\frac{1}{(1+\epsilon/2)\degree(v)}$.

For $k\geq (2+\epsilon)d$ the above bound for the probability of propagation
is always true, i.e. for every boundary condition of the block $B$. 
The  details of the argument  appear in  Section \ref{sec:SimpleCriterion}  in the appendix.
However,   to $k>1.76...d$, the extra challenge   is that the vertices 
do not necessarily  have a small probability of propagation.   
This is due to some, somehow, problematic  configuration on $\outBound B$. 
To this end, we show that after  a short burn-in period typically  such a problematic boundary configuration is 
highly unlike to happen.  See in the next section for further details.

\section{Utilizing uniformity - Rapid mixing for $k>1.76...d$ }\label{sec:from2dtoalphad}
 In the $k>2\Delta$ case,
it is illustrative to consider the case when vertices on the inner boundary of a block have only one
neighbor outside the block.  In this case our new weighting scheme simplifies to the standard Hamming
distance.  In this case the probability of propagation is $\leq 1/(k-\Delta)$ whereas the
branching factor (internal to the block) is $\leq\Delta-2$ and hence these offset when $k>2\Delta$.

Here we want to utilize that when a vertex $z$ has large internal branching factor
(i.e., most of $z$'s neighbors are internal to the block) then these neighbors are not worst-case
but are from the stationary distribution of the block (conditional on a fixed coloring on the
block's outer boundary).  Then we want to exploit the so-called ``local uniformity results'' first
utilized by Dyer and Frieze \cite{DFUniformity} (and then expanded upon in \cite{ TomsUniformity,GMP,DFHVConstantDegree,FOCS16}).
The relevant property in this context is that if a set of $\Delta$ vertices receive independently at random colors 
(uniformly distributed over all $k$ colors) then the expected number of available
colors (i.e., colors that do not appear in this set) 
is $\approx k\exp(-\Delta/k)$.   We'd like to replace the 
probability of propagation from  $1/(k-\Delta)$ to $1/ (k\exp(-\Delta/k))$ which yields the threshold $k>\alpha\Delta$
where $\alpha\approx 1.7632\ldots$ is the solution to $\Delta/(k\exp(-\Delta/k)) = 1$ for $k=\alpha\Delta$.

For a vertex $v$ and the block dynamics $(X_t)$, 
let $\AvialColors_{X_t}(v)$  denote the set of available colors for $v$:
\[ \AvialColors_{X_t}(v) := [k]\setminus X_t(N(v)).
\]
Roughly the local uniformity result says that after a short burn-in period of $O(n)$ steps, 
a vertex $v$ has at least the expected number of available colors with high probability (in $d$).  
Let $\U_t(v)$ denote the event that the block $B(v)$ containing $v$ has been
recolored at least once by time $\leq t$.
We prove the following result that after $C_0 n$ steps the dynamics gets the 
uniformity property at $v$ with high probability, and it maintains it for $Cn$ steps for arbitrary $C$
(by choosing $C_0$ sufficiently large).

\begin{theorem}[Local Uniformity]\label{thrm:Uniformity1.76}
For all $\eps,C>0$, there exists $C_0>0, d_0>1$, for all $d>d_0$, for $k\geq (\alpha+\eps)d$,
let 
${\cal I}=\left [C_0 N, (C+C_0)N\right],$ for $v\in V$,
\[
\Pr\left[ 
\exists t\in {\cal I}\;s.t. \;  |\AvialColors_{X_t}(v)|     \leq  \indicator{\U_t(v)}  (1-\eps^2) k \exp\left (-\degree(v)/k  \right)
 \right]
 \leq \textstyle d^{4} \exp\left( -d^{3/4}\right).
\]  
\end{theorem}

\noindent
The proof of Theorem \ref{thrm:Uniformity1.76} appears in Section \ref{sec:thrm:Uniformity1.76} of the appendix.

Theorem \ref{thrm:Uniformity1.76} builds on \cite{DFUniformity, TomsUniformity}. 
The basic idea is that the vertex $v$ typically gets local uniformity once  most of its neighbors are 
updated at least once, while their interaction is, somehow,  weak prior and during  ${\cal I}$.   
%
%
Since we consider block updates a,  potentially large,  fraction of  $N(v)$ belongs to the same block as $v$.  
Then,  it is possible that the vertex 
gets  local uniformity exactly the moment that its block is updated for  the first time. 
The use of the indicator $\indicator{\U_t(v)}$ expresses exactly this phenomenon.

\subsection{Block dynamics and Burn-in}\label{sec:BurnInAnalysis}
An additional complication with utilizing   local uniformity is the following:
 since the coupling starts from a 
worst-case pair of colorings, in order to attain the local uniformity properties we first
need to ``burn-in'' for $\Omega(n)$ steps so that most neighbors of most vertices are 
recolored at least once.   However during this burn-in stage the initial disagreement at $u^*$
is likely to spread.  

In \cite{DFHVConstantDegree} they consider a ball of radius $O(\sqrt{\Delta})$ around $u^*$.
They show, by a simple disagreement percolation argument, that disagreements
are exponentially (in $\Omega(\sqrt{\Delta})$)  unlikely to escape from this ball.  
Extending this approach to  block dynamics presents an extra challenge.  
Our blocks may be of unbounded size (i.e., a function of $n$) whereas
the ball in which we want to confine the disagreements is constant sized
(roughly $O(\sqrt{d})$ so that the volume of the ball is dominated by the tail 
bound in Theorem~\ref{thrm:Uniformity1.76}).

The disagreements we care about are those on the boundary of a block since these are
the ones that can further propagate.  Hence, let
\[  D_t = (X_t\oplus Y_t) \cap \ImpVrtx.
\]
denote the disagreements at time $t$ which lie on the boundary of some block,
and let $D_{\leq t} = \cup_{r\leq t} D_r$ denote the set of vertices that disagree at some point up
to time $t$.

First we derive a tail bound on the number of disagreements 
generated in $\inBound B$ when the block $B$  has a single
disagreement on its boundary.

\begin{proposition}\label{prop:Tail4ZetaIntro}
For all $\eps>0$, there exists $C>0, d_0>1$, for all $d>d_0$, for $k\geq (\alpha+\eps)d$ 
and any $u^*\in \ImpVrtx$  and any $B$ such that $u^*\in \outBound B$, the following holds.
For a pair of colorings $X_t$ and $Y_t$ such that $X_t\oplus Y_t=\{u^*\}$,
there is a coupling of one step of the block dynamics so that
\begin{equation}\nonumber
\Pr \left [|D_{t+1}\cap \inBound B| \geq \ell \right ]\leq C(dN)^{-1}  \exp\left( -\ell /C \right) \qquad \textrm{for any $\ell\geq 1$}.
\end{equation}
\end{proposition}

\noindent
The idea in proving Proposition \ref{prop:Tail4ZetaIntro}  is to stochastically dominate the disagreements in $B$
 with an independent Bernoulli percolation process. Then we employ a non-trivial martingale argument to get the desired tail bound. 
 The detailed proof appears in Section~\ref{sec:prop:Tail4ZetaIntro}. 


Extending the ideas we develop for  Proposition~\ref{prop:Tail4ZetaIntro} to a setting
where we have multiple disagreements we  prove that a single initial disagreement
at time $0$ is unlikely to spread very far after $O(N)$ steps.
Before formally stating the lemma, let us introduce some basic notation.
For an integer $R$ and vertex $w$, let $\ball(w, R)$ denote
the set of vertices within distance $R$ from $w$ (this is wrt to the graph $G$, independent
of the blocks $\B$). 

\begin{lemma}\label{lemma:Combined-BurnIn}
For all $\eps,C>0$, there exists $C'>0, d_0>1$, for all $d>d_0$, for $k=(\alpha+\eps)d$ 
the following holds.
Consider two colorings $X_0$ and $Y_0$ where $X_0\oplus Y_0=\{u^*\}$ for some $u^*\in V$.
There is a coupling of the block dynamics such that:
for any  $1 \leq \ell<d^{4/5}$,
\[
\Pr \left [|D_{ \leq CN}| \geq \ell \right ]\leq C'\exp\left( -\ell^{\frac{99}{100}}C' \right)
\]
and for $R=\left \lfloor  \epsilon^{-3}(\log d)\sqrt{d} \right \rfloor$ we have
\[
\Pr \left [ \left( D_{\leq CN}  \right) \not\subseteq 
\ball\left (u^*,  R\right) \right] \leq \textstyle 2\exp\left( -d^{0.49} C' \right). 
\]
\end{lemma}

\noindent
The proof of Lemma \ref{lemma:Combined-BurnIn} appears in Section \ref{sec:lemma:Combined-BurnIn} of the appendix.

\paragraph{Rapid mixing:}
We give here a brief sketch of how we derive rapid mixing of the block dynamics
from Theorem~\ref{thrm:Uniformity1.76} and Lemma~\ref{lemma:Combined-BurnIn};
 the high-level idea is inspired by the approach in \cite{DFHVConstantDegree} for graphs of maximum degree $\Delta$.
We apply path coupling and hence we start with a pair of colorings $X_0,Y_0$ 
which differ at a single vertex $u^*$.  
We focus our attention on the ball $\ball$ of radius $O((\log d)\sqrt{d})$ around~$u^*$.  
We first run the chains for a burn-in period of
$T=O(n)$ steps.  
By Lemma~\ref{lemma:Combined-BurnIn} with high probability (in $d$) the
disagreements are contained in this local ball $\ball$ around $u^*$. 
Hence we can focus attention inside this local ball $\ball$ (with high probability).
Since the volume of this ball is not too large,
by Theorem~\ref{thrm:Uniformity1.76} all of the low degree vertices have the local uniformity
property and they maintain it for $O(n)$ steps.  Hence
for $k>\alpha d$ we get contraction for disagreements at low degree vertices.  
Since the vertices at the boundaries of the block are all 
low degree vertices and these are the vertices with non-zero weight $\EdgeBlockWeight()$ 
in our path coupling
analysis as in the proof of Theorem~\ref{thm:block-regular}
for the $k>2\Delta$ case, then we get that the expected distance $\EdgeBlockWeight()$
contracts in every step.  Since the number of disagreements
is not too large (by the second part of Lemma~\ref{lemma:Combined-BurnIn})
after $O(n)$ steps we get that the expected weight is small, and we can
conclude that the mixing time is $O(N\log{N})$.

\subsection{Proof of Proposition \ref{prop:Tail4ZetaIntro}}\label{sec:prop:Tail4ZetaIntro}

We couple one step of the dynamics such that both copies update the same block.
In what follows we describe the coupling when the dynamics updates the block $B$.

We couple $X_{t+1}(B)$ and $Y_{t+1}(B)$ by coloring the vertices of $B$ in a vertex-by-vertex manner.
We start with the vertex $z \in B$ which neighbors the disagreement $u^*$.
Then we proceed by induction by first considering any uncolored vertex in $B$ which
neighbors a disagreement.
The colors $X_{t+1}(z)$ and $Y_{t+1}(z)$  are chosen from the marginal distribution over
the random coloring of $B$ conditional on the fixed coloring outside $B$, and 
the coupling minimizes  the probability that $X_{t+1}(z)\neq Y_{t+1}(z)$. 
For subsequent vertices $v\in B$,
the colors $X_{t+1}(v)$ and $Y_{t+1}(v)$ are from the marginal distributions induced by the pair
of configurations on $\outBound B$ as well as the configuration of the vertices in $B$ that the 
coupling considered in the previous steps.   
If  the current vertex does not neighbor any disagreements
then we can use the identity coupling $X_{t+1}(v)=Y_{t+1}(v)$. 
Similar  inductive couplings have also appeared in, e.g.,~\cite{DFFV, GMP}.

Note that the construction of the set of blocks ${\cal B}$ guarantees that there is exactly one vertex
$z\in B$ which is next to $u^*$. 
Since block B contains at most one cycle $C$, and due to the order of the vertices in the coupling definition, when we couple the color choice for 
$v\notin C$ there can be at most one disagreement in its neighborhood. For the vertices on cycle $C$, the block construction guarantees that $C$
 is deep inside the block (see condition 2(c) in Definition \ref{def:SparseBlockPart}), and hence disagreements are unlikely to even reach this cycle.

We  focus on  the probability that the disagreement ``percolates" from a disagreeing vertex $w\in B\cup\{u^*\}$ to 
some neighbor $v\in B$ in the aforementioned coupling. Specifically, we consider the case where  $\degree(v)\leq \TDeg$  and
$v$ does not belong to the cycle of $B$ (if any).
For such a vertex,  it is standard to show that the probability of the disagreement percolating, i.e., having 
$X_{t+1}(v)\neq Y_{t+1}(v)$ given $X_{t+1}(w)\neq Y_{t+1}(w)$, is upper bounded by the probability of the most likely color for $v$ in both copies of dynamics.
Choosing   $k\geq (\alpha+\epsilon)d$,   the probability of a disagreement 
is upper bounded by $1/((1+\epsilon)\degree_{in}(v))$, where  $\deg_{in}(v)$  the degree of $v$ within $B$.
This bound follows from our results from Section \ref{sec:SCD}, which build on~\cite{GMP}.
Roughly speaking,  the key is that for a random coloring of $B$ and a fixed coloring $\sigma$ on $\overline{B}$, 
then, as in~\cite{GMP}, for a low degree vertex $v$ we have 
$\ExpCond{|\AvialColors(v)|}{\sigma} \lesssim (k-\deg_{out}(v)) \exp(-\deg_{in}(v)/k)\approx (1+\epsilon)\degree_{in}(v)$.

For vertex $v$ which is of degree $>\TDeg$ or belongs to the cycle of the block $B$ (if any)
we just use the trivial bound $1$,  for the probability of disagreement.

We will analyze the spread of disagreements in the coupling above using the following
Bernoulli percolation process. 
Let ${\cal S}_p={\cal S}_p(B)$ be a random subset of the block $B$ such that each vertex 
$v\in B$ appears in ${\cal S}_p$, independently, with probability $p_v$, where 
for $v$ outside the cycle in $B$ we have
\begin{equation} \label{eq:DefPuSimple}
p_v = 
\begin{cases} 
 \frac{1}{(1+\epsilon)  \deg_{in}(v)}
& 
\textrm{if } \degree(v)\leq \TDeg   \\ 
1  & \mbox{otherwise}.
\end{cases}
\end{equation}
If $v$ is on the cycle of $B$, then $p_v=1$.

Consider the random set $X_{t+1}(B)\oplus Y_{t+1}(B)$ induced by the aforementioned coupling.
We will show that the disagreements occurring in our coupling are stochastically dominated  
by  the subset ${\cal C}_{u^*}\subseteq {\cal S}_p(B)$  which contains every vertex  $v$ for which there 
exists a path, using  vertices from ${\cal S}_p$,   that connects $v$ to $u^*$. 
In particular,  $X_{t+1}(B)\oplus Y_{t+1}(B)\subseteq {\cal C}_{u^*}$.
Thus,  let  ${\cal P}_{u^*}={\cal C}_{u^*}\cap \inBound B$. We have
\begin{equation}\label{eq:StochDomOneDis}
\Pr[|D_{t+1} \cap \inBound B| \geq \ell \ |\ \textrm{$B$ is updated at $t+1$}]\leq \Pr[|{\cal P}_{u^*}|\geq \ell ] \qquad \textrm{for any $\ell\geq 0$}.
\end{equation}

\noindent
Then using the independent Bernoulli process we derive the following tail bound.
\begin{proposition}\label{prop:SingleSourceBlockdis} 
In the same setting as in Proposition \ref{prop:Tail4ZetaIntro}, there exists $C>0$
such that for large $d>0$ the following is true:
For any block $B\in {\cal B}$ and any $u^*\in \outBound B$ 
the following holds:
\begin{equation}\label{eq:target4Tail4ZetaIntro}
\Pr[|{\cal P}_{u^*}| \geq  \ell ]\leq C d^{-1}\exp\left( -\ell/C \right) \qquad \textrm{for any $\ell\geq 1$}.
\end{equation}
\end{proposition}

\noindent
The proof of Proposition \ref{prop:SingleSourceBlockdis} appears in Section \ref{sec:prop:SingleSourceBlockdis}.

Proposition \ref{prop:Tail4ZetaIntro} follows from Proposition \ref{prop:SingleSourceBlockdis}, \eqref{eq:StochDomOneDis}
and noting that $B$ is updated in the dynamics with probability $1/N$.

\subsection{Proof of Proposition \ref{prop:SingleSourceBlockdis}}\label{sec:prop:SingleSourceBlockdis}
We define the following weight scheme for the vertices of $B$.
If $B$ is a tree, then we consider the tree $B\cup \{u^*\}$,
with root  $u^*$. Given the root, for each $w\in B$, 
let $\parent{w}$ denote the parent of $w$.

We assign weight $\beta(w)$ to each $w\in B\cup \{u^*\}$. We set   $\beta(u^*)=1$, while  for each $w\in B$ we have
\begin{equation}\label{def:ReverseWeight}
\beta(w)=\min\left \{1, \frac{\beta(\parent{w} )}{(1+\epsilon^2)\ \degree_{in}(\parent{w} ) }
\left( p_w \right)^{-1}
 \right \},
\end{equation}

\noindent
If the block $B$ is unicyclic, then  we choose a  spanning tree of $B$, e.g., $B'$,
and define the parent relation w.r.t.  $B'\cup \{u^*\}$, rooted at $u$. Then we consider the same weight scheme
as in \eqref{def:ReverseWeight}.  Note that we use $B'$ to specify the parent relation only, i.e.,
$p_w$ is defined w.r.t. the degrees in  $B$.

As in Section \ref{sec:prop:Tail4ZetaIntro}, consider the random set ${\cal S}_p\subseteq B$,
where each vertex $v\in B$ appears in ${\cal S}_p$ with probability $p_v$, defined in \eqref{eq:DefPuSimple}.
Let ${\cal C}_{u^*}$ contain every vertex $w\in B$ for which there exists a path  of vertices in ${\cal S}_p$
 that connects $w$ to $u^*$. Note that it always holds that ${\cal P}_{u^*}\subseteq {\cal C}_{u^*}$.
Also,  let
$$
{\cal Z} = \sum_{w \in B} \mathbf{1}{\{ w \in {\cal C}_{u^*} \}} \ \beta(w).
$$

\noindent
From the definition of  $\beta(\cdot)$ it follows that  for each vertex $w\in B$ we have  $0\leq \beta(w)\leq 1$.
Furthermore, we have the following result for the weight of vertices in $B\cap \ImpVrtx$.

\begin{lemma}\label{lemma:Weight2OfBoudnary}
Consider the above weight schema. For any $w\in B\cap \ImpVrtx$ we have  $\beta(w)\geq 1/2$.
\end{lemma}

\noindent
The proof of Lemma \ref{lemma:Weight2OfBoudnary} appears in Section \ref{sec:lemma:Weight2OfBoudnary} of the appendix.

Recall that ${\cal P}_{u^*}={\cal C}_{u^*}\cap \inBound B$.
In light of Lemma \ref{lemma:Weight2OfBoudnary}, it always holds that
$ |{\cal P}_{u^*} | \leq 2{\cal Z}$ which implies that 
\begin{equation}\label{eq:CalZVsTheta}
\Pr[| {\cal P}_{u ^*} | \geq \ell ]\leq \Pr[{\cal Z}\geq \ell /2].
\end{equation}

\noindent
Eq. \eqref{eq:target4Tail4ZetaIntro}  will follow by getting an appropriate tail bound for ${\cal Z}$ and using \eqref{eq:CalZVsTheta}.
Let $z$ be the single neighbor of $u^*$ inside  block $B$. 
For $\ell\geq 1$,  we have that
\begin{equation}\label{eq:Base4FriedmanApp}
\Pr[{\cal Z}\geq \ell /2] \leq  \Pr[{\cal Z}\geq \ell /2 \ | \ z \in {\cal C}_{u^*} ] \Pr[z \in{\cal C}_{u^*} ] \leq Cd^{-1} \Pr[{\cal Z}\geq \ell /2 \ | \ z \in  {\cal C}_{u^*} ].
\end{equation}
The proposition will follow by bounding appropriately the probability term $\Pr[{\cal Z}\geq \ell /2 \ | \ z\in {\cal C}_{u^*} ]$.
For this we are using  a martingale argument. In particular we use the following result  from \cite{cncntration, FreedmanOrig}.
\begin{theorem}[Freedman]\label{thrm:CnctrIneqIntro}
Suppose $W_1 , . . . , W_n$ is a {\em martingale difference sequence}, and $b$  is an uniform upper bound on the steps 
$W_i$. Let $V$ denote the sum of conditional variances,
\[
\textstyle V=\sum^n_{i=1}\Variance(W_i\ |\ W_1, \ldots, W_{i-1}).
\]
Then for every $\alpha, s>0$  we have that
\[
 \Pr\left [ \sum W_i > \alpha \  \textrm{and}\ V\leq s \right] \leq  \exp\left( -\frac{\alpha^2}{2s+2\alpha b/3}\right).
\]
\end{theorem}

\noindent
Consider a process where we expose ${\cal C}_{u^*}$ in a breadth-first-search manner. 
We start by revealing the vertex right next to $u^*$.
Let $z\in B$ be the vertex next to $u^*$ and let  $F_0$ be the event that $z \in {\cal C}_{u^*}$. For 
$i>0$, let  $F_i$ be  the outcome of exposing the $i$-th vertex. 
Let
\begin{eqnarray}
X_0 \ =\  \ExpCond{ {\cal Z}}{F_0 }&\textrm{and}& 
X_i  \  =\  \ExpCond{ {\cal Z}} { F_0, \ldots, F_i}, \nonumber
\end{eqnarray}
for $i\geq 1$.
It is standard to show that  $X_0, X_1, \ldots$ is a martingale sequence.
Also, consider  the martingale difference sequence $Y_i=X_i-X_{i-1}$, for $i\geq 1$.

So as to use Theorem \ref{thrm:CnctrIneqIntro}, we  show the following:
Let  $V=\sum_i \Variance(Y_i \ | \ Y_1, Y_2, \ldots )$. We have that 
\begin{equation}\label{eq:FreedmanConditions} 
(a)  \  X_0 \leq C_1 \qquad (b) \  |X_i-X_{i-1} | \leq  s  \qquad (c) \    V\leq C_2 {\cal Z},
\end{equation}
for positive constants $C_1, C_2$ and $s$.  Before showing that
\eqref{eq:FreedmanConditions} is indeed true, let us show how we use it to get the 
tail bound for ${\cal Z}$.

Assume that the martingale sequence $X_0, X_1, \ldots, $ runs for $T$ steps, i.e., after $T$ steps we have revealed
${\cal C}_{u^*} $.   From Theorem \ref{thrm:CnctrIneqIntro} and \eqref{eq:FreedmanConditions} we get  the following:   
there exists $\hat{C}>0$ such that  for any $\alpha>0$ we have 
\begin{eqnarray}
\Pr[{\cal Z} = \alpha \ |\ z \in {\cal C}_{u^*}  ] & = &\textstyle \Pr\left [ \sum_{i} Y_i = \alpha+X_0 \ \textrm{and}\ V\leq C_2 \alpha \right] \nonumber\\
& \leq & \textstyle \Pr\left [ \sum_{i} Y_i \geq  \alpha+X_0 \ \textrm{and}\ V\leq  C_2 \alpha \right] \leq  \textstyle \exp\left( -2{\alpha}/\hat{C}  \right), 
\end{eqnarray}
where $C_2$ is  defined in   \eqref{eq:FreedmanConditions}.
The first equality follows from the observation that  we always have $V\leq C_2{\cal Z} $.
From the above it is elementary that, for large $C>0$, we have 
\begin{equation}
\Pr[{\cal Z} \geq \alpha \ |\ z\in {\cal C}_{u^*} ]\leq \exp\left( -2{\alpha}/{C}  \right).  \label{eq:FriedmanApp}
\end{equation}
Combining \eqref{eq:FriedmanApp} and \eqref{eq:Base4FriedmanApp} we get that
for $\ell>0$ it holds that
$\Pr[{\cal Z}\geq \ell /2] \leq Cd^{-1}\exp\left( -\ell /C  \right).$
The proposition follows by plugging  the  inequality into  \eqref{eq:CalZVsTheta}.

It remains to show \eqref{eq:FreedmanConditions}.
First we observe the following:
For a vertex $w\in B$, let $F(w)$ be the set of vertices $u$ such that $w=\parent{u}$.
We have that
\begin{equation}\label{eq:claim:MartingaleExpctIncrWeight}
\ExpCond{
 \sum_{v\in F(w)}\beta(v)\ \mathbf{1}\left\{ v\in {\cal C}_{u^*} \right\}}{  w\in {\cal C}_{u^*} }   \leq \frac{\beta(w)}{(1+\epsilon^2)}.
\end{equation}
To see the above note that 
\begin{eqnarray}
\textstyle 
\ExpCond{ \sum_{v\in F(w)}\beta(v) \  \mathbf{1}\left\{ v\in {\cal C}_{u^*} \right\} }{  w\in {\cal C}_{u^*} }   &= &
\textstyle  \sum_{y\in F(w)}  \Pr[y\in {\cal C}_{u^*}  \ | \ w \in {\cal C}_{u^*} ]\ \beta(y) \nonumber \\
&\leq & \degree_{in}(w) \cdot \max_{y\in F(w)}\left \{  \Pr[y\in  {\cal C}_{u^*} \ | \ w \in {\cal C}_{u^*} ]\ \beta(y) \right\}.  \quad
\label{eq:TargetMartingaleExpctIncrWeight}
\end{eqnarray}
Since  $\Pr[y\in {\cal C}_{u^*} \ | \ w\in {\cal C}_{u^*}  ]\leq p_y$, where  $p_{y}$ is  defined in \eqref{eq:DefPuSimple}.
The definition of $\beta(y)$ yields 
\begin{eqnarray}
\Pr[y\in {\cal C}_{u^*} \ | \ w\in {\cal C}_{u^*}  ]\ \beta(y) &\leq  & p_y  \beta(y) 
\ \leq \ \frac{\beta(w)}{\degree_{in}(w) (1+\epsilon^2) }. \nonumber
\end{eqnarray}
Eq. \eqref{eq:claim:MartingaleExpctIncrWeight} follows by plugging the above into
\eqref{eq:TargetMartingaleExpctIncrWeight}.

Now we proceed to prove (a) in \eqref{eq:FreedmanConditions}.
Recall that $z\in B$ is the only vertex next to $u^*\in \partial B$.
Recall, also,  that  $F_0$  is the event that $z\in {\cal C}_{u^*}$.
A simple induction  and \eqref{eq:claim:MartingaleExpctIncrWeight} implies that
\[
\ExpCond{ {\cal Z} }{  z\in {\cal C}_{u^*}  } \leq {2\beta(z)}/{\epsilon^2}.
\]
Since we always have $0<\beta(z)\leq 1$,  (a) in  \eqref{eq:FreedmanConditions} holds for any  $C_1\geq 2\epsilon^{-2}$.

As far as (b) in \eqref{eq:FreedmanConditions} is concerned,
this follows directly from \eqref{eq:claim:MartingaleExpctIncrWeight}  and the fact that
for every $v\in F(w)$ we have $0<\beta(v)\leq 1$.

We proceed by proving (c) in \eqref{eq:FreedmanConditions}.
For a vertex $w\in B$ such that $w\in {\cal C}_{u^*} $, let ${\cal C}_{u^*}^w={\cal C}_{u^*} \cap T_w$,
where $T_w$ is the subtree rooted at $w$, while 
$$
{\cal Z}_w= \textstyle \sum_{v\in T_w}\mathbf{1}\{ v\in {\cal C}_{u^*}^w  \} \ \beta(v).
$$
 Assume that at step $i$ we reveal vertex $w_i$, we have
\begin{eqnarray}
V_i & \leq & \ExpCond{ (X_i-X_{i-1})^2 }{  F_0, F_1, \ldots, F_{i-1} } \nonumber \\
&\leq & \left(  \ExpCond{ {\cal Z}_{w_i} }{  \ w_i\in {\cal C}_{u^*} } \right)^2 \ \leq \  \left( {\beta(w_i)}/{\epsilon^2}\right)^2.\nonumber
\end{eqnarray}
The last inequality follows from  \eqref{eq:claim:MartingaleExpctIncrWeight} and  a simple
induction.  If $w_i \in \outBound {\cal C}_{u^*}$, i.e. it si of small degree and agreeing,  then it is direct that the conditional variance is 
smaller,   it is at most $c_a d^{-2}\beta^2(w_i)$,  for a fixed $c_a>0$. Otherwise, $w_i$ has conditional variance 0.

Using the above, and the fact that $\beta(v)\leq 1$, for any $v\in B$,  we have that
\begin{eqnarray}
V &=&\sum_{i}V_i \ \leq \ 2\sum_{v\in  {\cal C}_{u^*}  } {\beta(v)}/{(\epsilon^4)}
\ \leq \ 2 {\cal Z}/\epsilon^4. \nonumber
\end{eqnarray}
For the third  inequality we need the following: In $V$ there is a contribution from
the vertices in ${\cal C}_{u^*}$, i.e.,   each $v\in {\cal C}_{u^*}$ contributes $\beta^2(v)/\epsilon^4 \leq \beta(v)/\epsilon^4$.
Also, there is a contribution from the vertices in $\outBound {\cal C}_{u^*}\cap B$. For the later 
  we  use the fact that for every $v\in {\cal C}_{u^*}$ the contribution of its children 
that belong to $\outBound {\cal C}_{u^*}\cap B$  is at most $c_a d^{-2} \sum_{w\in F(v)}\beta(w)\leq  c_b d^{-1}\beta(v)$, 
where $c_a$ is  defined previously and  $c_b>0$ is a constant.  Note that the bound on the previous sum follows by working 
as in \eqref{eq:TargetMartingaleExpctIncrWeight}.

Then,  (c) in \eqref{eq:FreedmanConditions}  follows by setting $C_2=2\epsilon^4$. This concludes the proof of Proposition \ref{prop:SingleSourceBlockdis}.
$\hfill \Box$

\section{Conclusions}

Our main contribution is to reduce the ratio $k/d$ to $\alpha\approx 1.763\dots$ for
rapid mixing of the Glauber dynamics on sparse random graphs.  The important aspect is
that the ratio is now comparable to the ratio $k/\Delta$ for
related results concerning rapid mixing of the Glauber dynamics and SSM (strong
spatial mixing) on graphs of bounded degree $\Delta$.
Any improvement in the ratio $\alpha$ would likely
lead to improved results on SSM \cite{GMP}.  In particular, 
our analysis of the spread of disagreements on
a block update builds upon work in \cite{GMP}.  For their purposes they analyze the
expected change in the number of disagreements, whereas we need a concentration bound.  
Hence, significantly improving this ratio $\alpha$ appears to be a major challenge.

\newpage
\appendix

\section{Some remarks about the breakpoints and blocks}\label{sec:BreakPointsObservations}

\noindent
For a graph $G$ which admits a sparse block partition ${\cal B}={\cal B}(\epsilon, d,\Delta)$
we can get an upper bound on the rate at which its grows, starting 
from a breakpoint.  Somehow, it is not surprising that starting from a breakpoint
we have branching factor $\approx d$. More formally, we have the following result.

\begin{lemma}\label{lemma:GrowthFromBP}
Let some $\epsilon>0$, $d>0$, $\Delta >0$ and let $G $ be a graph which admits a sparse block partition 
${\cal B}={\cal B}(\epsilon, d, \Delta)$. Then, for every integer $r\geq 0$ for every $r$-breakpoint $v$ and
for every integer $0\leq \ell \leq r$  the following is true: 

The number of vertices at distance $\ell$ from $v$ is at most $((1+\epsilon/3)d)^\ell$.
\end{lemma}

\begin{proof}[Proof of Lemma \ref{lemma:GrowthFromBP}]

\noindent
For every vertex $w$ in  $G$, and  for every integer $\ell\geq 0$,  recall that  
$\ball_{\ell}(w)$ contains all the vertices within distance $\ell$ from vertex $w$.
Furthermore, let $T^w_{\ell}$ be the shortest path tree of the induced subgraph of $G$ 
which includes only the vertices in  $\ball_{\ell}(w)$. The lemma follows by showing that for 
every $r$-breakpoint $v$ in $G$,  the number of vertices at level $\ell$ of $T^v_{\ell}$ is at most $((1+\epsilon/3)d)^{\ell}$.

Let $D(v, {\ell})$ be the ratio between the number of vertices at level $\ell$ of $T^v_{\ell}$ and $((1+\epsilon/3)d)^{\ell}$.
We  show that $D(v,\ell)\leq 1$. For this, note that  $D(v,\ell)$ satisfies the following recursive relation:
\[
 D(v,\ell)  \leq  \frac{\degree(v)}{(1+\epsilon/3)d} \times \max_{y\in N(v) } \left\{  D(y, \ell-1) \right \}   
\]
where for  $y$,  a neighbor of  $v$,   the quantity  $D(y, \ell-1)$ is equal to the ratio between of the number of
vertices at level $\ell-1$ of the subtree $T_y$ and $((1+\epsilon/3)d)^{\ell-1}$.  $T_y$ is the subtree
of $T^v_{\ell}$ that hangs from the vertex $y$.
Repeating the same recursive argument as above we get that
\begin{equation}
 D(v,\ell)  \leq
\max_{\P'=(u_0=v ,u_1,\dots,u_\ell)}  \prod_{i=0}^{\ell-1} 
\frac{\degree(u_i)}{(1+\epsilon/3)d} , \qquad \label{eq:GraphIncreaseBase}
\end{equation}
where the maximum  is over all paths $\P'$ of length $\ell$ in $T^v_{\ell}$ that start from vertex $v$.

Let $M\subseteq \{u_0, \ldots, u_\ell\}$ be the subset of vertices in $\P'$ which are of high degree, i.e.,
of degree greater than $\TDeg=(1+\epsilon/6)d$. Let $m=|M|$. From \eqref{eq:GraphIncreaseBase} we get that
\begin{eqnarray} 
 D(v,\ell)  &\leq &\left( \frac{1+\epsilon/6}{1+\epsilon/3}\right)^{\ell - m}
 \prod_{u_i\in M}  \frac{\degree(u_i)}{(1+\epsilon/3)d}  
\  \leq \ \left( \frac{(1+\epsilon/6)(1+\epsilon/10)}{1+\epsilon/3}\right)^{\ell - m} d^{-15m}\leq 1. \nonumber
\end{eqnarray} 
where $m=|M|$. The second inequality uses Corollary \ref{cor:FromBreakPointProd}  to bound the product of the degrees in $M$.
The lemma follows.
\end{proof}

Another observation which we use in many different places in the paper is the following corollary,
which follows directly from \eqref{def:PathWeight}.

\begin{corollary}\label{cor:FromBreakPointProd}
For all $\epsilon >0$, $\Delta>0$, there exists $d_0>0$ such that for any $d\geq d_0$, 
for every  graph   $G$ which admits block partition  ${\cal B}(\epsilon, d, \Delta)$, 
and any $v\in \ImpVrtx$ the following is true:

For a multi-vertex block $B$ which is incident to $v$, for any vertex $w\in B$ and a path ${\cal P}$ inside $B$ that connects 
$w$ to $v$ the following holds:
\[
\textstyle \prod_{u\in M} d^{15}\ \degree(u) \leq \left(1+\epsilon/10 \right)^{\ell-m+1},
\]
where $M$ is the set of high-degree vertices in $\cal P$, $\ell$ is the length of the path
 and $m=|M|$.
\end{corollary}

\section{A simple criterion for rapid-mixing}\label{sec:SimpleCriterion}

As in the case of maximum degree $\Delta$, for showing rapid mixing with expected degree 
$d$, we need to show a result which is analogous to  \eqref{eq:BlockUpdtCovergentSimple}.
That is, assume we have some graph $G\in \IntrstGraphFam(\epsilon, d, \Delta)$ with set of blocks ${\cal B}$.
We have $(X_t), (Y_t)$ to copies of block dynamics. At time $t$ we update block $B$, 
while there is exactly one $u^* \in \outBound B$ such that $X_{t}(u^*)\neq Y_{t}(u^*)$. 
For showing rapid mixing it suffices to have that the expected number of disagreements
generated by the update of block $B$ is less than one.  In particular, having such a bound
for the expected number of disagreement, rapid mixing follows by following the same line 
of arguments as those we use for Theorem \ref{thm:block-regular}. 

We couple $X_{t+1}(B)$ and $Y_{t+1}(B)$ by coloring the vertices of $B$ in a vertex-by-vertex manner
as we present at the beginning of Section \ref{sec:prop:Tail4ZetaIntro}.
Our focus is on the {\em probability of propagation}. That is, the probability
vertex $v\in B$ becomes a disagreement  in the coupling, given that  
its neighbor $w\in B\cup\{u^*\}$, which is closest to $u^*$,  is a disagreement, too. Let us call this probability $p_v$.

For the coupling $(X_t)$ and $(Y_t)$ such that $X_t\oplus Y_t=\{u^*\}$ we describe above,  we say that the block $B\in {\cal B}$ is in 
a {\em convergent configuration} if the following is true: We can couple the configurations $X_t(B)$ and $Y_t(B)$
such that  for every $v\in B$ the probability of propagation is bounded as follows:
If $v$ is an internal vertex in the block $B$, it is a low degree vertex, i.e., $\degree(v)\leq \TDeg$ and it  does not belong to 
a  cycle in $B$ (if any) we have  
$$   \textstyle  p_v \leq \min\left \{ \frac{1}{(1+\epsilon/2)\degree(v)}, \frac{2}{d} \right \}.  $$
The same bound holds for   $v\in \ImpVrtx \cap B$ which is   within  radius $(\log d)^2$ from $u^*$, as well.

For a graph $G\in \IntrstGraphFam(\epsilon, d, \Delta)$, wether or not some  block $B$ is in a
convergent configuration depends {\em only} on the configuration that $X_t, Y_t$ specify for $\outBound B$.
It the following result we show that if the block is in a convergent configuration the number of
disagreements that are generated is less than one, on average.

\begin{theorem} \label{thrm:BlockUpdtCovergent}
In the same setting as Theorem \ref{thrm:RapidMixingAvrgDegGraph176} the following is true:

Let $(X_t)_{t\geq 0}, (Y_t)_{t\geq 0}$ be  two copies of the  block dynamics  on
the coloring (or hard-core) model on $G$ such that  for some $t\geq 0$ we
 have $X_t\oplus Y_t=\{u^*\}$, where $u^*\in \ImpVrtx$. 
Let ${\cal E}$ be the event that   $X_t, Y_t$ are such that every $B\in {\cal B}$ for which  $u^*\in \outBound B$,
is in a convergent configuration.
For any such $B$  we have that
\[
\ExpCond{
\left( 
\EdgeBlockWeight(X_{t+1}, Y_{t+1})
- \EdgeBlockWeight (X_{t}, Y_{t})
\right) 
\mathbf{1}\{ {\cal E} \}}{ X_t, Y_t, \ B\  \textrm{is updated at $t+1$}
}
\leq   n^2(1-\epsilon/4).
\]
\end{theorem}

\noindent
The proof of Theorem \ref{thrm:BlockUpdtCovergent} appears in   Section \ref{sec:thrm:BlockUpdtCovergent}.
%

%


\newcommand{\WDis}{{\Phi}}

\section{ Analysis for Rapid Mixing - Proof of Theorem \ref{thrm:RapidMixingAvrgDegGraph176}}\label{sec:RapidMixingColor}

\subsection{Spread of disagreements during Burn-In}\label{sec:BurnInResults}

\noindent
For proving Theorem \ref{thrm:RapidMixingAvrgDegGraph176}, apart from Lemma \ref{lemma:Combined-BurnIn}
we also need the following result.

\begin{proposition}\label{prop:RapidMixAux1}

In the same setting as Theorem \ref{thrm:RapidMixingAvrgDegGraph176} the following is true:

Let $(X_t)_{t\geq 0}$ and $(Y_t)_{t\geq 0}$ be two copies of block dynamics. Assume that   $X_{0}\oplus Y_{0}=\{u^*\}$. Let $T=\left \lfloor CN/\epsilon\right \rfloor$. 
Then there is a coupling such that the following holds:
\begin{enumerate}
\item There exists $C'>0$, independent of $d$, such that
$$
\Exp{|(X_T\oplus Y_{T} ) \cap \ImpVrtx | } \leq  \textstyle \exp\left( C'/ \epsilon \right).
$$ 
\item Let $\ManyDisEvent_T$ be the event that at some time $t\leq T$ we have 
$|(X_t\oplus Y_t)\cap \ImpVrtx|>d^{2/3}$.  Then
\[
\Exp{  |(X_T\oplus Y_T)\cap \ImpVrtx|  \ \mathbf{1}\{\ManyDisEvent_T\}  }   \leq \textstyle \exp\left( -\sqrt{d}\right).
\]
\end{enumerate}
\end{proposition}

\noindent
The proof of Proposition \ref{prop:RapidMixAux1} appears in Section \ref{sec:prop:RapidMixAux1}.

\subsection{Results for Local Uniformity}\label{sec:UniformityVrsRapidMixing}

Additionally to Theorem \ref{thrm:Uniformity1.76} we need the following results:
Recall that for the block dynamics $(X_t)_{t\geq 0}$, and a  vertex $u$, we let $\AvialColors_{X_t}(u)$  be the set of  colors which
 are not used for the  coloring $X_t(N(u))$, where $N(u)$  is the neighborhood of vertex $u$.
Furthermore, for a vertex $u$ and  $t\geq 0$,  let the indicator variable $\indicator{\U_t(v)}$ be equal to 
$1$  if vertex $u$ has been updated  up to time $t$ at least once in $(X_t)_{t\geq 0}$.  Otherwise it is 0.

Lemma \ref{lemma:GrowthFromBP},  Theorem \ref{thrm:Uniformity1.76} and a simple 
union bound imply the following corollary.

\begin{corollary}\label{cor:UniformityUnion}
In the same setting as in Theorem \ref{thrm:Uniformity1.76} the following is true:
Let  $v\in \ImpVrtx$ and let $(X_t)_{t\geq 0}$ be  the block dynamics  on $G$.  
For  $\mathbold{I}_1=\left \lfloor N\log\left(\gamma^{-3}\right) \right \rfloor$ and $\mathbold{I}_2= \left \lfloor CN  \right \rfloor$,  
let the time interval ${\cal I}=\left [\mathbold{I}_1, \mathbold{I}_2\right]$.  For each  $w\in \ball(v, r) \cap \ImpVrtx$, where 
$R=10(\log d)\sqrt{d}$ let the event 
 $$
{\cal Z}_w:= \exists   t \in {\cal I}\;s.t. \;
 |\AvialColors_{X_t}(w)|  \leq  \indicator{\U_t(w)}   (1-\gamma)k \exp\left (-\degree(w)/k  \right).
 $$
 Then, it holds that
\[
\Pr\left [ \textstyle \bigcup_{ w\in \ball(v, R)\cap \ImpVrtx} {\cal Z}_w\right ] \leq 
\textstyle \exp\left( -d^{3/5} \right).
\]
\end{corollary}

\noindent
Theorem \ref{thrm:Uniformity1.76} states that for  $(X_t)_{t\geq 0}$  there is a time period
${\cal I}$ during which   some vertex  $v\in \ImpVrtx$ has local uniformity with large probability. 
Corollary \ref{cor:UniformityUnion}, extends this result  by showing local uniformity  not only for 
$v$, but also for all the vertices in $\ImpVrtx$ which are within distance $10(\log d) \sqrt{d}$ from $v$.

\begin{theorem} \label{thrm:BlockUpdtNowSubcritical}

In the same setting as Theorem \ref{thrm:RapidMixingAvrgDegGraph176} the following is true:

Let $(X_t)_{t\geq 0}, (Y_t)_{t\geq 0}$ be  two copies of the block dynamics  on $G$
such that  for some $t\geq 0$ we have $X_t\oplus Y_t=\{u^*\}$, where $ u^* \in \ImpVrtx$. 
Let ${\cal E}(t)$ be the event that for every $z\in \ball(u^*, (\log d)^2)\cap \ImpVrtx$, we have that
\[
\min\left \{|\AvialColors_{X_{t+1}}(z)|, |\AvialColors_{Y_{t+1}}(z)| \right \} \geq (1-\epsilon/10) k \exp\left (-\degree(z)/k  \right).
\]
For any block $B\in {\cal B}$ such that $u^*\in \outBound B$, it holds that
\[
\ExpCond{
\left( 
\EdgeBlockWeight(X_{t+1}, Y_{t+1})-
\EdgeBlockWeight (X_{t}, Y_{t})
\right) \mathbf{1}\{ {\cal E} (t)\}} { X_t, Y_t, \ B\  \textrm{is updated at $t+1$}}
\leq   n^2(1-\epsilon/4).
\]
\end{theorem}

\noindent
Theorem \ref{thrm:BlockUpdtNowSubcritical} follows as a corollary from Theorem \ref{thrm:BlockUpdtCovergent}
once we notice that when the event ${\cal E}$ occurs the block $B$ is in a convergent configuration.

\subsection{ Proof of Theorem \ref{thrm:RapidMixingAvrgDegGraph176}}

\begin{proposition}\label{prop:RapidMixAux2}

In the same setting as Theorem \ref{thrm:RapidMixingAvrgDegGraph176}, 
there exists $C_1>0$ such that for large $d>0$ the following is true:

Let $(X_t)_{t\geq 0}$ and $(Y_t)_{t\geq 0}$ be two copies of block dynamics  
with set of block ${\cal B}$. Assume that  that $X_{0}\oplus Y_{0}=\{u^* \}$, where 
$u^*\in \ImpVrtx$.
Let $T_m=\lfloor  C_1 N/\epsilon \rfloor$. Then there is a coupling such that 
\[
\Exp{\EdgeBlockWeight(X_{T_m}, Y_{T_m}) }  \leq  (1/3)\ \EdgeBlockWeight(X_0, Y_0).
\]
\end{proposition}

\noindent
The proof of Proposition \ref{prop:RapidMixAux2} appears in Section \ref{sec:prop:RapidMixAux2}.

\begin{proof}[Proof of Theorem \ref{thrm:RapidMixingAvrgDegGraph176}]
For arbitrary colorings $\sigma, \tau$, consider two copies of block dynamics 
$(X_t)_{t\geq 0}$ and $(Y_t)_{t\geq 0}$ such that $X_0=\sigma$ and $Y_0=\tau$. 
The theorem follows by showing that there is a sufficiently large constant $C>5$, such that
for $T=C n \log n$ we have that $\Pr[X_T\neq Y_T] \leq e^{-1}$. 
It suffices to show that
\begin{equation}\label{eq:Target4RapidMixingAvrgDegGraph176}
\Pr[\EdgeBlockWeight(X_{T}, Y_{T} )> 0] \leq e^{-1}.
\end{equation}

\noindent
For bounding $\Pr[\EdgeBlockWeight(X_{T'}, Y_{T'} )> 0]$ we use path coupling.

Letting $h=(X_0\oplus Y_0) $ and  an arbitrary ordering of the vertices in 
$(X_0\oplus Y_0)$,  e.g., $w_1,\ldots, w_{h}$, we  interpolate $X_0, Y_0$ by using the configurations 
$\{Q_i\}^h_{i=0}$, such that  $Q_0=X_0,  Q_1, \ldots, Q_{h}=Y_0$. 
Furthermore, $Q_{i}$ is obtained from $Q_{i-1}$ by  changing the color of  
$w_i$ from $X_t(w_i)$ to $Y_t(w_i)$. 
Also, let  $Q'_{i+1}, Q'_{i}$ be the resulting pair after coupling $Q_{i+1}$ with $Q_i$ for
$T$ many steps.

If $w_i$ is an internal vertex in some block, then, as we argued in Theorem 
\ref{thm:block-regular} the disagreement does not spread. It only vanishes once we update
its block. Then, we get that
\begin{equation}\nonumber
\Exp{\EdgeBlockWeight(Q'_{i+1}, Q'_i)  } \leq \left(1-1/N \right)^{Cn\log n} \EdgeBlockWeight(Q_{i+1}, Q_i)   \leq n^{-5}\EdgeBlockWeight(Q_{i+1}, Q_i),
\end{equation}
where in the last inequality we use the fact that $N\leq n$. Note that 
$\EdgeBlockWeight(Q_{i+1}, Q_i)=1$.

For $w_i$ which on the boundary of its block,  we use  Proposition \ref{prop:RapidMixAux2}  and get that
$$
\Exp{\EdgeBlockWeight(Q'_{i+1}, Q'_i)  }  \leq  n^{-5} \ \EdgeBlockWeight(Q_{i}, Q_{i+1}).
$$
Then,  path coupling implies that
\begin{equation}\label{eq:ExpctWeightBound}
\Exp{\EdgeBlockWeight(X_{T}, Y_{T}) }
\leq  n^{-5} \EdgeBlockWeight(X_{0}, Y_{0})  \leq n^{-1}.
\qquad \mbox{[since $\EdgeBlockWeight(X_{0}, Y_{0})<2d n^3$]}
\end{equation}
Then we get \eqref{eq:Target4RapidMixingAvrgDegGraph176}  by using  \eqref{eq:ExpctWeightBound}  and Markov's inequality.

For showing that the block update requires $O(k^3 B_{\rm max})$ steps we use the fact that the blocks are trees with at most one 
extra edge. Implementing a transition of the block  dynamics is equivalent to generating a random {\em list coloring} of the
block $B$. 
List coloring is a generalization of the coloring problem, where each vertex $u$ is assigned with a list of available colors
$L(u)$.  Assume that $L(u)\subseteq [k]$.
 In out setting,  when updating block $B$, each vertex $w\in B$ can choose from all but the colors appearing in $N(w)\setminus B$.

It is standard to show that dynamic programing can compute the number of list colorings of a tree efficiently. In particular, 
for a tree on $h$ vertices, the number of list coloring can be computed in time $h \cdot k$.
For our case we consider counting list colorings of a unicyclic block, as well. 
For such a  component, we can simply consider all $\leq ? k^2$ colorings for the endpoints of the extra edge (i.e. arbitrary edge 
in the cycle)  and then recurse on the remaining tree.  It is immediate that this counting requires time $k^3 \cdot r$, for a block of size $r$.
All the above imply that the block updates requires no more time than $O(k^3 B_{\rm max})$.

The theorem follows.
\end{proof}

\subsection{Proof of Proposition \ref{prop:RapidMixAux2}}\label{sec:prop:RapidMixAux2}

Let $T_b=\lfloor N \log\left( (\epsilon/15)^{-1} \right) \rfloor$.  Since $T_m=\lfloor  C_1N/\epsilon \rfloor$, we apply  
Theorem \ref{thrm:Uniformity1.76} and Corollary \ref{cor:UniformityUnion} to conclude that the  necessary
 local  uniformity pro\-per\-ties hold with high probability for all vertices in 
$\ball (v, R')\cap \ImpVrtx$, where $R'=10(\log d)\sqrt{d}$, for all  $t\in I:=[T_b,T_m]$.
We show that the expected $\EdgeBlockWeight(X_t, Y_t)$  decreases for $t\in I$.

For $t\geq T_b$ consider the following events:
\begin{itemize}
\item $\ManyDisEvent(t)$ denotes the event that at some time $s\leq t$, we have $|(X_s\oplus Y_s)\cap \ImpVrtx|\geq d^{2/3}$
\item $\OutOfBalls(t)$  denotes the event that $D_{\leq t} \not\subseteq  \ball(v, R)$, for $R=(\log d)\sqrt{d}$
\item  $\FewColors(t)$ denotes the event that there exists a time  $s\in [T_b,t]$ and $z\in \ball (v, R')\cap \ImpVrtx$,
for $R'=10(\log d)\sqrt{d}$,  such that
\[
\AvialColors_{X_t}(z) < \indicator{\U_t(z)}  (1-\epsilon/15) k  \exp\left(-\degree(z)/k \right).
\]
$ \indicator{\U_t(z)}$ is equal to one if $z$ is updated up to time $t$ (including  $t$),  otherwise it is zero. 
\end{itemize}

\noindent
For the sake of brevity,  let the events
\[
{\cal B}(t)=\OutOfBalls(t)\cup\FewColors(t) \quad \textrm{and} \quad {\cal G}(t)=\bar{\ManyDisEvent(t)}\cap \bar{{\cal B}(t)}.
\]
For any $t>0$,  let $\EdgeBlockWeight_t=\EdgeBlockWeight(X_{t}, Y_{t})$.  We have that
\begin{eqnarray}
\Exp{\EdgeBlockWeight_{T_m}}  &=& \Exp{ \EdgeBlockWeight_{T_m} \mathbf{1}\{\ManyDisEvent\}}+
\Exp{ \EdgeBlockWeight_{T_m} \mathbf{1}\{\bar{\ManyDisEvent} \}  \  \mathbf{1}\{ {\cal B} \}  }+
\Exp{ \EdgeBlockWeight_{T_m} \mathbf{1}\{ {\cal G}\} } \nonumber \\
&\leq & \Exp{ \EdgeBlockWeight_{T_m} \mathbf{1}\{\ManyDisEvent\} } +
2d^{1+2/3} n^2\ \Pr[{\cal B}] +
\Exp{ \EdgeBlockWeight_{T_m} \mathbf{1}\{ {\cal G}\} }. \label{eq:Basis4Aux2}
\end{eqnarray}

\noindent
The second derivation uses  that for each $w\in \ImpVrtx$ we have $\degree(w)\leq \TDeg<2d$.

We have that
\begin{equation}\label{eq:ManyDis4BasisAux2}
\Exp{\EdgeBlockWeight_{T_m} \mathbf{1}\{\ManyDisEvent\}}  \ \leq \ 
 n+n^2\TDeg \ \Exp{  |(X_{T_m}\oplus Y_{T_m})\cap \ImpVrtx|  \ \mathbf{1}\{\ManyDisEvent\}  }
\ \leq \  \textstyle 2n^2d \exp\left(-\sqrt{d} \right),
\end{equation}
where the second inequality follows from Proposition \ref{prop:RapidMixAux1}.
Furthermore, we have that
\begin{equation}\label{eq:ProbOfB4BasisAux2}
\Pr[{\cal B}]  \leq  \Pr[{\cal B}_1(T_m)]+ \Pr[{\cal B}_2(T_m)]  \ \leq \ \textstyle \exp  \left(-d^{1/3}\right).
\end{equation}
The first inequality above follows from the union bound, while the second is from
Corollary \ref{cor:UniformityUnion} and Theorem \ref{thrm:BallResult}.
Finally, we use that
\begin{equation}\label{eq:ConvergenceAllGood}
\Exp{\EdgeBlockWeight_{T_m} \mathbf{1}\{ {\cal G}\}} \leq  (1/9) n^2\ \degree_{out}(u^*).
\end{equation}

\noindent
Before showing that \eqref{eq:ConvergenceAllGood} is indeed true,  we note that the proposition follows
by plugging \eqref{eq:ManyDis4BasisAux2},  \eqref{eq:ProbOfB4BasisAux2} and \eqref{eq:ConvergenceAllGood}
into \eqref{eq:Basis4Aux2} and noting that $\EdgeBlockWeight(X_0, Y_0)=n^2\degree_{out}(u^*)$.

We conclude this proof by showing that \eqref{eq:ConvergenceAllGood}
is indeed true.  For this we use path coupling.
Let ${\cal M}_0=X_t, {\cal M}_1, {\cal M}_2, \ldots, {\cal M}_{h_t}=Y_t$  be a sequence of colorings where
$h_t=|(X_t\oplus Y_t) |$.  
Consider an arbitrary ordering of the vertices in $(X_t\oplus Y_t)$, 
e.g., $w_1,  \ldots, w_{h_t}$.  For each $i$,  we obtain ${\cal M}_{i+1}$ from ${\cal M}_i$ by changing the color of  
$w_i$ from $X_t(w_i)$ to $Y_t(w_i)$.

We  couple ${\cal M}_i$ and ${\cal M}_{i+1}$, maximally, in one step of the block-dynamics to obtain
${\cal M}'_i$, ${\cal M}'_{i+1}$. More precisely, both chains recolor the same block, and maximize the 
probability  of choosing the same new color for the chosen vertex. Let $B_i$ be the block that $w_i$ belongs to.

If $w_i$ is internal in the block $B_i$, then we have that
\begin{eqnarray}\label{eq:w_iInternal}
\ExpCond{ \EdgeBlockWeight ({\cal M}'_i,  {\cal M}'_{i+1})- \EdgeBlockWeight ({\cal M}_i,  {\cal M}_{i+1})}{ {\cal M}_i,  {\cal M}_{i+1} }
&\leq & -1/N.
\end{eqnarray}

Consider $w_i \in \inBound B_i$,  With probability $1/N$ both chains recolor block $B_i$.
Since there is no disagreement at $\outBound B_i$,  we can couple 
${\cal M}_i$ and ${\cal M}_{i+1}$ and    the ``distance" reduces by $n^2\degree_{out}(w_i)$.

Now, consider $z\in N(w_i)\setminus B_i$ and  assume that $z$ belongs to a {\em single vertex block} $B$.
Let $c_1={\cal M}_i(w_i)$ and $c_2={\cal M}_{i+1}(w_i)$. 
Then, a direct observation is that since ${\cal M}_i(w_i)=c_1$ and $z$ is a neighbor of $w_i$,
we have ${\cal M}'_i(z)\neq c_1$ with probability 1.  On the other hand, it could be that $c_1$ is available for
$W'_{i+1}(z)$, if $c_1$ is not used in ${\cal M}_{i+1}$ to color any of the neighbors of $z$.
Similarly,  we have that we have ${\cal M}'_{i+1}(z)\neq c_2$ with probability 1, while
${\cal M}'_{i}(z)$ could be set $c_2$ if   $c_2$ is not used in ${\cal M}_{i}$ to color any of the
 neighbors of $z$.

Therefore, given ${\cal M}_i, {\cal M}_{i+1}$,  for vertex $z\in N(w_w)$ which belongs to a single vertex block, 
we have that
\begin{eqnarray}
\delta_s(z)&:=& n^2 \degree(z)\times \Pr[{\cal M}'_i(z)\neq {\cal M}'_{i+1}(z)\ | \ {\cal M}_i, {\cal M}_{i+1}, \ \textrm{$z$ is updated} ] \nonumber \\
&\leq& n^2\degree(z)\times \frac{\mathbf{1}\left\{ U({\cal M}_i,z,w_i, c_1,c_2) \right\} }
{\min\{\AvialColors_{{\cal M}_i}(z), \ \AvialColors_{{\cal M}_{i+1}}(z)\}},
\end{eqnarray}
where 
$$
U({\cal M}_i,z,w_i, c_1,c_2) = \left \{
\begin{array}{lcl}
1 &\quad & \textrm{if $\{c_1, c_2\}\not\subset X_t(N(w)\setminus \{c\})$} \\ \vspace{-.3cm}\\
0 && \textrm{otherwise.}
\end{array}
\right .
$$

\noindent
Consider $z\in N(w_i)\setminus B_i$ and  assume that $z$ belongs to a multi vertex block
which we call $B_z$.
Then, the number of disagreements introduced is 
$$
\delta_m(z):= 
\ExpCond{ \EdgeBlockWeight({\cal M}_i, {\cal M}_{i+1}) -\EdgeBlockWeight ({\cal M}'_i,  {\cal M}'_{i+1})}
{ {\cal M}_i,{\cal M}_{i+1}, \ \textrm{$B_z$ is updated} } .
$$
Then, we get that
\begin{eqnarray}
\lefteqn{
\ExpCond{ \EdgeBlockWeight ({\cal M}'_i,  {\cal M}'_{i+1})- \EdgeBlockWeight ({\cal M}_i,  {\cal M}_{i+1})}{ {\cal M}_i,  {\cal M}_{i+1} }
} \hspace{2.5cm} \nonumber \\ 
&\leq & \textstyle
N^{-1} \left (-n^2\degree_{out}(w_i)  +  \sum_{z\in N(w_i)\setminus B_i }  
\indicator{ S_z}\delta_{s}(z) +  
\left(1-\indicator{ S_z} \right)\   \delta_m(z)   
\right)\qquad   \label{eq:OneStepDistanceModUniform}
\end{eqnarray}
where $\mathbf{1}\{S_z\}$ is equal to one if vertex $z$ belongs to a single vertex block, otherwise it is zero.

We proceed by bounding $\delta_m(z)$ and $\delta_s(z)$, for every $z\in N(w_i)\setminus B_i$.
First note that the bound for $X_0$ in  \eqref{eq:FreedmanConditions}
implies that updating $B_z$, the block that $z$ belongs to, the expected number of vertices in
$B_z\cap \ImpVrtx$ is $C/d$, for some large constant $C$ which is independent of $d$.
Since every vertex in $\ImpVrtx$ has degree at most $\TDeg=(1+\epsilon/6)d$, 
updating $B_z$ we increase the expected distance between the configurations  by $(1+\epsilon/6)Cn^2$.

The above implies that there is $C_2>0$, independent of $d$, such that 
\[
\ExpCond{ \EdgeBlockWeight ({\cal M}'_i,  {\cal M}'_{i+1})- \EdgeBlockWeight ({\cal M}_i,  {\cal M}_{i+1})}{ {\cal M}_i,  {\cal M}_{i+1} }
\leq  C_2 N^{-1} n^2 \degree_{out}(w_i).
\]

\noindent
Therefore, given $X_t, Y_t$, we have
\begin{equation}\label{eq:HamDistBeforeBurnIn}
\ExpCond{ \EdgeBlockWeight(X_{t+1}, Y_{t+1} )}{ X_t, Y_t} \leq\left(1+ C_2/N \right)\EdgeBlockWeight(X_t, Y_t).
\end{equation}
This bound will be used only for the burn-in phase, i.e., the first $T_b$ steps. For the remaining 
$T_m-T_b$ steps we show that we have {\em contraction}.

 For all   $t\in [T_b, T_m]$, assuming that assuming that
 ${\cal G}(t)$ holds we have the following: For all $0\leq i\leq h_t$,   $z\in \ball(w_i, R)\cap \ImpVrtx$ ,  we have
 \[
 \AvialColors_{{\cal M}_i }(z)\geq \AvialColors_{X_t}(z)-d^{2/3}\geq \Theta_0-d^{2/3}.
 \]
The first inequality follows from the assumption that $\ManyDisEvent(t)$ occurs. The second inequality
comes from our assumption that $\FewColors(t)$ holds. 
Hence, for $t\in [T_b, T_m]$, given ${\cal M}_i,{\cal M}_{i+1}$ and assuming ${\cal G}(t)$, then for $z\in N(w_i)\setminus B_i$ 
which belongs to a single vertex block, we have that
\begin{equation}\label{eq:Delta_sGood}
\delta_s(z) \leq n^2\degree(z)\left(  \Theta_0-d^{2/3} \right)^{-1}\leq n^2(1+\epsilon/3)^{-1}.
\end{equation}
If $z\in N(w_i)\setminus B_i$ belongs to a multi vertex block, then from 
Theorem \ref{thrm:BlockUpdtNowSubcritical}  we have 
\begin{equation}\label{eq:Delta_mGood}
\delta_m(z)\leq n^2(1-\epsilon/4).
\end{equation}  
Combining \eqref{eq:Delta_sGood}, \eqref{eq:Delta_mGood} and  \eqref{eq:OneStepDistanceModUniform}
we get that
\begin{eqnarray}
\lefteqn{
\ExpCond{ \EdgeBlockWeight ({\cal M}'_i, {\cal M}'_{i+1})- \EdgeBlockWeight({\cal M}_i, {\cal M}_{i+1})}{{\cal M}_i, {\cal M}_{i+1} }
}  \hspace{2.6cm}  \nonumber \\
&\leq & n^2
  \left[ -\degree_{out}(w_i)+ (1-\epsilon/5)\degree_{out}(w_i)  \right]
\leq  -  (\epsilon/5) N^{-1}n^2 \degree_{out}(w_i). \nonumber 
\end{eqnarray}
The above and \eqref{eq:w_iInternal} imply that
\begin{equation}\label{eq:HamDistAfterBurnIn}
\ExpCond{ \EdgeBlockWeight(X_{t+1}, Y_{t+1}) \ {\cal G}(t)}{ X_{t}, Y_{t} } \leq  \left(1- (\epsilon/6)N^{-1} \right) \EdgeBlockWeight(X_t, Y_t).
\end{equation}
Let $t\in [T_b, T_m-1]$. We have 
\begin{eqnarray}
\Exp{ \EdgeBlockWeight_{t+1} \ \mathbf{1}\{ {\cal G}(t)\}} &=&
\Exp{  \ExpCond {\EdgeBlockWeight_{t+1} \mathbf{1}\{ {\cal G}(t)\} }{X_0, Y_0, \ldots, X_t, Y_t } } \nonumber \\
&=&
\Exp{  \ExpCond{ \EdgeBlockWeight_{t+1}}{X_0, Y_0, \ldots, X_t, Y_t } \ \mathbf{1}\{ {\cal G}(t)\} } \nonumber \\
&\leq &
\left(1-(\epsilon/5) N^{-1} \right)\Exp{ \EdgeBlockWeight_{t} \ \mathbf{1}\{ {\cal G}(t)\} } \nonumber \\
&\leq &
\left(1-(\epsilon/5)N^{-1} \right)\Exp{ \EdgeBlockWeight_{t}\  \mathbf{1}\{ {\cal G}(t-1)\} }. \nonumber
\end{eqnarray}
The first equality is Fubini's Theorem, while the second equality is because ${\cal G}(t)$ is determined by
$X_0, Y_0, \ldots X_t, Y_t$. The first inequality uses \eqref{eq:HamDistAfterBurnIn}, while the last derivation
follows from the observation that ${\cal G}(t-1)\subset {\cal G}(t)$. 
Using a simple induction, we get
\[
\Exp{  \EdgeBlockWeight_{T_m} \  \mathbf{1}\{ {\cal G} \}}  \leq \left( 1-(\epsilon/5)N^{-1})\right)^{T_m-T_b}
\Exp{ \EdgeBlockWeight_{T_b} \ \mathbf{1}\{ {\cal G}(T_b)\} }.
\]
Also, using  \eqref{eq:HamDistBeforeBurnIn}  and the same arguments as above, we get that
\[
\Exp{ \EdgeBlockWeight_{T_b} \ \mathbf{1}\{ {\cal G} \} }  \leq \left( 1+C_2/N \right)^{T_b} \EdgeBlockWeight_0.
\]
Combining the two above inequalities  we get
\begin{equation}\label{eq:TbVsTM}
\Exp{ \EdgeBlockWeight_{T_m} \mathbf{1}\{ {\cal G} \} } \leq
\left( 1- (\epsilon/5)N^{-1} \right)^{T_m-T_b}\left( 1+ C_2 N^{-1} \right)^{T_b} \EdgeBlockWeight_0.
\end{equation}
The proposition follows by choosing sufficiently large $C_1>0$ in the expression  $T_m=\lfloor C_1N/\epsilon \rfloor$.

\section{Spatial Correlation Decay}\label{sec:SCD}

In this section we present some  results for  the coloring model.  These results are mainly used in the 
context of {\em disagreement percolation} \cite{DisPerc} to, essentially, derive  spatial correlation decay. 
Particularly, they are useful for studying the spread of disagreements during  burn-in of the block dynamics,  see
Section \ref{sec:DisPerc}, as well as the comparison arguments in Section \ref{sec:thm:1.76-Glauber}.

For some given $\epsilon, d, \Delta$ and any graph $G\in \IntrstGraphFam(\epsilon, d,\Delta)$,
we denote by  $\HighDegreeSet$  the set of vertices $v$ such that $\degree(v)>\TDeg$.
We  use the  technical result ~\cite[Lemma 15]{GMP}  to get the following  corollary.

\begin{corollary} 
\label{corr:GMPBiasColour}
For  $\epsilon, d, \Delta, k$   as in  Theorem \ref{thm:1.76-main},
let $G\in \IntrstGraphFam(\epsilon, d,\Delta)$.
Also, let $Z$ be a random $k$-coloring of  $G$.
For any $B\in {\cal B}$, for   any $v\in B$ which does not belong to a
cycle inside $B$  and being such that $\degree(v)\leq \TDeg$, while  $N(v)\cap  \HighDegreeSet=\emptyset$ the following is true:

For any  $B' \subseteq B\backslash\{u\}$, let $B^+=B'\cup \outBound B$. 
For any $c\in [k]$  and any fixed $k$-coloring $\sigma\in [k]^{B\cup \outBound B}$ 
we have that 
\begin{equation}\nonumber
\Pr[Z(u)=c\ |\  Z(B^+_u)=\sigma(B^+_u)] \leq \frac{1}{\max\{1, |N(u)\backslash  B^+| \}}\frac{1}{1+\eps}.
\end{equation}
\end{corollary}

\noindent
Perhaps the above corollary is most useful when we consider
  $u \in B\cap \ImpVrtx$ and $B'=\emptyset$.  Then,  essentially, it implies that 
\[
\Pr[Z(u)=c\ |\  Z(B^+_u)=\sigma(B^+_u)] \leq \frac{1}{\degree_{in}(u) }\frac{1}{1+\eps}.
\]

\noindent
Corollary \ref{corr:GMPBiasColour} is restricted to low degree vertices which are not next to a high degree vertex. 
For the vertices deep inside a block $B$ which  are not as those  in Corollary \ref{corr:GMPBiasColour},
we have the following result:

\begin{proposition}\label{prop:SpatialMixing4HighDegrees}
For  $\epsilon, d, \Delta, k$   as in  Theorem \ref{thm:1.76-main},
let $G\in \IntrstGraphFam(\epsilon, d,\Delta)$.
Let $Z$ be a random $k$-coloring of  $G$. 
For any $B\in {\cal B}$,  let   $w \in B$ for which either of the  following three
holds: either  $w\in \HighDegreeSet$, either $w\notin \HighDegreeSet$ but
$N(w)\cap  \HighDegreeSet\neq \emptyset$, or $w$ belongs to the unique cycle
in  $B$,  the following is true:

For any  $u\in N(w)$, let $B^+= \outBound B\cup \{u\}$. 
For any $c\in [k]$  and any fixed $k$-coloring $\sigma\in [k]^{B\cup \outBound B}$ it holds that
\begin{equation}\label{eq:SpatialMixing4HighDegrees}
\Pr[Z(w)=c\ |\  Z(B^+)=\sigma(B^+ )]\leq (k- 2)^{-1}+20d^{-2}.
\end{equation}
\end{proposition}

\noindent
The proof of Proposition \ref{prop:SpatialMixing4HighDegrees} appears in Section \ref{sec:prop:SpatialMixing4HighDegrees}.

Note that a vertex  $w$ as in  Proposition \ref{prop:SpatialMixing4HighDegrees} should be, somehow, 
away from the boundary of its block. The above proposition implies that any configuration at $\outBound {\cal B}$ 
has essentially no effect on the marginal of the configuration at $w$.
Finally,  we  have the following easy to show result.
\begin{corollary}\label{cor:GeneralBiasColour}
For any $k>0$, for any $k$-colorable graph $G=(V,E)$ and any $k$-coloring 
$\sigma$ the following is true:
Let $Z$ be a random $k$-coloring of $G$. For
any  $v\in V$ and any  $c\in [k]$  it holds that
\begin{equation}\nonumber
\Pr[Z(u)=c\ |\  Z(N(u) )=\sigma(N(u)) ]  \leq  \left \{
\begin{array}{lcl}
\frac{1}{k-\degree(v)} &\quad&\textrm{if $\degree(u)<k$} \\
1 && \textrm{otherwise.}
\end{array}
\right.
\end{equation}
\end{corollary}

\subsection{Proof of Proposition  \ref{prop:SpatialMixing4HighDegrees}}\label{sec:prop:SpatialMixing4HighDegrees}

So as to prove Proposition \ref{prop:SpatialMixing4HighDegrees} first we consider the
case where $w$ is either a high degree vertex or next to a high degree vertex, i.e.,
$w$ does not belong to a cycle in $B$, if any.
For such vertex $w$ we will show that \eqref{eq:SpatialMixing4HighDegrees} is true.

First, consider the case where $B$ is a unicyclic block, e.g. consider the block in Figure \ref{fig:UnicyclicBlock}.  
Let $C$ be the cycle in $B$. Let $C_{adj}$ be the set of vertices in $B$ that is adjacent to the cycle.
Our assumptions imply that  there is  $x\in C_{adj}$  such that  $w\in T_x$. We let $T_{x,w}$ be the subtree of 
$T_x$ rooted at vertex $w$.

Let $e=\{w,v \}$ be the edge that connects $T_{x,w}$ with the rest of the block $B$. W.l.o.g. assume that
$v\neq u$. There is a 
probability measure $\nu:[k]\to [0,1]$ such that the following holds:
Let $Z$ be a random coloring of $B\cup \outBound B$. 
\begin{eqnarray}
\Pr[Z(w)=c \ | \ Z(B^+)=\sigma( B^+)] &=& \sum_{q\in [k]}\nu(q) \Pr[X(w)=c \ | \ Z(B^+)=\sigma(B^+), Z(v)=q] 
\nonumber \\
&\leq & \max_{q\in [k]}\left\{   \Pr[X(w)=c \ | \ Z(B^+)=\sigma(B^+), Z(v)=q]  \right\}.
\label{eq:Reduction2SingleBoundary}
\end{eqnarray}
It is elementary that we can write the probability term $\Pr[Z(w)=c \ | \ Z(B^+)=\sigma(B^+), Z(v)=q] $
in terms of the Gibbs distribution over $T_{x,w}$.  That is,  let $X$ be a random
$k$-coloring of $T_{x,w}$, then 
\begin{eqnarray}
\lefteqn {
\Pr[Z(w)=c \ | \ Z(B^+ )=\sigma( B^+), Z(v)=q] 
}\hspace{1.8in} \nonumber \\
 &= &
\Pr[X(w)=c \ | \ X(B^+\cap T_{x,w} )=\sigma(B^+\cap T_{x,w}), X(w) \neq q]  \nonumber \\ 
&=& \frac{\Pr[X(w)=c \ | \ X( B^+  \cap T_{x,w})=\sigma(B^+\cap T_{x,w})]  }{\sum_{c'\in [k]\setminus \{q, \sigma(u)\}} \Pr[X(w)=c' \ | \ X(B^+ \cap T_{w,x})=\sigma( B^+\cap T_{x,w}) ]  }. \qquad  \label{eq:BayesRuleWithSpatial}
\end{eqnarray}

\noindent
To this end, we utilize the following result, whose proof appears in Section \ref{sec:prop:UnBiasHighDegree4Trees}. 
\begin{proposition}\label{prop:UnBiasHighDegree4Trees}
For  $\epsilon, d, \Delta, k$   as in  Theorem \ref{thm:1.76-main},
 let $G\in \IntrstGraphFam(\epsilon, d,\Delta)$.
Consider  $B\in {\cal B}$ which contains a single cycle $C$. For any $x\in C_{adj}$, for any $w\in T_x$ such that either 
$w\in \HighDegreeSet \cap T_x$  or  $N(w)\cap \HighDegreeSet \neq \emptyset$, 
for any $c\in [k]$ and any $\tau$, a $k$-coloring of $B\cap \outBound B$,
the following is true:

For   $X$  a random $k$-coloring of $T_{x, w}$ we have that 
\begin{equation}\nonumber
\left| \Pr[X(w)=c\ | \ Z(\outBound B)=\tau(\outBound B)  ]-1/k\right| \leq d^{-11}. 
\end{equation}
\end{proposition}

\noindent
Combining \eqref{eq:BayesRuleWithSpatial}  with Proposition \ref{prop:UnBiasHighDegree4Trees}, we get that
\[
\Pr[Z(w)=c \ | \ Z( B^+ )=\sigma(B^+), Z(v)=q] 
 \leq (k-2)^{-1}+d^{-10}.
\]
Eq. \eqref{eq:SpatialMixing4HighDegrees} follows  from the above and \eqref{eq:Reduction2SingleBoundary} 
for   the case where $B$ is unicyclic.  The case where $B$ is a tree  is very similar,  for this reason we omit it.
For proving the proposition, it remains to consider the case where $B$ is unicyclic and
$w$ is a vertex on the unique cycle $C$ in the block.

 Let the cycle $C:=w_0, w_1, \ldots, w_{\ell-1}$ be the unique cycle in B, for some $\ell\geq 3$.
For each $w_i\in C$, let ${\cal T}_{w_i}$ be the subgraph of $B$ that corresponds to the set of vertices in the 
connected component of  $B$ that contains vertex $w_i$ once we delete all the edges of $C$.
Let $\outBound {\cal T}_{w_i}$ be the subset of vertices in $\outBound B$
which are incident with ${\cal T}_{w_i}$. 

For what follows, we assume that $w$ is a vertex in $C$ and let  $T^+={\cal T}_w\cup \outBound {\cal T}_w$.
Working as for Proposition \ref{prop:UnBiasHighDegree4Trees}, we get the following:
let $Z$ be a random $k$-coloring of $T^+$. Then,
for every $c\in [k]$ and any $\tau$, a $k$-coloring of $T^+$, we have that
\begin{equation}\label{eq:SpatialTPlusI}
\left| \Pr\left [Z(w)=c | Z(\outBound {\cal T}_{w})=\tau(\outBound {\cal T}_w)  \right] -k^{-1} \right | \leq d^{-10}.
\end{equation}
Note that the above  applies only for $T^+$ and not the whole block $B$ with 
its boundary.

However,  \eqref{eq:SpatialTPlusI} and the observation  that each $w_i$ has exactly 2 neighbors
in  $C$ imply the  following: Let $\sigma, \tau$ two $k$-colorings of $G$,  and let 
$X,  Y$ be two random colorings of $G$. Conditional on that  $X(\outBound B)=\sigma(\outBound B)$
and $Y(\outBound B)=\tau(\outBound B)$  there is a coupling such that the probability
$X(w_i) \neq Y(w_i)$  is less than $3/k$.
To see this, note  that for any color assignment of
$w_{i+1}, w_{i-1}$ (the neighbors of $w_i$) in  $X, Y$ there is always a coupling  
such that $\Pr[X(w_i)\neq Y({w_i})] \leq  3/k$.

Assume that $w_{j+1}$ and $w_{j-1}$ are the neighbors of $w$ in $C$, i.e., $w=w_j$
Using the previous observation and a union bound, there is a coupling such that 
the probability of having either $X(w_{j-1})\neq Y({w_{j-1}})$ or
$X(w_{j+1})\neq Y(w_{j+1})$ is less than 6/k. 

Given the assignments $X(w_{j-1}),  Y(w_{j-1}),  X(w_{j+1}),  Y(w_{j+1})$, 
we have the following: If $X(w_{j-1})=Y(w_{j-1})$ and  $X(w_{j+1})= Y(w_{j+1})$,
then from \eqref{eq:SpatialTPlusI} there is a coupling such that 
the probability of having  $X(w_{j})\neq  Y(w_{j})$ is at most $d^{-8}$.
On the other hand, if $X(w_{j-1})\neq Y(w_{j-1})$, or
$X(w_{j+1})\neq  Y(w_{j+1})$,  then there is a coupling such that
the probability of having  $X(w_{j})\neq  Y({w_{j}})$ is at most $3/k$.
This implies that  there is a coupling such that $X(w_{j})\neq  Y({w_{j}})$
with probability less than $20/k^2$. This completes the proof.

\subsection{Proof of Proposition  \ref{prop:UnBiasHighDegree4Trees} }\label{sec:prop:UnBiasHighDegree4Trees}

Let $\outBound T=T_{x,w}\cap \outBound B$, also, we let $T=T_{x,w}\cup \outBound T$. 
It suffices to show the following: Let $\sigma_1, \sigma_2$ be $k$-colorings of $T$. 
For  random colorings $X, Z$ of the tree $T$ and any color $c\in [k]$,   we have that
\begin{equation}\label{eq:Basis:UnBiasHighDegree4Trees}
\left | \Pr[X(w)=c\ | X(\outBound T)=\sigma_1(\outBound T)]- \Pr[Z(w)=c\ | Z(\outBound T)=\sigma_2(\outBound T)] \right|
\leq d^{-13}.
\end{equation}

\noindent
Let $u_1,  \ldots, u_m$ be an enumeration of the vertices in $\outBound T$, i.e., $m=|\outBound T|$.
Let the sequence of boundary conditions $\tau_0, \ldots, \tau_{m}$ at $\outBound T$.
 For  $i\in [m]$, it holds that $\tau_{i-1}$ and $\tau_{i}$ differ only on the  assignment
of vertex $u_i$, i.e., $\tau_{i-1}(u_i)=\sigma_1(u_i)$ and $\tau_i(u_i)=\sigma_2(u_i)$. 
Triangle inequality implies that
\begin{eqnarray}
\lefteqn {
\left | \Pr[X(w)=c\ | X(\outBound T)=\sigma_1(\outBound T)]- \Pr[Z(w)=c\ | X(\outBound T)=\sigma_2(\outBound T)] \right|
} \hspace{1.6in} \nonumber \\
&\leq & \sum^m_{i=1}
\left | \Pr[X(w)=c\ | X(\outBound T)=\tau_{i-1}]- \Pr[Z(w)=c\ | Z(\outBound T)=\tau_i] \right|. \nonumber 
\end{eqnarray}

\noindent
For each term  $\left | \Pr[X(w)=c\ | X(\outBound T)=\tau_{i-1}]- \Pr[Z(w)=c\ | Z(\outBound T)=\tau_i] \right| $ note that
we have a single disagreement at  $\partial T$.  
For any   coupling of $X, Z$ a path $P\in T$ such that for every $u\in P$ we 
have $X(u)\neq Z(u)$ is called path of disagreement.  Using the 
Disagreement Percolation coupling construction from \cite{DisPerc}  we have the following:
\begin{equation} 
\left | \Pr[X(w)=c\ | X(\outBound T)=\tau_{i-1}]- \Pr[Z(w)=c\ | Z(\outBound T)=\tau_i] \right|
\leq 
\Exp{   \mathbf{1}\{ {\cal P}_i \textrm { is a path of disagreement} \}  }, \qquad  \label{eq:DisPercReduction4BiasColour}
\end{equation}
where the expectation above . is w.r.t. the coupling we  use and
${\cal P}_i$ is the  only path from $u_i$ to $w$.

Since $T$ is a tree ,  whenever  the coupling of $X, Z$ decides the coloring for some vertex $u$, 
the maximum number of disagreements in its neighborhood is at most one. 
Furthermore, for a vertex $u$ whose number of disagreement in the neighborhood is at most 1, 
there is a coupling such that the probability of the event $X(u) \neq Z(u)$ is upper bounded
by the probability of the most likely color for $u$ in the two chains. 
For each vertex $u\in T$, let $\xi(u)$ be the probability of disagreement in the coupling.
Disagreement percolation is dominated by an independent process, that is,
\begin{equation}\label{eq:InfluenceBound4Bias}
\Exp{ \mathbf{1}\{{\cal P}_i \textrm { is a path of disagreement} \}  } \leq \textstyle \prod_{v \in {\cal P}_i} \xi(u).
\end{equation}

\noindent
For every $u\in T$, consider  $p_{u}(0)$, as  defined in \eqref{eq:DefPu}.
We  show that for every $u\in T$ it holds that 
\begin{equation}\label{eq:TargetXiVsPi}
\xi (u)\leq p_{u}(0).
\end{equation}
Before showing that \eqref{eq:TargetXiVsPi} is indeed true, let as show how, using \eqref{eq:TargetXiVsPi},
we get the proposition. 

 Combining  \eqref{eq:DisPercReduction4BiasColour}, \eqref{eq:InfluenceBound4Bias} and \eqref{eq:TargetXiVsPi}
we have that
\begin{eqnarray}
\left | \Pr[X(w)=c\ | X(\outBound T)=\sigma_1(\outBound T)]- \Pr[Z(w)=c\ | X(\outBound T)=\sigma_2(\outBound T)] \right|
& \leq &\sum^m_{i=1} \prod_{v \in {\cal P}_i} p_{u}(0). \qquad \label{eq:TVDReducedToInd}
\end{eqnarray}

\noindent
Consider the  independent process where each vertex $u\in T$ is set with probability 
$p_u(0)$ disagreeing. Let $D(T)$ be the number of paths of disagreement from the root $w$ 
to the vertices which are incident to $\outBound T$. Then, it holds that
\begin{equation}
\Exp{D(T)}  = \textstyle \sum^m_{i=1} \prod_{v \in {\cal P}_i}  p_{u}(0). \label{eq:TVDReducedToExpct} 
\end{equation}
We are going to get an upper bound for the quantities in \eqref{eq:TVDReducedToExpct}. 
Assume first that $w\in \HighDegreeSet$.
Let $D_{\ell}(T)$ denote the number of paths of disagreement from the root $w$
that have length $\ell$. It holds that
\[
\Exp{ D_{\ell}(T) }  = p_w(0)  \textstyle \sum_{y\in N(u)} \Exp{ D_{\ell -1} (T_y) },
\]
where ${T}_y$ is the subtree of ${T}$ rooted at $y$, child of $w$ in $T$.
From the above, we get that
\begin{eqnarray}
\Exp{ D_{\ell}(T) }
&< & p_w(0)  \ \degree_{in}(w^j_r)\ \max_{y\in N(w) } \left\{ \Exp{ D_{\ell-1}(T_y)} \right\}   
\nonumber \\
&\leq& 
\max_{\P'=(u_0=w ,u_1,\dots,u_\ell)}
p_{u_{\ell }}(0)  \prod_{i=0}^{\ell-1} p_{u_i}(0) \times
\left [ \degree_{in}(u_i) \right]. \qquad \label{eq:EDEllTGenBound4Bias}
\end{eqnarray}
Now, recall that  $u_{\ell}\in \ImpVrtx$. Then, weighting schema \eqref{def:PathWeight} implies the
following: Let $M$ be the set of high degree vertices in $\P'$ and let $s=|M|$. Then, 
using Corollary \ref{cor:FromBreakPointProd} and \eqref{eq:EDEllTGenBound4Bias}
we get that 
\begin{eqnarray}
\Exp{ D_{\ell}(T) }  &\leq & 
\max_{\P'=(u_0=w^j_r ,u_1,\dots,u_\ell)} p_{u_{\ell}}(0) \left( \prod_{u_i\notin M} p_{u_i}(0) \times  \degree_{in}(u_i)  \right )
\left( \prod_{u_i \in M}  \degree_{in}(u_i)  \right ) \nonumber \\
&\leq &
\max_{\P'=(u_0=w^j_r ,u_1,\dots,u_\ell)} 
p_{u_{\ell}}(0) \left( \prod_{u_i\notin M} p_{u_i}(0) \times \degree_{in}(u_i) \right  )
\frac{\left( (1+\epsilon/6) \right)^{\ell}}{((1+\epsilon/6)d^{15} )^s}
 \nonumber \\
 &\leq & \textstyle 2 \left( \frac{1+\epsilon/6}{1+\epsilon} \right )^{\ell-s} d^{-15s} 
\ \leq \  2 \left(  {1+2\epsilon/3}  \right)^{-\ell} (d/2)^{-15s}.  \nonumber 
\end{eqnarray}
Note that we used Corollary \ref{cor:FromBreakPointProd} in the second derivation.
The above implies that  
\begin{equation}\label{eq:TVDReducedToExpcIndBoundHighDegree}
\Exp{ D(T)} \leq C' d^{-15},
\end{equation}
for large $C'>0$. Consider $w\notin \HighDegreeSet$ but $N(w)\cap \HighDegreeSet\neq \emptyset$ and
$\bar{w}$ is a high degree neighbour in $N(w)$. Then, since $p_{\bar{w}}=1$, it is direct
to see that the paths of disagreement that reach ${w}$ reach $\bar{w}$, as well. This 
observation, combined with \eqref{eq:TVDReducedToExpcIndBoundHighDegree} implies
that 
\begin{equation}\label{eq:TVDReducedToExpcIndBound}
\Exp{ D(T)} \leq C' d^{-15},
\end{equation}
regardless of weather the root $w\in \HighDegreeSet$ or $N(w)\cap \HighDegreeSet\neq \emptyset$.
Combining \eqref{eq:TVDReducedToInd}, \eqref{eq:TVDReducedToExpct} and
\eqref{eq:TVDReducedToExpcIndBound} we have
\[
\left | \Pr[X(w)=c\ | X(\outBound T)=\sigma_1(\outBound T)]- \Pr[Z(w)=c\ | X(\outBound T)=\sigma_2(\outBound T)] \right|
\leq d^{-14}.
\]

\noindent
It remains to show that \eqref{eq:TargetXiVsPi} is indeed true.
In light of Corollaries  \ref{corr:GMPBiasColour},  \ref{cor:GeneralBiasColour} 
at each step of disagreement percolation which decides on vertex $u$,  where $u$ is such that
$\degree(u)\leq \TDeg$ and $N(u)\cap \HighDegreeSet=\emptyset$,  we get that $\xi(u)\leq p_u(0)$.
Also, for a vertex $u\in \HighDegreeSet$, we trivially have $\xi(u)\leq p_u(0)$, since
for such a vertex $p_u(0)=1$. It remains to consider vertices $u\in T$ such that $u\notin \HighDegreeSet$ and 
$N(u)\cap \HighDegreeSet\neq \emptyset$.

Recall that $\xi(u)$ is the probability
of the most biased color for $u$, in both $X, Y$. 
Consider $T_u$, the subtree of $T$ rooted at $u$, for some $u\in T$
such that $u\notin \HighDegreeSet$ and  $N(u)\cap \HighDegreeSet\neq \emptyset$. 
Also,  consider the independent percolation process where each vertex $v$ is disagreeing 
with probability $\xi(v)$.  We are going to show the following: 
if for every $v\in T_u\setminus \{u\}$  \eqref{eq:TargetXiVsPi} holds, then $\xi(u)\leq p_u(0)$.
Given that, \eqref{eq:TargetXiVsPi} follows by employing a simple induction.

Since \eqref{eq:TargetXiVsPi} holds for every $v\in T_u\setminus \{u\}$, 
with an analysis  similar to what we had before, we get
$
\Exp{ D(T_u)} \leq 2d^{-12}.$
This implies directly  that $\xi(u)\leq k^{-1}+2d^{-12}\leq p_u(0)$, for large $d$.
This completes the proof.

\section{Disagreement Percolation  Results}\label{sec:DisPerc}

\noindent
Given some $\epsilon, \Delta>0$ and sufficiently large $d$, consider $G=(V,E)$ such that 
$G\in \IntrstGraphFam(\epsilon, d, \Delta) $ with set of blocks ${\cal B}$.
Also assume that $k\geq (\alpha+\epsilon)d$. For each vertex $v\in V$ we let $B_v\in {\cal B}$ 
denote  the block in which $v$ belongs. Also, recall that $N=|{\cal B}|$.

\begin{figure}
	\centering
		\includegraphics[height=6.3cm]{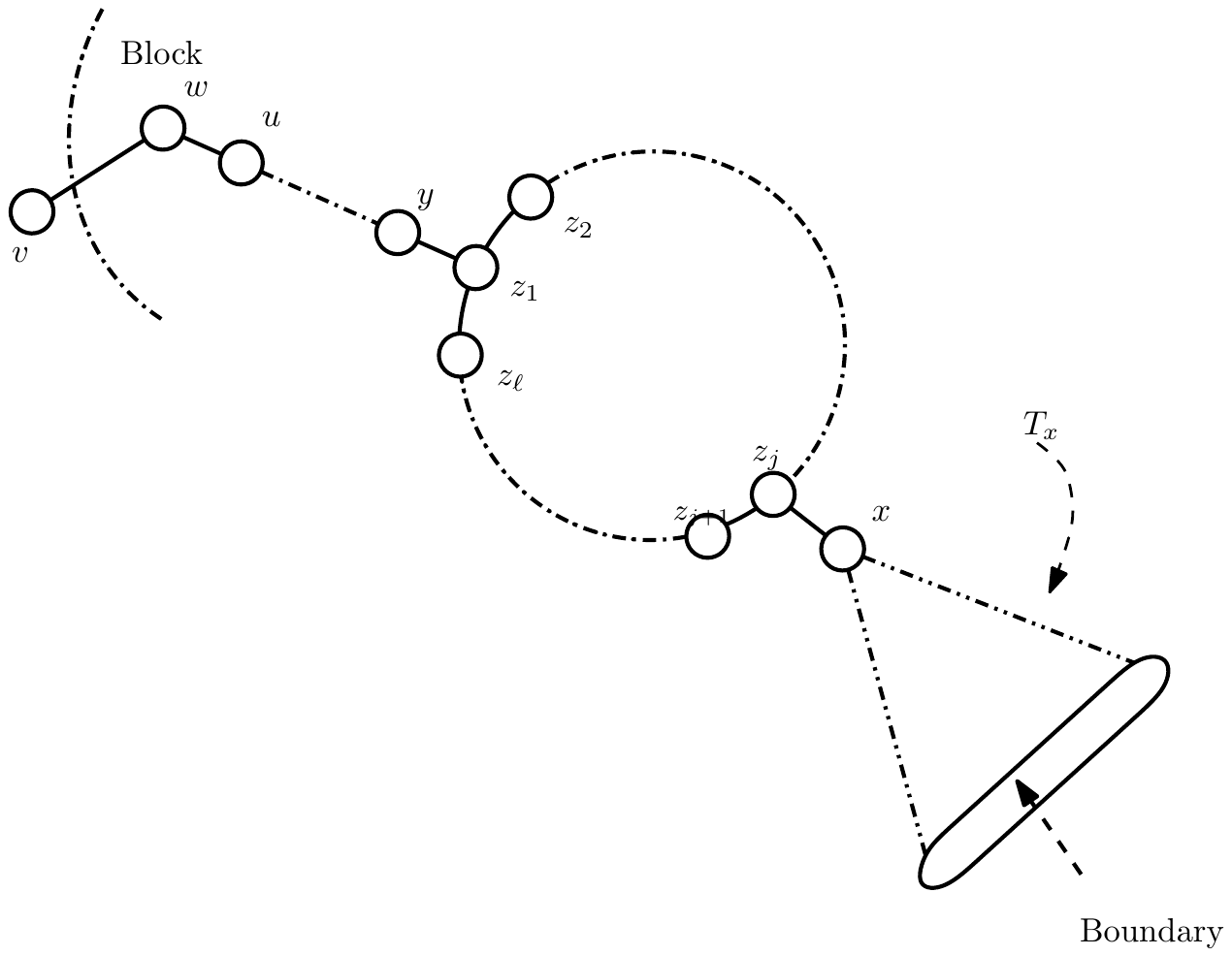}
		\caption{Unicyclic Block}
	\label{fig:UnicyclicBlock}
\end{figure}

Due to our assumptions about  $\cal B$  each $u\in \ImpVrtx$ is either a breakpoint or a vertex adjacent 
to a  breakpoint.
Consider   two copies of the block dynamics $(X_t)_{t\geq 0}$ and $(Y_t)_{t\geq 0}$.
Assume that the two copies of block dynamics are coupled such that at each transition 
the same block is updated in both of them.
In what follows we describe how do we couple the update of a block $B$ in the two chains.
To avoid trivialities, assume that $B$ contains more
than one vertices.
Let   $\Lambda =(X_t \oplus Y_t)\cap \outBound B$ and assume  that at time $t+1$ both 
$(X_t)_{t\geq 0}$ and $(Y_t)_{t\geq 0}$  update block $B$.  
 Our focus is on the set $\WDis_{t+1}\cap B$. 
Recall that for each $t\geq 0$ let  $\WDis_t=X_{t}\oplus Y_t$. 
Also, we have   $\WDis_{\leq t}=\bigcup^t_{s=0} \WDis_s.$

\paragraph{Coupling}
The coupling  decides  $X_{t+1}(B)$  and $Y_{t+1}(B)$ in steps. At each steps it considers
a single vertex $u\in B$ and decides $X_{t+1}(u)$, $Y_{t+1}(u)$ conditional the configurations
at $\outBound B$ and the configurations of the vertices in $B$ that were considered in the coupling
before $u$. The coupling of $X_{t+1}(u)$, $Y_{t+1}(u)$ is {\em maximal}, i.e., minimizes the
probability of the event  $X_{t+1}(u) \neq Y_{t+1}(u)$.

Initially the disagreements are only in $\Lambda\subseteq \outBound B$, but in subsequent steps 
there could also be disagreements inside $B$. The coupling gives priority to vertices which
are next to a disagreement. That is, as long as there are vertices next to a disagreeing vertex 
such that their color is not specified,  the coupling chooses one according to the following rule:

 Consider some, arbitrary, ordering of the vertices in $\Lambda$. E.g. say $u\in \Lambda$ is the first 
 vertex.   The coupling creates a maximal component of disagreeing
vertices around $u$, which we call ${\cal C}_u$. Initially  ${\cal C}_u$ contains only $u$. 
Every time we consider   some arbitrary vertex $w$ which is adjacent to ${\cal C}_u$ and its 
 coloring has not been decided.
The coupling decides  both $X_{t+1}(w)$ and $Y_{t+1}(w)$.
If this vertex ends up being a disagreement it is inserted into  ${\cal C}_u$. Otherwise it is not.
 That is, as we decide the coloring of the vertices of $B$,  ${\cal C}_u$
may  grow. The growth of ${\cal C}_u$ stops when it has no neighbors in $B$ that are 
uncolored.   Then the coupling considers the next vertex in $\Lambda$ in the same manner.

\begin{remark}
For two or more vertices in $\Lambda$, their corresponding components can be identical. E.g. 
let $u, w \in \Lambda$ and ${\cal C}_{u}$ contains  $v$ which is adjacent to $w$. Then,
 ${\cal C}_{u}$ and ${\cal C}_{w}$ are identical. 
\end{remark}

Let $\bar{\psi}=\bar{\psi}_{B, \Lambda}( X_{t}, Y_{t})$ be the distribution over the subset of vertices of $B$,
induced by the disagreeing vertices in the coupling above. That is $\mathbold{\theta}_{\Lambda}$ distributed as in 
$\psi$ contains all the disagreeing vertices from the coupling of $X_{t+1}(B)$ and $Y_{t+1}(B)$.
Note that we have that 
\begin{equation}\label{lemma:StochOrderFromRealToWorstCase}
\left( X_{t+1}(B)\oplus Y_{t+1}(B) \right) \subseteq  \mathbold{\theta}_{\Lambda},
\end{equation}

\noindent
We study the distribution $\bar{\psi}=\bar{\psi}_B( X_{t}, Y_{t})$ by means of  measures
which are easier to analyze. 

For some $\delta>0$,  let ${\cal S}_{\delta}={\cal S}_{\delta}(B)$ be a random subset of the block $B$ such that each vertex 
$v\in B$ appears in ${\cal S}_p$, independently, with probability $p_v(\delta)$ where
\begin{equation}\label{eq:DefPu}
\textstyle p_u (\delta)= 
\begin{cases}
(1+\delta)\min \left \{ \left ( (1+\epsilon) \deg_{in}(u) \right)^{-1}, (k-\degree(u) )^{-1} \right \}  & 
\textrm{if } \degree(v)\leq \TDeg \\ 
1  & \mbox{otherwise}.
\end{cases}
\end{equation}
For  {\em unicyclic } $B$  we have the following: for each $u$ outside the cycle $p_u(\delta)$ 
is the same as above. If $u$ belongs to the cycle, then $p_u(\delta)=1$.

Given ${\cal S}_p$ and $u\in \outBound B$,
let $\mathbold{\theta}_u\subseteq {\cal S}_p$  contain every vertex $w\in B$ such that
there is a path using vertices in $S_p$ that connect $u$ and $w$. 
We let  $\psi_{v}(\delta)=\psi_{v, B}(\delta)$ be the distribution induced by $\mathbold{\theta}_u$.

\begin{proposition}[Stochastic Domination]\label{prop:StochasticDomination}
For all $\epsilon$,  there exist  $d_0$ such that for all $d\geq d_0$,  for $k\geq (\alpha+\epsilon)d$ 
and every  graph  $G\in \IntrstGraphFam(\epsilon, d, \Delta)$, where $\Delta>0$ can depend on $n$,
the following is true:

Consider  some block $B$ and   two $k$-colorings of $G$ $\sigma,\tau$ such that
for   $\Lambda=(\sigma\oplus\tau)\cap B$ and     $|\Lambda| \leq d^{9/10}$.
For $u\in \Lambda$,  let the  {\em independent}  random variables $\mathbold{\theta}_u$,
be distributed as $\psi_u(\epsilon^3)$, respectively.
Let  $\mathbold{\theta}_{\Lambda}$ be distributed as in $\bar{\psi}=\bar{\psi}_{\Lambda, B}(\sigma,\tau)$.

There is a coupling between $\mathbold{\theta}_{\Lambda}$ and  $\cup_{u\in \Lambda} \mathbold{\theta}_u$
such that with probability 1 we have
\[
\mathbold{\theta}_{\Lambda} \subseteq \textstyle \bigcup_{u\in \Lambda} \mathbold{\theta}_u.
\]
\end{proposition}

\noindent
The proof of Proposition \ref{prop:StochasticDomination} appears in Section \ref{sec:prop:StochasticDomination}.

Using the above proposition we get the following useful result.

\begin{lemma}\label{lemma:BigJumpInBurnIn}
For all $\epsilon, \Delta, C>0$, there exist $C', d_0>0$, 
such that for  all $d > d_0$, for $k=(\alpha+\epsilon)d$ 
and every  graph  $G\in \IntrstGraphFam(\epsilon, d, \Delta)$, where $\Delta>0$ can depend on $n$
the following is true:

Consider two copies of block dynamics $(X_t)_{t\geq 0}$ and $(Y_t)_{t\geq 0}$ such that
$| X_0\oplus Y_0  |=S$,  for some integer  $0<S  \leq d^{4/5}$.  Letting 
$r=CN/\log d$,  there is a coupling such that
\[
\Pr\left [  |D_{\leq r}  |  \geq (1+q)S \right] \le C'\exp\left( -qS/C'\right),
\]
for  any $q$ such that $\left (\log d \right)^{-1/2}\leq q$ and $(1+q)S\leq d^{9/10}$.
\end{lemma}

\noindent
The proof of Lemma \ref{lemma:BigJumpInBurnIn} appears in Section \ref{sec:lemma:BigJumpInBurnIn}

\subsection{Proof of Proposition \ref{prop:StochasticDomination}}\label{sec:prop:StochasticDomination}

For the sake of brevity, we let $\delta=\epsilon^2$. Consider,  first, the case where  $B$ has multiple vertices. 
For each $u\in \Lambda$  consider an independent copy of ${\cal S}_{u}$. Each ${\cal S}_{u}$
is a subset of $B$ where each vertex $v$ is included, independently of the other vertices with probability $p_v(\delta)$,
where $p_v$ is defined in \eqref{eq:DefPu}.  Then we define each $\mathbold{\theta}_u$ w.r.t. ${\cal S}_u$.

In the coupling we reveal the vertices in $\mathbold{\theta}_{\Lambda}$ in the same order as we consider them in
the coupling in Section \ref{sec:DisPerc}, i.e, we gave priority to vertices next to disagreements. The disagreeing
vertices are the vertices which are already inside $\mathbold{\Lambda}$ and those which are not, are non disagreeing.
That is we couple $\mathbold{\theta}_{\Lambda}$ and $\cup_{u\in \Lambda}\mathbold{\theta}_u$ in steps.

At $i$-th step assume that we deal with vertex $w_i\in B$, while we have revealed  $\theta^i $ from $\mathbold{\theta}_{\Lambda}$
and  $\theta^i_u$, from $\mathbold{\theta}_u$   where $u\in \Lambda$.
It suffices to show that for every $i\geq 1$,  we have that
$
 \theta^{i-1}  \subseteq \bigcup_{u\in \Lambda}\theta^{i-1}_u$, 
{while  there is $\Lambda' \subseteq \Lambda$ such that  the probability that $w_i\in \theta^i$ is upper bounded by the
probability $w_i\in \bigcup_{u\in \Lambda'}\theta^i_u$. }  Note that    $\theta^0=\bigcup_{u}\theta^0_u=\Lambda$.

Let ${\cal M}_i$ be the set of paths of unrevealed  vertices in $B$,  from $w_i$ to the components of $\theta^{i-1}$.
Note that $\mathbold{\theta}_{\Lambda}$ may have more than one components.
We have the following results.

\begin{claim}\label{claim:DisPathProbs}
For any integer $i\geq 1$,  
If $\theta^{i-1}  \subseteq \bigcup_{u\in \Lambda}\theta^{i-1}_u$ holds, then 
\[
\Pr[w_i\in \theta^i]\leq \textstyle \sum_{ P \in  {\cal M}_i} \prod_{v \in P} p_{w}(0).
\]
\end{claim}

\noindent
The proof of Claim \ref{claim:DisPathProbs} appears after this proof.

\begin{claim}\label{claim:Degrees4CloseDis}
For  integer $i \geq 1$, assume that
$w_i$, at step $i$ of the coupling,   is within distance two from  at least two disagreements.  Then the following is true:

 If $w_i$ does not belong to a cycle inside $B$,  then, for every $v \in \ball(w_i,4)$ it holds that $\degree(v)\leq \TDeg$.
If   $w_i$  belongs to a cycle inside $B$, then there can be at most 2 paths in ${\cal M}_i$
of length $1$.
\end{claim}

\noindent
Claim \ref{claim:Degrees4CloseDis} follows easily from the definition of the set of blocks ${\cal B}$ and the way we
have defined the coupling for the update of block $B$, in Section \ref{sec:DisPerc}. For this reason we omit 
this proof.

For $i\geq 1$, let the event ${\cal A}_i:=$ ``$\theta^{i-1} \subseteq \bigcup_{u}\theta^{i-1}_u $.
We show that for every $i\geq 1$,  we have that
\begin{equation}\label{eq:Target4StochasticDomReal}
\Pr[ w\in \theta^i  \ | \ {\cal A}_i ] \leq \Pr[\cup_u  \left (w_i\in \theta^i_u  \right)\ | \ {\cal A}_i].
\end{equation}

\noindent
First we assume that  $B$ is a tree.  Let $q=|N(w_i)\cap \theta^{i-1} |$, i.e., $q$ is the number of  disagreement  right next to $w_i$ at step $i$.
We consider the following cases regarding $q$:  $q=1$,  $q>1$ and   $q=0$.

\paragraph{Case : $q=1$.}
Assume that $w_i$ is right next to $v \in N(w_i)\cap \theta^{i-1}$.
Furthermore, conditioning on the event ${\cal A}_i$ implies that 
there is a  non empty $\Lambda'\subseteq \Lambda$ such that for every $u\in \Lambda'$
we have    $v\in \theta^{i-1}_u $.

We consider two cases regarding the degree of $w_i$.
The first is $\degree(w_i)> \TDeg$ and the second is $\degree(w_i)<\TDeg$.
The first case is trivial since, by definition we have 
 $\Pr[\forall u\in \Lambda'\ w_i\in \theta^i_{u}  ]=1$.

We proceed with the case    $\degree(w_i)\leq \TDeg$.   
Note that $\Pr[w_i \in \theta^i \ | \ {\cal A}_i]$ is maximized when  there are $|\Lambda|-1$ disagreements at distance 2 from $w_i$, 
let $\bar{N} \subseteq N(w_i)$ contain the neighbors of $w_i$ which are  are adjacent to these  disagreements.
Letting $p_{\max}= \max_{z\in \bar{N}} \{ p_{z}(0)\}$,   Claims \ref{claim:DisPathProbs} implies that
\begin{eqnarray}
\Pr[w_i \in \theta^i \ | \ {\cal A}_i] & \leq  & p_w(0) + p_w(0) \sum_{z\in \bar{N}}p_{z}(0) 
\nonumber \\
 & \leq  & \textstyle 
\left ( 1 + |\bar{N}|\ p_{\max}  \right)p_w(0) 
\ \leq \ (1+d^{-1/12}) \  p_w(0), \label{eq:QEq1LowPTau} 
\end{eqnarray}
For \eqref{eq:QEq1LowPTau}  we use that  for any $z\in \bar{N}$
we have that $p_z(0) \leq Cd^{-1}$ and   $|\bar{N}|<|\Lambda| \leq  d^{9/10}$.
As far as $\theta^{i}_u$s are regarded, we have the following:
\begin{eqnarray}\label{eq:QEq1LowPSigma} 
\Pr[\cup_{u\in \Lambda'} \left ( w_i\in  \theta^{i}_{u} \right) \ | \ {\cal A}_i] &\geq &p_w(\delta) 
\ = \ (1+\epsilon^3)\ p_w(0),
\end{eqnarray}
where second derivation holds because $|\Lambda'|\geq 1$.
Eq.  \eqref{eq:Target4StochasticDomReal} follows from \eqref{eq:QEq1LowPSigma} and \eqref{eq:QEq1LowPTau}.

\paragraph{Case: $1<q\leq |\Lambda |$.} Due to Claim \ref{claim:Degrees4CloseDis} we have that
if $q>1$, then  $\degree(w_i)\leq \TDeg$. Furthermore, conditioning on ${\cal A}_i$, implies that 
there is a  non empty $\Lambda'\subseteq \Lambda$ such that $|\Lambda'|\geq q$, while
for every $u\in \Lambda'$ we have that   $v \in \theta^{i-1}_u $, where $v\in N(w_i)\cap \theta^{i-1}$.

Note that   ${\cal M}_i$ contains at most $|\Lambda|-q$ paths of length greater than 1.
This fact implies that the probability of having $w_i\in \theta^i$  is maximized
by assuming that there are $|\Lambda|-q$ disagreements at distance 2 from $w_i$.
Let $\bar{N}\subseteq N(w_i)$  contain the neighbors of $w_i$ that are adjacent to these disagreements.
Claim \ref{claim:DisPathProbs} implies that
\begin{eqnarray}
\Pr[w_i\in \theta^i  \ | \ {\cal A}_i] & \leq  & q\ p_{w_i}(0)+ p_w(0)\sum_{z\in \bar{N} }p_{z}(0)\nonumber \\
 & \leq  & q \ p_{w_i}(0) \textstyle 
\left ( 1 +   \sum_{z\in \bar{N}}p_{z}(0) \right)  \hspace*{1cm}\mbox{[since $q>1$]}\nonumber \\
 & \leq  & q \ p_{w_i}(0) \textstyle 
\left ( 1 +  d^{-1/10}\right). \label{eq:QEq1HighPTau} 
\end{eqnarray}
where the last derivation follows with the same arguments as those we use for \eqref{eq:QEq1LowPTau}.

 Applying inclusion-exclusion we have
\begin{eqnarray}
\Pr[\cup_{u\in \Lambda'} \left ( w_i\in  \theta^{i}_{u} \right) \ | \ {\cal A}_i]
&\geq & |\Lambda'| \  p_{w_i}(\delta)-{|\Lambda'| \choose 2}\left( p_{w_i}(\delta)\right)^2 \nonumber \\
&=& |\Lambda'| \  p_{w_i}(\delta)\left(1-(|\Lambda'|-1) p_{w_i}(\delta)/2\right)   \nonumber \\
&\geq& \textstyle (1+\delta)q \  p_{w_i}(0)\left(1-(Cd^{1/10})^{-1}\right), \label{eq:QEq1HighPSigma}
\end{eqnarray}
for large $C>0$.
For  the last derivation we use the following facts: it holds that  $q\leq |\Lambda'|\leq  d^{9/10}$.
Furthermore, according to Claim \ref{claim:Degrees4CloseDis},  we have $\degree(w_i)\leq \TDeg$. For such vertex it is easy to show that there exist appropriate constance $C>0$
such that  $p_{w_i}(\delta)\geq (Cd)^{-1}$.

Choosing  any fixed $\delta>0$ and large $d>0$, from  \eqref{eq:QEq1HighPSigma} and \eqref{eq:QEq1HighPTau},
we get that \eqref{eq:Target4StochasticDomReal} is indeed true.

\paragraph{Case: $q=0$.} This case  is  straightforward. Due to the way we define
the coupling, once  $q=0$ we have that $\Pr[ w_i\in \theta^i \ | \ {\cal A}_i ]=0$, whereas 
$\Pr[\cup_u  \left (w_i\in \theta^i_u \right)\ | \ {\cal A}_i]\geq 0$.
Then \eqref{eq:Target4StochasticDomReal} is indeed true.
\\ \vspace{-.2cm}

\noindent
Now consider the case where $B$ is {\em unicyclic}.  In such block the cycles are hidden away from $\outBound B$.
The order we consider the vertices in the coupling ensures  that we can only have more than 2 disagreements
around a vertex only when the vertex is close to the boundary. Recall that close to the boundary there are only
vertices of degree at most $\TDeg$.

For the case where $B$ is unicyclic we work as in   the case where $B$ is a tree. 
The cases where $w_i$ is not in the cycle of $B$
is identical to the previous, i.e., when  $B$ is a tree.  The case where $w_i$ belongs to the cycle 
of $B$ follows trivially because  in $p_{w_i}(\delta)=1$ for such a vertex.

We conclude with the single vertex block. This  is identical to the
case where $B$ is a tree and $q=|\Lambda|$. 
The proposition follows.
$\hfill \Box$

\begin{proof}[Proof of Claim \ref{claim:DisPathProbs}]

For the sake of simplicity consider a 
$X,Y$ be two random colorings of $B\cup \outBound B$ and
$X(\outBound)\oplus Y(\outBound)=\Lambda$.
Assume that $X(B)$ and $Y(B)$ are coupled as specified in Section \ref{sec:DisPerc}.

We reveal $\mathbold{\theta}$ in steps, as we reveal the configuration of $X(B)$ and $Y(B)$ in the coupling.
Assume that step $i$ we reveal vertex $w_i$ and let $\theta^{i}$ be the configuration of $\mathbold{\theta}_{\Lambda}$
we have revealed.

Let $B_i\subseteq B$ be the set of vertices whose coloring has not been specified at step $i$.
 Let $\partial B_i\subset B $ contain the vertices  in $B$ whose coloring has 
been  decided by step $i$  and they are next to a vertex whose color has not been specified. 
The claim follows by showing that 
\begin{equation}\label{eq:Target4claim:DisPathProbs}
\ExpCond{ \mathbf{1}\{X(w_i) \neq Y(w_i)\} }{ X(\outBound B_i), \ 
Y(\outBound  B_i) } \leq \sum_{ P \in  {\cal M}_i} \prod_{v \in P} p_{w}(0).
\end{equation}

Let ${\tt Dis} \subseteq  \partial  B_i $ contain every vertex $u$ such $X(u)\neq Y(u)$.
Clearly,  ${\cal M}_i$ be the set of paths in $B$ from $w_i$ to some vertex in ${\tt Dis}$ such that all
but the last vertex in the path belongs to $B_i$.

Let $x_1, x_2, \ldots, x_m$ be some arbitrary ordering of the vertices in ${\tt Dis}$. 
Let $\tau_0, \tau_2, \ldots \tau_m$ be colorings of $\outBound B_i$ such that 
$\tau_0=X(\outBound B_i),\tau_m=Y(\outBound B_i)$, while
$\tau_{j}$ and $\tau_{j+1}$ differ only on the assignment of vertex $x_i$. In particular, 
$\tau_{j-1}(x_i)=X(x_j)$ and $\tau_{j}(x_j)=Y(x_j)$.

For $j=0, \ldots m-1$, consider a coupling of $W_j, W_{j+1}$, two random colorings of 
$B_i\cap \outBound B_i$ such that $W_j(\outBound B_i)=\tau_i$ and $W_{j+1}(\outBound B_i)=\tau_{j+1}$.
Note that for each $W_j, W_{j+1}$, the boundary conditions differ on the assignment of  exactly one vertex. 
It holds that
\begin{eqnarray}\label{eq:ReductionToSeperateInf}
\lefteqn{
\ExpCond{ \mathbf{1}\{X(w_i) \neq Y(w_i)\} }{  X(\outBound B_i), \ Y(\outBound  B_i)} 
} \hspace{.6in} \nonumber \\
&\leq& \sum^{m-1}_{j=0}
\ExpCond{ \mathbf{1}\{W_j(w_i) \neq W_{j+1}(w_i)\} }{ W_j(\outBound B_i)=\tau_j, \ 
W_{j+1}(\outBound  B_i)=\tau_{j+1}}. 
\end{eqnarray}

\noindent 
 Let ${\cal M}_{i,j}$ be the set of paths in $B_i$ that connect $x_j$ to $w_i$. 
In the coupling of $W_{j}, W_{j+1}$ a path $P$  such that   $W_{j}(u)\neq W_{j+1}(u)$ 
for  every $u\in P$, is called path of disagreement.  
It holds that
\begin{eqnarray} 
\lefteqn{
\ExpCond{ \mathbf{1}\{W_j(w_i) \neq W_{j+1}(w_i)\} }{  W_j(\outBound B_i)=\tau_j, \  W_{j+1}(\outBound  B_i)=\tau_{j+1}}
} \hspace{2.4in} \nonumber \\
&\leq &  \sum_{  P\in {\cal M}_{i,j} }
\Exp{   \mathbf{1}\{P \textrm { is a path of disagreement} \}  }.  \label{eq:DisPercReduction}
\end{eqnarray}

\noindent
Recall that $B_i\subseteq B$ is a tree with at most one extra edge.  This implies that whenever  the coupling
of $W_j, W_{j+1}$ decides the coloring for some vertex $u$, if $u$ does not belong to a cycle, the maximum
number of disagreements in its neighborhood is at most one. If $u$ is on the cycle the maximum number of disagreements
in its neighborhood is at most 2. 

Furthermore, for a vertex $u$ whose number of disagreement in the neighborhood is at most 1, 
there is a coupling such that the probability of the event $W_j(u) \neq W_{j+1}(u)$ is upper bounded
by the probability of the most likely for $u$ in the two chains. 
In light of Corollaries \ref{corr:GMPBiasColour} and  \ref{cor:GeneralBiasColour} 
at each step of disagreement percolation which decides on vertex $u$, the probability of having a
 new disagreement is at most $p_{u}(0)$, as  defined in \eqref{eq:DefPu}.  
For  each $P\in {\cal M}_{i,j}$ we have that
\begin{equation}\label{eq:InfluenceBound}
\Exp{ \mathbf{1}\{P \textrm { is a path of disagreement} \}  } \leq \textstyle \prod_{v \in P} p_{w}(0).
\end{equation}

\noindent
The claim follows by combining  \eqref{eq:ReductionToSeperateInf}, \eqref{eq:DisPercReduction}, 
\eqref{eq:InfluenceBound} and get
\begin{eqnarray}
\ExpCond{ \mathbf{1}\{X(w_i) \neq Y(w_i)\}} { X(\outBound B_i), \ 
Y(\outBound  B_i)}  
&\leq& \sum^{m-1}_{j=0} \sum_{  P\in {\cal M}_{i,j} } \prod_{v \in P} p_{w}(0) 
\ \leq \ \sum_{  P\in {\cal M}_{i} } \prod_{v \in P} p_{w}(0). \nonumber
\end{eqnarray}
\end{proof}

\subsubsection{Proof of Lemma \ref{lemma:BigJumpInBurnIn}}\label{sec:lemma:BigJumpInBurnIn}

The proof of Lemma \ref{lemma:BigJumpInBurnIn} makes use of the concepts and results from Section \ref{sec:DisPerc}.

Consider the evolution of the maximally coupled $(X_t)_{t\geq 0}$ and $(Y_t)_{t\geq 0}$ 
such that $|X_0\oplus Y_0|=S$.  Assume  that we couple the two chains from time $0$  up to time 
$T=\min\{t', CN/\log d\}$, where $t'$ is the random time at which  we first have 
$|(X_{t'} \oplus Y_{t'})\cap \ImpVrtx |\geq (1+q)S$.

Letting $A_T=|D_{\leq T} \cap \ImpVrtx|$, it is direct to show that
\begin{equation}\label{eq:Basis4BigJump}
\Pr\left [ \exists t \in [0, CN/\log d]  \ \textrm{s.t. } \  |(X_t \oplus Y_t ) \cap \ImpVrtx |  \geq (1+q)S \right]  \leq
\Pr\left [ A_T  \geq (1+q)S \right].
\end{equation}

\noindent
The lemma will follow by bounding appropriately the r.h.s. of \eqref{eq:Basis4BigJump}.

Each time a block next to a disagreeing vertex is  updated, we say that we have a {\em disagreement
update.}  Since different disagreeing vertices can be adjacent to the same block, we can have
multiple disagreement updates with a single  block  update. 
Let $W$ be the number of disagreement updates up to time $T$. The number of new disagreements generated by  $W$ 
disagreement updates can be dealt by considering $W$  independent processes,  as implied by  
Proposition \ref{prop:StochasticDomination}. For $j=1,\ldots, W$, assume that the disagreement update
influences the block $B_j$ and involves vertex $w_j\in \outBound B_i$.   For each $w_i$, we define $\mathbold{\theta}_i$, 
distributed as  $\psi_{w_i}(\epsilon^3)$, independent with each the other.
Finally, for each $\mathbold{\theta}_i$ let  $\mathbold{\zeta}_i=\mathbold{\theta}_i \cap \ImpVrtx$.

Proposition \ref{prop:StochasticDomination} implies the following: 
for $m=7C(1+q)S\frac{\TDeg}{\log d}$, we have
\begin{eqnarray}
\Pr\left [ A_T  \geq (1+q)S \right]  & \leq & \Pr\left[ \sum_{j \in [W] } |\mathbold{\zeta}_j| \geq qS\right] 
\ \leq \ \Pr\left [\sum_{j\in [W]}| \mathbold{\zeta}_j| \geq qS \ | \ W \leq  m \right ]  + 
\Pr\left [W> m  \right]. \qquad \label{eq:4DefZetaj}
\end{eqnarray}

\noindent
So as to bound the second probability term in the r.h.s. of \eqref{eq:4DefZetaj} we use the following result.

\begin{claim}\label{lemma:Tail4DisUpdates}
In the setting of Lemma \ref{lemma:BigJumpInBurnIn},   and for $m=7C(1+q)S\frac{\TDeg}{\log d}$, we have that 
\[
\Pr[W> m] \leq \exp\left( - d/(2 \log d)\right).
\]
\end{claim}

\begin{proof}
Let ${\tt Dis} \subseteq \ImpVrtx$ be the set of vertices which become disagreeing at least once during the time interval $[0, T]$, 
 For each $u\in {\tt Dis}$,  let  $W_u$ be the number of adjacent blocks that are updated up to time $T$.Note that $W_u$ does not
 consider whether $u$ is disagreeing when a neighboring block is updates.
This implies that $\textstyle W\leq \sum_{u\in {\tt Dis}}W_u.$
It turn, we get that
\begin{equation}\label{eq:WVSWus}
\Pr[W> m ] \leq \Pr[\exists u\in {\tt Dis} \ \textrm{s.t.} \ W_u\geq 7 C\TDeg/\log d ].
\end{equation}
Note that the above holds since we always have $|{\tt Dis}|\leq (1+q)S\leq d$.

Each vertex $u\in {\tt Dis}$ has degree at most $\TDeg$.  That is, there are at most $\TDeg$ blocks that
are neighboring to $u$. At each step we have a neighboring block  updated with probability, at most, $\TDeg/N$. Since $T\leq CN/\log d$, 
we have that $W_u$ is dominated by ${\tt Binomial}(CN/\log d, \TDeg/N)$. Using this observation and 
Chernoff bounds we get that 
\begin{equation}\label{eq:TailsForSingeDisagreementUpdate}
\Pr[W_u\geq 7 C\TDeg/\log d ] \le \exp\left( -7{ C\TDeg}/{\log d}\right).
\end{equation}
A   simple union bound over $u\in {\tt Dis}$, and \eqref{eq:TailsForSingeDisagreementUpdate}  implies the following 
\begin{equation}\label{eq:WuDeviateMuch}
\Pr[\exists u\in {\tt Dis} \ \textrm{s.t.} \ W_u\geq 7 C\TDeg/\log d ] \le \TDeg \exp\left( -{7 C\TDeg}/{\log d}\right).
\end{equation}
For the above inequality we also use the observation that  $|{\tt Dis} |\leq (1+q)S\leq \TDeg$. 
The claim follows by plugging  \eqref{eq:WuDeviateMuch} into  \eqref{eq:WVSWus}
and recalling that $d<\TDeg <2d$.
\end{proof}

\noindent
In light of Claim \ref{lemma:Tail4DisUpdates}, it suffices to show that 
\begin{equation}\label{eq:Target4:lemma:BigJumpInBurnIn}
\textstyle \Pr\left [\left . \sum^W_{j=1}|\mathbold{\zeta}_j |  \geq q S \  \right | \  W \leq m   \right] \ \leq \ \textstyle \Pr\left [  \sum^m_{j=1}|\mathbold{\zeta}_j |  \geq q S   \right] \ \leq \ 
{ \textstyle \exp\left( -3qS/(2C) \right)}.
\end{equation}
In the second inequality we may assume $m$, worst case, disagreement updates. 
It holds that
\begin{equation}
\Pr\left [  \sum^m_{j=1}|\mathbold{\zeta}_j |  \geq q S   \right]  \  \leq \   \sum_{ \substack{ (\alpha_1, \ldots, \alpha_m)\\ \sum_{i}\alpha_i=qS\\ \alpha_i\geq 0}}
\prod^m_{j=1} \Pr[ | \zeta_{j} |  \geq \alpha_j]  \  \leq  \ 
\sum_{ \substack{ (\alpha_1, \ldots, \alpha_m )\\ \sum_{i}\alpha_i=qS\\ \alpha_i\geq 0}}
 \prod_{\substack{ j \in [m ] \\ j: \ \alpha_j\neq 0}  } \Pr\left [ | \zeta_{j} | \geq \alpha_j \right ] \nonumber \\ 
 \end{equation}
 Then Proposition \ref{prop:SingleSourceBlockdis} implies that
 \begin{eqnarray}
\Pr\left [  \sum^m_{j=1}|\mathbold{\zeta}_j |  \geq q S   \right] &\leq &
\sum_{ \substack{ (\alpha_1, \ldots, \alpha_m)\\ \sum_{i}\alpha_i=qS\\ \alpha_i\geq 0}}
 \prod_{\substack{ j \in [m] \\ j: \ \alpha_j\neq 0}  } Cd^{-1} \exp\left(-\alpha_j /C \right)
\nonumber \\ 
&\leq & \exp\left(-2qS/C \right)
\sum_{ \substack{ (\alpha_1, \ldots, \alpha_m)\\ \sum_{i}\alpha_i=qS\\ \alpha_i\geq 0}}
 \prod_{\substack{ j \in [m] \\ j: \ \alpha_j\neq 0}  } Cd^{-1} 
 \hspace*{1.75cm}\mbox{[since $\sum_i\alpha_i=qS$]}\nonumber \\ 
&\leq &  \exp\left( -2qS/C \right) \sum^m_{r=1} {m \choose r} {qS-1 \choose r-1} (C/d)^{r}. \label{eq:TailBoundSumParts}
\end{eqnarray}

\begin{claim}\label{claim:EntropyFromPartitions}
Set $\ell=qS$.  It holds that
\begin{equation}\label{eq:claim:EntropyFromPartitions}
 \sum^m_{r=1} {m \choose r} {\ell-1 \choose r-1} (C/d)^{r} \leq \frac{m^2C}{d} \exp\left( \sqrt{1+\frac{4eC\ell m}{d}}  \right).
\end{equation}
\end{claim}
\begin{proof}
Applying Stirling's approximation for factorial $s!\geq \sqrt{2\pi s}(s/e)^s$, we have that
\begin{eqnarray}
\sum^m_{r=1} { m \choose r} {\ell-1 \choose r-1} (C/d)^{r} &=&\sum^{m}_{r=1} \frac{1}{(r-1)!r!}\left[ \frac{mC(\ell-1)}{d}\right]^{r-1}\frac{mC}{d}
\nonumber \\
&\leq &\sum^{m-1}_{j =0} \frac{1}{j! (j+1)!}\left[ \frac{mC(\ell-1)}{d}\right]^{j}\frac{mC}{d} \qquad\qquad\mbox{[we set $j=r-1$]}
\nonumber \\
&\leq &\sum^{m-1}_{j =0} \frac{1}{2\pi \sqrt{j (j+1)}}
\left[ \frac{mC(\ell-1)e^2}{j (j+1)d}\right]^{j} \frac{mC e}{d}. \label{eq:Target4ClainEntropyFromPartitions}
\end{eqnarray}
Consider the function $f(x)=(A/(x(x+1)))^x$, for real $x>0$ and $A\gg 1$. Direct calculations imply  that 
$f'(x)=f(x)\left( \log\left(\frac{A}{x(x+1)}\right)-\frac{2x+1}{x+1}\right)$. Since $f(x)>0$, for any $x>0$,
and the fact that $\log\left(\frac{A}{(x(x+1)}\right)$ is monotonically decreasing and $\frac{2x+1}{x+1}$
is monotonically increasing, imply that the equation $f'(x)=0$ has at most one solution. 
In particular, it has one solution  $x_0$ which satisfies 
\begin{equation}\label{eq:GeneralX0VsA}
\frac{A}{x_0(x_0+1)}=\textstyle \exp\left( 2\left (1-\frac{1}{2(x_0+1)} \right)\right).
\end{equation}
Noting that the above implies that $e<\frac{A}{x_0(x_0+1)}$,
elementary calculations yield 
$x_0 \leq \frac{-1+\sqrt{1+4A/e}}{2}$.
From the definition of $f(x)$ and \eqref{eq:GeneralX0VsA} we have that
\begin{eqnarray}
f(x_0)&=&  \textstyle \exp\left( 2x_0\left (1-\frac{1}{2(x_0+1)} \right)\right) 
\ \leq \ 
\exp\left( 2x_0\right) \ \leq \ 
 \exp\left(-1+ \sqrt{1+4A/e} \right),  \label{eq:UpperBound4MaxF}
\end{eqnarray}
where in the first inequality we use that $x_0>0$.
Substituting $A$ with $\frac{MC(qS-1)e^2}{d}$, then \eqref{eq:UpperBound4MaxF}  implies that 
for any integer $j>0$ we have that
\begin{equation}\label{eq:BoundForMSOverD}
\left[ \frac{mC(\ell-1)e^2}{j (j+1)d}\right]^{j} \leq \exp\left(-1+ \sqrt{1+4mC(\ell-1)e/d} \right).
\end{equation}
Plugging \eqref{eq:BoundForMSOverD} into \eqref{eq:Target4ClainEntropyFromPartitions} 
we get \eqref{eq:claim:EntropyFromPartitions}. The claim follows.
\end{proof}

\noindent
Recall that $m=7C(1+q)S\frac{\TDeg}{\log d}$.
For any  $S \leq d^{4/5}$ and any $q>(\log d)^{-1/2}$  it holds that $\lim_{d\to \infty}\frac{\sqrt{qSM/d}}{qS}=0$.
Combining  this observation, Claim \ref{claim:EntropyFromPartitions}, 
and   \eqref{eq:TailBoundSumParts} we get \eqref{eq:Target4:lemma:BigJumpInBurnIn}.
The lemma follows.

\section{Proof of Lemma \ref{lemma:Combined-BurnIn}}\label{sec:lemma:Combined-BurnIn}

Lemma \ref{lemma:Combined-BurnIn} follows as a corollary from the following two results.

\begin{lemma}\label{lemma:RateOfDisSpreadBurnIn}
For all $\epsilon, \Delta, C>0$, there exist $C', d_0>0$, 
such that for  all $d > d_0$, for $k\geq (\alpha+\epsilon)d$ 
and every  graph  $G\in \IntrstGraphFam(\epsilon, d, \Delta)$, where $\Delta>0$ can depend on $n$
 the following is true:

Let $(X_t)_{t\ge 0}$ and $(Y_t)_{t\ge 0}$ be   two copies of block dynamics  such that
$X_0\oplus Y_0=\{u^*\}$, for some vertex $u^*$. There is a coupling such that for any  $
1 \leq \ell<d^{4/5}$, we have 
\begin{equation}\nonumber
\Pr \left [|D_{ \leq CN} | \geq \ell \right ]\leq \textstyle C'\exp\left( -\ell^{\frac{99}{100}} C' \right).
\end{equation}
\end{lemma}

\begin{proof} 
Let the time interval ${\cal I}=[0, T]$, where $T=CN$. Consider the partition of ${\cal I}$ into 
$\log d$ time intervals ${\cal I}_1, \ldots, {\cal I}_{\log d}$ such that  $|{\cal I}_j|=\left \lceil |{\cal I}|/\log d \right \rceil$
(the last interval  can be smaller).
We let $t_j$ be the first time step in ${\cal I}_j$, e.g., ${\cal I}_j=[t_j, \ldots, t_{j+1}-1]$.
Also, we fix some small number $0<\gamma<10^{-3}$, independent of $d$.

Let $j'$ be the minimal $j\in [1, \ldots, \log d]$ such that  $|D_{\leq  t_{j'}} | >\ell^{1-\gamma}$.
That is, for any $j<j'$ we have $|D_{\leq  t_{j}} |  \leq \ell^{1-\gamma}$.
Let $\hat{C}>0$ be a large and  let ${\cal A}$ be the event that $|D_{\leq  t_{j'}}| \geq \hat{C} \ell^{1-\gamma}$.  
It holds that
\begin{equation}\label{eq:basis4lemma:RateOfDisSpreadBurnIn}
\Pr \left [|D_{ \leq CN} | \geq \ell \right ] \leq \Pr[{\cal A}]+
\Pr \left [|D_{ \leq CN} | \geq \ell \ | \  {\cal A}^c\right ].
\end{equation}
First consider $\Pr[{\cal A}]$. If $|D_{\leq  t_{j'}} | \geq \hat{C} \ell^{1-\gamma}$
and  $|D_{\leq  t_{j'-1}} |  \leq  \ell^{1-\gamma}$, then during the interval ${\cal I}_{j-1}$
there was a ``big jump" on the number of disagreements in $\ImpVrtx$. That is, more than
$(\hat{C}-1) \ell^{1-\gamma}$ new disagreement where created.  From  Lemma \ref{lemma:BigJumpInBurnIn}
we get that such a jump only occurs with probability at most  $C_1\exp\left(- \ell^{1-\gamma}/C_0 \right)$,
for large constants constant $C_0, C_1>0$. This implies that
\begin{equation}\label{eq:PrABoundRateOfDisSpreadBurnIn}
\Pr[{\cal A}] \leq C_1\exp\left(- \ell^{1-\gamma} C_0 \right).
\end{equation}

\noindent
Assuming   that $|D_{\leq  t_{j'}} |  < \hat{C} \ell^{1-\gamma}$, 
so as to have $|D_{ \leq CN} | \geq \ell$,   there should be at least one $j\geq j'$ such that during 
the interval  ${\cal I}_{j}$   the number of disagreements increased by a factor, more than,
$(1+\gamma/2)$.  From  Lemma \ref{lemma:BigJumpInBurnIn}
we have that such a jump occurs with probability at most $C_2\exp\left(- \ell^{1-\gamma}/C_3 \right)$,
for appropriate constants $C_2,C_3>0$. This implies that
\begin{equation}\label{eq:PrCondOnACBoundRateOfDisSpreadBurnIn}
\Pr \left [|D_{ \leq CN}| \geq \ell \ | \  {\cal A}^c\right ] \leq C_2\exp\left(- \ell^{1-\gamma} C_3 \right).
\end{equation}
The lemma follows by plugging \eqref{eq:PrCondOnACBoundRateOfDisSpreadBurnIn} and \eqref{eq:PrABoundRateOfDisSpreadBurnIn} into 
\eqref{eq:basis4lemma:RateOfDisSpreadBurnIn}.
\end{proof}

\begin{proposition}\label{thrm:BallResult}
In the same setting as Theorem \ref{thrm:RapidMixingAvrgDegGraph176} the following is true:

Let $(X_t)_{t\geq 0}$ and $(Y_t)_{t\geq 0}$ be two copies of block dynamics. Assume that $X_0\oplus Y_0=\{u^*\}$, for some vertex $u^*$. 
For $r=\left \lfloor \epsilon^{-3} (\log d)\sqrt{d}  \right  \rfloor $ It holds that
\[
\Pr \left [ \left( D_{\leq CN}  \right) \not\subseteq 
\ball\left (u^*, r \right) \right] \leq \textstyle 2\exp\left( -d^{0.49}/C \right).
\]
\end{proposition}

\noindent
For the full proof of  the proposition see in Section \ref{sec:thrm:BallResult}.

\subsection{Proof of Proposition \ref{thrm:BallResult}}\label{sec:thrm:BallResult}

Recall that for each time $t\geq 0$ let  $\Dis_t=(X_{t}\oplus Y_t)\cap \ImpVrtx$.  Also, we let 
$\Dis_{\leq t}=\bigcup^t_{s=0} \Dis_s.$
Also, we let $\WDis_t=(X_{t}\oplus Y_t)$, i.e, as opposed to $\Dis_t$, $\WDis_t$
is not restricted to $\ImpVrtx$. Analogously to $\Dis_{\leq t}$, we define
$$
\WDis_{\leq  s} =\bigcup^s_{t=0}\WDis_t
$$

Let $C'>\epsilon^{-3}$ and let $R=C' (\log d)\sqrt{d}$. The subgraph of $G$ induced by  $\ball(u^*, R+1)$ is a tree.
This follows from the condition 2.c of Definition \ref{def:SparseBlockPart}.
Let $T_0$ be the random time at which $\WDis_{\leq T_0}$ includes, for the first time, 
vertices outside $\ball(u^*, R)$.  For $T=\min\{T_0, CN\}$,  
let ${\cal A}$ be the event that  $\WDis_{\leq T}  \not\in \ball(u^*, R)$.  
Also, let ${\cal E}$ be the event that $|\Dis_{\le T}|\leq \sqrt{d}$.
It holds that
\begin{equation} 
\Pr \left [ D_{\leq T} \not\in \ball(u^*, R) \right] \  \leq \   \Pr[{\cal A}]  
\  \leq\  \Pr[{\cal E}^c]+\Pr[{\cal A}\ |\ {\cal E}]. \label{eq:Base4TailBound4DtGamma}
\end{equation}
The proposition  follows by bounding appropriately $\Pr[{\cal E}^c]$, $\Pr[{\cal A}\ |\ {\cal E}]$.

Noting that  $T\leq CN$, Lemma \ref{lemma:RateOfDisSpreadBurnIn} implies that 
\begin{equation}\label{eq:CalE^cTailBound}
\Pr[{\cal E}^c]\leq  \textstyle \exp\left( -d^{1/2-\gamma}/C\right).
\end{equation}

\noindent
As far as $\Pr[{\cal A}\ |\ {\cal E}]$ is regarded, we have the following:
Consider some vertex $w\in S(u^*,R+1)$. 
Let ${\cal P}(u^*,w)$ be the  unique  path that connects $u^*,w$ in $\ball(u^*, R+1)$.
Let $B_1, B_2, \ldots, B_h$ be the sequence of block we encounter as we traverse the  
${\cal P}(u^*,w)$ from $u^*$ towards $w$. 

Consider the subpath induced by ${\cal P}(u^*,w)\cap B_j$, for every $j\in[h]$. Let $v^j_a, v^j_b$ be the first and
the last vertex in this subpath  as we traverse vertices from $u^*$ to $w$. It could be that
$v^j_a, v^j_b$ are identical, i.e., for some $j$ we have $|{\cal P}(u^*,w)\cap B_j| =1$.
Let $\ImpVrtx_j$ be the set that contains every $u\in {\cal P}(u^*,w)\cap B_j\cap \ImpVrtx$.
Note that if $j<h$, then both $v^j_a, v^j_b\in \ImpVrtx_j$. Also, it holds that $v^h_a\in \ImpVrtx$,
whereas $v^h_b=w$ could be an internal vertex of $B_h$. 

For every $i\in [h]$, let $t_i\in [T]\cup\{\infty\}$ be the least $t$ such that $v^i_b\in D_{\leq t}$. 
So as to have $w\in D_{\leq T}$, it is necessary to have $t_h\leq T$.
Let $Q_w$ be the event that  $t_h\leq T$.

Since every for every $i <q$ we have $v^i_{a}, v^i_b\in \ImpVrtx$,
conditioning on the event ${\cal E}$ implies that $h\leq \sqrt{d}$.
With this observation and  a simple union bound, we get that
\begin{equation}\label{eq:PrAEReductFromUnionBound}
\textstyle \Pr[{\cal A}\ |\ {\cal E}] \ \leq \   \sum_{w\in \hat{S}(u^*, R+1) }\Pr[Q_w \ | \ {\cal E}],
\end{equation}
where   $\hat{ S}(u^*,R+1)\subseteq S(u^*,R+1)$ contains the  vertices $u$  such that 
the path ${\cal P}(u^*, u)$ contains   at most $\sqrt{d}$ vertices in $\ImpVrtx$.
We  get a upper bound  for $\Pr[{\cal A}\ |\ {\cal E}] $,  by   bounding  appropriately each  
$\Pr[Q_w \ | \ {\cal E}]$ in \eqref{eq:PrAEReductFromUnionBound}
and using the fact that  that the number of summads in \eqref{eq:PrAEReductFromUnionBound}
is at most $((1+\epsilon/3) d)^{R+1}$.

Recall that it is assumed that $u^*\in \ImpVrtx$. This implies that $u^*$ is either a break-point
or a vertex next to a break-point. Then, the bound on the cardinality of ${\cal S}(u^*,R+1)$ follows from
 Lemma \ref{lemma:GrowthFromBP}.

\begin{proposition}\label{prop:StepIProbPropag}
Let  $\epsilon, k, d, {\cal B},  G, u^*,C, C'$ and the copies of block dynamics 
$(X_t)_{t\ge 0}$ and $(Y_t)_{t\ge 0}$
as defined in the statement of Theorem \ref{thrm:BallResult}. 
Also, let  $R=C'(\log d)\sqrt{d}$. 

For a vertex  $w\in {\cal S}(u^*, R+1)$,
and the path ${\cal P}(u^*,w)$  consider the sequence of blocks $B_1, B_2, \ldots, B_h$  as defined above.
For every $i\in [h]$, letting ${\cal F}_i$ be the $\sigma$-algebra generated by $t_1, \ldots, t_{i-1}$,
we have
\[
\beta(i):=  \max_{{\cal F}_{i}}\  \Pr[ t_{i} < T \ | \ {\cal F}_{i}] \ \leq  \ 
(1.45)^{r_i} \left( (1+\epsilon)d \right)^{-\ell_i},
\]
where $r_i=|\ImpVrtx_i|$ and $\ell_i$ is the length of ${\cal P}(u^*,w)\cap B_i$.
\end{proposition}

\noindent
The proof of Proposition \ref{prop:StepIProbPropag} appears  in Section \ref{sec:prop:StepIProbPropag}.

Additionally, we have that
\[
 \Pr[Q(w)\ | \ {\cal E}] \ \leq \ \frac{\Pr[Q(w)] }{\Pr[ {\cal E}] } \ \leq \ 2\Pr[Q(w) ] \ \leq\  2  \prod^h_{j=1} \beta(j),
\]
where the third inequality follows from \eqref{eq:CalE^cTailBound}, while $\beta(j)$ is defined in Proposition \ref{prop:StepIProbPropag}. 

Using  Proposition  \ref{prop:StepIProbPropag} we have that 
\begin{equation}
\Pr[ Q(w)\ | \ {\cal E}]   \leq   2 (1.45)^{\sum^h_{i=1}r_i}\left( (1+\epsilon)d \right)^{-\sum^h_{j=1}\ell_j} 
\ \leq \ 2 (1.45)^{\sqrt{d} }\left( (1+\epsilon)d \right)^{-(R+1-h)}.  \label{eq:BasicPw|CondE}
\end{equation}
The sum of $r_j$s, counts the number of vertices in $\ImpVrtx \cap {\cal P}(u^*,w)$.
In the last inequality, above, we used the fact that $\sum_{j}r_j\leq \sqrt{d}$, due to the choice of the path.
Also,  we have argued,  previously, that   $h\leq \sqrt{d}$. 
Since $(1+\epsilon)>(1+\epsilon/3)(1+\epsilon/2)$ and   $C>\epsilon^{-3}$,  the above inequality 
yields
\begin{eqnarray}
\Pr[ Q(w)\ | \ {\cal E}] &\leq & 2    ((1+\epsilon/3)d )^{-(R+1)} \left( (2d)^{\sqrt{d}}(1+\epsilon/2)^{-(R+1)}\right) \nonumber \\
&\leq &   \textstyle ((1+\epsilon/3)d )^{-(R+1)} 
 \exp \left( -\epsilon^{-1} (\log d) \sqrt{d}  \right).
\label{eq:BasicBound4PrQw} 
\end{eqnarray}

\noindent
Combining  \eqref{eq:BasicBound4PrQw},  the observation that   $|\hat{ S}(u^*,R+1)\leq ((1+\epsilon/3)d )^{-(R+1)}|$,
and \eqref{eq:PrAEReductFromUnionBound} we get that
\begin{equation}\label{eq:PrCalACondCalEFinal}
\Pr[{\cal A}\ |\ {\cal E}] \ \leq   \textstyle \exp \left( -\epsilon^{-1}(\log d)\sqrt{d}  \right).
\end{equation}

\noindent
The proposition follows by plugging \eqref{eq:PrCalACondCalEFinal}  and  \eqref{eq:CalE^cTailBound}
into \eqref{eq:Base4TailBound4DtGamma}. 
$\hfill \Box$

\subsubsection{Proof of Proposition \ref{prop:StepIProbPropag}} \label{sec:prop:StepIProbPropag}

A direct corollary from Corollaries \ref{corr:GMPBiasColour} and  \ref{cor:GeneralBiasColour} is the following
result.
\begin{corollary}\label{corollary:Point2PointDisagreement}
Let  $\epsilon, k, d, {\cal B}, G, u^*,R$ and the copies of block dynamics 
$(X_t)_{t\ge 0}$ and $(Y_t)_{t\ge 0}$
as defined in the statement of Proposition  \ref{prop:StepIProbPropag}.

For any $B$ such that $u^*\in \outBound B$ and for any  $u\in \inBound B\cap \ball(u^*, R)$  it holds that  
\[
\Pr[u\in D_1 \ | \ \textrm{$B$ is updated at time $t=1$}] \leq  (1.45)^{r} \left( (1+\epsilon)d \right)^{- \ell},
\]
where $\ell$ is the length of the shortest path, in $B$, between  $u^*$ and $u$ and $r$
is the number of vertices in $\ImpVrtx$ which belong to  this path.
\end{corollary}

\noindent
The proof is immediate, for this reason we omit it.

\begin{proof}[Proof of Proposition \ref{prop:StepIProbPropag}]

Since the coupling stops when the disagreements escape $\ball(u^*,R+1)$, the girth assumption 
about $G$ implies that if $v^i_{b}$ becomes disagreeing then the disagreement can only  come from 
the disagreement at vertex $v^{i-1}_{b}$. 

Let ${\cal P}_i$  be the subpath  ${\cal P}(u^*,w)\cap B_i$.  Also, recall that  $\ImpVrtx_i= {\cal P}_i \cap \ImpVrtx$.
Consider the vertices  $u_1, \ldots, u_{s} \in \ImpVrtx_i$ in the order we discover them as we
traverse ${\cal P}(u^*,w)$ from $u^*$ towards $w$, i.e.,  $s=|\ImpVrtx_i |$.
Between $u_j$ and $u_{j+1}$ we encounter the vertices $w^j_1, w^j_2, \ldots, w^j_m$.

Note that $B_i\cap \ball(u^*, R+1)$ induces a graph which is a tree which we call it ${\cal T}$. We assume that
the root of the tree is vertex $v^i_a$. Also, for each vertex $u\in T\cap {\cal P}_i$, let ${\cal T}(u)$ be the subtree 
rooted at  $u$ and contains $B_i\cap \ball(u^*, R+1)$ apart from the vertices in  ${\cal P}_i$ that follow $u$.
Also, we let  $\Gamma(u)={\cal T}(u) \cap (\ImpVrtx\setminus \Gamma_i)$.

Consider the first update of block $B_i$, given that there is a disagreement at vertex $v^{i-1}_{b} $. Then,
according to Corollary  \ref{corollary:Point2PointDisagreement} the disagreement reaches vertex 
$v^{i}_{b}$ with probability   $\rho(\ell_i)$, where
$$
 \rho(\ell_i)\leq (1.45)^{r_i} \left( (1+\epsilon)d \right)^{-(\ell_i+1)}.
$$

\noindent
Consider, now, the next  update of $B_i$, i.e., the second one. It could be that
during the first update the disagreement did not reach $v^i_{b}$. However, it could have proceeded towards
this vertex, as follows: There is some $j < s$ and $r>0$ such that during the first update
the disagreeing reached up to vertex $w^j_{r} \in {\cal P}_i$. Furthermore, the disagreement continued  towards
some vertex in $\Gamma(w^j_{ r})$, i.e., following a different direction than that of ${\cal P}_i$. 
Then, between the first and second update of $B_i$, it could be that
some breakpoints, outside $B_i$,  but   neighboring to  disagreeing vertices in $\Gamma(w^j_{ r})$,
$\Gamma(w^j_{ r-1}), \ldots$ 
were updated and became disagreeing and remained disagreeing even during the second update of $B_i$.
In such a situation the probability of creating a disagreement on $v^i_{b}$ during the second update 
could become higher.   In particular,  if for $\ell'$, the distance between the closest disagreeing break
point next to $B_i$ and $v^{i}_{b}$ during the second update we have $\ell \ll \ell_i$, then  the 
probability of getting $v^i_{b}$  is substantially higher than $\rho(\ell_i)$.

We have a  similar situation  if   the disagreement at the first updated  stop at some vertex $u_j$, for
$j<s$ and between the first and second update of $B_i$, neighboring breakpoints became 
disagreeing.

\begin{claim}\label{claim:IJRProbUpperBound}
Assume that the last update of $B_i$ reached up to vertex $w^j_r$ in ${\cal P}_i$, for some $j,r=1,\ldots$
Let ${\cal I}_{j,r}$ be the event that $w^j_r$ becomes disagreeing at the next update. 
Then, it holds that
\[
\Pr[{\cal I}_{j,r}]\leq d^{-3/2}.
\]
\end{claim}
\noindent
The proof of Claim \ref{claim:IJRProbUpperBound} is right after this proof.

\begin{claim}\label{claim:BPDisProbUpperBound}
Assume that the last update of $B_i$ reached up to vertex $u_j$ in ${\cal P}_i$, for some $j < s$.
Let ${\cal I}_{j}$ be the event that $u_j$ becomes disagreeing at the next updated. 
 Then, it holds that
\[
\Pr[{\cal I}_{j,r}]\leq (\log d)^4d^{-1}.
\]
\end{claim}

\noindent
The proof of Claim \ref{claim:BPDisProbUpperBound} is very similar to the proof of Claim 
\ref{claim:IJRProbUpperBound} for this reason we omit it.

Let $U$ be the number of updates of block $B_i$ from $t_{i-1}+1$ up to time $T$. 
Also let $m$ be the number of updates of $B_i$, out of these $U$, at which the disagreement
propagates further towards $u_s=v^i_{b}$.
Given $U$ the probability that $v^i_{b}$ becomes disagreeing is at most
$ (1.45)^{r_i} \left( (1+\epsilon)d \right)^{-(\ell_i)}  \times \gamma(U)$, where 
\[
\textstyle \gamma(U)= \sum^{U}_{m=1} {U \choose m} {\ell_i-1 \choose m} \left( \frac{(\log d)^5}{d}\right)^{m}.
\]
Moreover, it holds that
\begin{equation}\label{eq:basis:prop:StepIProbPropag}
\beta(i)  \leq    (1.45)^{r_i} \left( (1+\epsilon)d \right)^{-(\ell_i)}  \times \gamma(U) \ExpCond{ \gamma(U)}{  {\cal F}_i},
\end{equation}
where the expectation is w.r.t. to the randomness of $U$.

So as to proceed, consider the following:
Noting that $\ell_i\leq d^{3/5}$, we have that
\begin{eqnarray}
\gamma(U) &\leq &   \sum^U_{m= 1}{U \choose m} 
\left( \frac{(\ell_i-1)e (\log d)^5}{d m}\right)^{m} \ 
\leq   \     \sum^U_{m = 1}{U \choose m} 
\left( d^{-1/5} \right)^{m}  \nonumber\\
&\leq &  U  d^{-1/5} \sum^{U-1}_{m=0}{U-1 \choose m} 
\left( d^{-1/5}  \right)^{m} \nonumber \\
 &\leq &  U  d^{-1/5}  \ \exp \left(U d^{-1/5} \right). \nonumber
\end{eqnarray}
Since $T\leq CN$, conditional on ${\cal F}_i$,  $U$ is dominated
by ${\tt Binomial}(CN,1/N)$.  Noting that $f(x)= ax \exp( a x)$ is an increasing function of $x$, when $a>0$,
it is standard to show that
\[
\ExpCond{ \gamma(U)}{ {\cal F}_i}  \ \leq \  \ExpCond{ U  d^{-1/5}  \ \exp \left(U d^{-1/5}  \right)}{  {\cal F}_i }
\ \leq  \textstyle \ C  d^{-1/5}  \ \exp \left(C d^{-1/5} \right) 
\ \leq \  2C d^{-1/5}.
\]
The proposition follows by plugging the above inequality into \eqref{eq:basis:prop:StepIProbPropag}.
\end{proof}

\begin{proof}[Proof of Claim \ref{claim:IJRProbUpperBound}]

For the sake of brevity, in this proof, we let ${\cal T}={\cal T}(w^{j}_r)$. Let $D_{\ell}({\cal T})$ be the number of 
disagreeing vertices in ${\cal T} \cap (\ImpVrtx\setminus \Gamma_i)$ which are at distance $\ell$ from $w^j_r$.
It holds that
\[
\textstyle \Exp{ D_{\ell}({\cal T}) }  = \Pr[\textrm{$w^j_r$ disagrees}]\sum_{y\in N(w^{j}_r)\cap B} \Exp{ D_{\ell-1}({\cal T}_y) },
\]
where ${\cal T}_y$ is the subtree of ${\cal T}$ rooted at $y$.  The general form for the  above inequality, i.e., for any 
$w\in {\cal T}$ at level $i  < \ell$, is as follows:
\[
\textstyle \Exp{ D_{\ell-i}({\cal T}_w) }  =  \Pr[\textrm{$w$ disagrees}]\sum_{y\in N(w) \cap B } \Exp{ D_{\ell-i-1 } ({\cal T}_y) }.
\]

\noindent
From the above, we get that
\begin{eqnarray}
\Exp{ D_{\ell}({\cal T}) }
&< & \Pr[\textrm{$w^j_r$ disagrees}] \ \degree_{in}(w^j_r)\ \max_{y\in N(w^j_r)\cap B} \left\{ \Exp{ D_{\ell-1}({\cal T}_y) } \right\}   
\nonumber \\
&\leq& 
\max_{\P'=(u_0=w^j_r ,u_1,\dots,u_\ell)}
p_{u_{\ell }}(0)  \prod_{i=0}^{\ell-1} p_{u_i}(0) \times
\left [ \degree_{in}(u_i) \right], \qquad \label{eq:EDEllTGenBound}
\end{eqnarray}
where the quantities $p_{v_i}$ above are defined in \eqref{eq:DefPu}.
Let $M$ be the set of high degree vertices in $\P'$ and let $m=|M|$. 
Recalling that $u_{\ell}\in \ImpVrtx$, Corollary \ref{cor:FromBreakPointProd} and  \eqref{eq:EDEllTGenBound} imply that
\begin{eqnarray}
\Exp {D_{\ell}({\cal T}) }  &\leq & 
\max_{\P'=(u_0=w^j_r ,u_1,\dots,u_\ell)} p_{u_{\ell}}(0) \left( \prod_{u_i\notin M} p_{u_i}(0) \times\left [ \degree_{in}(u_i) \right]\right )
\left( \prod_{u_i \in M}  \degree_{in}(u_i)  \right ) \nonumber \\
&\leq &
\max_{\P'=(u_0=w^j_r ,u_1,\dots,u_\ell)} p_{u_{\ell}}(0) \left( \prod_{u_i\notin M} p_{u_i}(0) \times\left [ \degree_{in}(u_i) \right]\right )
\frac{\left( (1+\epsilon/6) \right)^{\ell}}{((1+\epsilon/6)d^{15} )^m}
 \nonumber \\
 &\leq & 2 \left(  {1+2\epsilon/3}  \right)^{-\ell} d^{-14m}. \label{eq:DlTBound}
\end{eqnarray}

\noindent
Consider the disagreeing vertex  $w\in \ImpVrtx(w^j_r)$ at distance $\ell$ from $w^j_r$. 
The vertex $w$ has at most $\TDeg-1$ neighbors in $N(w)\setminus B$.
The number of steps between two consecutive updates of $B$ is at most $CN$.
A vertex in  $ N(w)\setminus B$ is chosen to be updated with probability  $|N(w)\setminus B| /N\leq (1+\epsilon/6)d/N$ 
and each update creates a new disagreement with probability at most $1/((1+\epsilon)d)$.

From the above remarks, we conclude that the number of vertices in $N(w)\setminus B$ which becomes
disagreeing between two consecutive updates of $B$ is dominated by the binomial distribution with
parameters $CN$,  $({(1+\epsilon/2) N})^{-1}$.   Chernoff bounds implies that with probability greater than
$1-\exp(-(\log d)^5)$, the number of disagreements of $w$, when $B$ is updated again is less than $(\log d)^5$.

At the next update of $B$, the disagreements next to $\ImpVrtx(w^j_r)$ travel back, towards vertex $w^j_r$.
Let $R_{\ell}({\cal T})$ be the number of paths of disagreements, of length $\ell$, that reach back $w^j_r$. 
Let $K$ be the event that there exists some $ w\in \ImpVrtx(w^j_r)$, at distance $\ell$ from $w^r_j$, 
which has more than $(\log d)^5$ disagreements
in its neighborhood.  It holds that
\begin{equation}\label{eq:Basis4PRlTPositive}
\Pr[R_{\ell}({\cal T}) >0 ] \ \leq \ \Pr[K]+\Pr[R_{\ell}({\cal T}) >0 \ |\ K^c]  \ \leq \  \Pr[K]+\ExpCond{ R_{\ell}({\cal T}) }{  K^c} .
\end{equation}
From the union bound we have $\Pr[K\ |\ D_{\ell}({\cal T})] \leq D_{\ell}({\cal T}) \exp(-(\log d)^5)$. Also,  it holds that 
\begin{equation}\label{eq:Basis4PRlTPositiveFirst}
\Pr[K] \leq \left(  {1+2\epsilon/3}  \right)^{-\ell} d^{-14m} \exp(-(\log d)^5).
\end{equation}
Furthermore, we have that
\begin{eqnarray}
\ExpCond{ R_{\ell}({\cal T})}{  K^c }  &\leq & \Exp{ D_{\ell}({\cal T}) } (\log d)^5 \left( (1+\epsilon/2)d\right)^{-(\ell-m)} 
\nonumber  \\
&\leq &   (\log d)^5 \left( (1+2\epsilon/3)(1+\epsilon/2)d\right)^{-\ell} (d/(1+\epsilon/2))^{-14m} \nonumber \\
&\leq &   (\log d)^5 \left( (1+\epsilon)d\right)^{-\ell} (d/2)^{-14m} .
\label{eq:Basis4PRlTPositiveSecond}
\end{eqnarray}
Plugging \eqref{eq:Basis4PRlTPositiveFirst} and \eqref{eq:Basis4PRlTPositiveSecond} into  \eqref{eq:Basis4PRlTPositive}
we get that
\begin{equation}
\Pr[R_{\ell}({\cal T}) >0 ]  \leq \left(  {1+2\epsilon/3}  \right)^{-\ell} d^{-14m} \exp(-(\log d)^5)+
(\log d)^5 \left( (1+\epsilon)d\right)^{-\ell} (d/2)^{-14m}. \nonumber
\end{equation}

\noindent
Let ${\cal E}(w^j_r)$ be the event that when we have of disagreement  at $w^j_r$ coming from ${\cal T}$.
It holds that
$$
\Pr[{\cal E}(w^j_r)] \leq \textstyle \sum_{\ell \geq 2}\Pr[R_{\ell}({\cal T}) >0 ]\ \leq \  \hat{C}{(\log d)^5}/{d^2},
$$
for large $\hat{C}>0$.   Note that we set $\ell \geq 2$ in the above summation  since we assumed that $w^j_r\notin \ImpVrtx$. 
The claim follows
\end{proof}

\section{Proof of Percolation Results}

\subsection{Proof of Lemma \ref{lemma:Weight2OfBoudnary}}\label{sec:lemma:Weight2OfBoudnary}

The proof of Lemma  \ref{lemma:Weight2OfBoudnary} assumes the results in Section \ref{sec:BreakPointsObservations}.
Also, for each vertex $w\in B$ let
\begin{equation}\label{def:ReverseWeightSSS}
\chi(w) = 
\frac{\beta(\parent{w} )}{(1+\epsilon^2)\ \degree_{in}(\parent{w} ) }
\left( p_w \right)^{-1}.
\end{equation}
For $w$ it holds that $\beta(w)=\min\{1, \chi(w)\}$. The lemma follows by showing that $\chi(w)\geq 1/2$ for
every $w\in \inBound B$.

Consider some vertex $u \in \inBound B$. Let $w$ be the closest  ancestor of $u$ such
that $\beta(w)=1$. Let ${\cal P}(u,w)$ be the unique path (sequence of ancestors) in $B$ that connects 
$v,w$. E.g. let the path ${\cal P}:=v_0,v_1, \ldots, v_{\ell}$, where $u=v_0$ and $w=v_{\ell}$.

As far as $\chi(v_0)$ is regarded we have the following:
\begin{equation}\label{eq:Basis4Weight2OfBoudnary}
\chi(v_0)  \geq  \frac{\beta(v_1 )}{(1+\epsilon^2)\ \degree_{in}(v_1 ) } (p_{v_0})^{-1}  \ \geq \  
 \frac{ (p_{v_0})^{-1} }{(1+\epsilon^2)^{\ell}\degree_{in}(v_{\ell})}  \prod^{\ell -1}_{i=1}\frac{  (p_{v_i})^{-1}  }{\degree_{in}(v_i)}.
\end{equation}
We proceed by getting a lower bound for the product on the r.h.s. of the inequality above.
Let $S_1\subseteq \{1,\ldots, \ell-1\}$ be such that for every $j\in S_1$ we have $\degree(v_j)>\TDeg$.
Let $S_2\subseteq \{1,\ldots, \ell-1\}$ be such that for every $j\in S_2$ we have $\degree(v_j) \leq \TDeg$.
Then, we have 
\begin{eqnarray}
\prod^{\ell -1}_{i=1}\frac{  (p_{v_i})^{-1}  }{ \degree_{in}(v_i)} &=&
\left(\prod_{i\in S_1}\frac{  1 }{ \degree_{in}(v_i)} \right)
\left(\prod_{i\in S_2}\frac{  (p_{v_i})^{-1} }{ \degree_{in}(v_i)} \right) 
\ \geq \ \left(\prod_{i\in S_1}\frac{  1 }{ \degree_{in}(v_i)} \right)
\left (  (1+\epsilon)/(1+\epsilon)^2 \right)^{|S_2|}, \qquad \label{eq:StepA4ProductLowerBound}
\end{eqnarray}
where in the last derivation we use the fact that
for $v\in S_2$ we have $(\rho_v)^{-1}\geq \left (\frac{(1+\epsilon)}{1+\epsilon^3}\degree_{in}(v) \right)^{-1}$.

Since $v_0\in \ImpVrtx\cap B$, from  Corollary \ref{cor:FromBreakPointProd} we have 
\begin{equation}\label{eq:HighDegreeLowerBound}
\textstyle \prod_{v: \degree(v)>\TDeg} \left[  \degree(v) \right]^{-1} \geq  (1+r)^{-\ell-2+m} d^{-15m},
\end{equation}
where $r=\epsilon/10$ and $m$ is the number of large degree vertices in ${\cal P}$.
Note that the r.h.s. of \eqref{eq:HighDegreeLowerBound} includes vertex $v_{\ell}$, if it is a
high degree vertex. 

Assume first that $\degree(v_{\ell}) \leq \TDeg$. Note that if $|S_1|=m$, then  $|S_2|=\ell-m$. Using this 
observation and \eqref{eq:HighDegreeLowerBound} for  \eqref{eq:StepA4ProductLowerBound}, we get that
\begin{eqnarray}
\prod^{\ell -1}_{i=1}\frac{  (p_{v_i})^{-1}  }{ \degree_{in}(v_i)} &\geq &
(1+r)^{-2}\left( \frac{1+\epsilon}{(1+\epsilon^3)(1+r)}\right)^{\ell} \left( \frac{(1+r)(1+\epsilon^3)}{d^{15}(1+\epsilon)} 
\right)^{m}.
\end{eqnarray}
Plugging the above inequality into \eqref{eq:Basis4Weight2OfBoudnary} we get 
\begin{equation} \nonumber
\chi(v_0)  \geq  \frac{ (1+r)^{-2} (p_{v_0})^{-1} }{\degree_{in}(v_{\ell})}
\left( \frac{1+\epsilon}{(1+\epsilon^2)(1+\epsilon^3)(1+r)}\right)^{\ell} \left( \frac{(1+r)(1+\epsilon^3)}{d^{15}(1+\epsilon)} 
\right)^{m}.
\end{equation}

\noindent
Using the fact that  $r=\epsilon/10$ and $(1+\epsilon)\geq (1+\epsilon/9)(1+\epsilon/2)$, from the above we get that 
\begin{eqnarray}
\chi(v_0)  
&\geq &(1+r)^{-2}\frac{(p_{v_0})^{-1} }{\degree_{in}(v_{\ell})}   \left( 1+\epsilon/2 \right)^{\ell}
\left(\frac{(1+\epsilon/11)}{d^{15}(1+\epsilon)}\right)^m 
\nonumber \\
&\geq &(1+r)^{-2}\frac{(p_{v_0})^{-1}}{\degree_{in}(v_{\ell})}   \exp\left(  \epsilon (1-\epsilon/4)\ell/2\right)
\left( \frac{(1+\epsilon/11)}{d^{15}(1+\epsilon)}\right)^m  \mbox{[as $\ln(1+x)\geq x-x^2/2 $]}. \quad \label{eq:chio0Organized}
\end{eqnarray}

\noindent
For every  $v\in S_1$ there should be a certain number $\ell_v$ of low degree vertices to compensate for
the high weight $W(v)$. In particular, for every $v\in S_1$ it holds that
\[
\ell_v \geq   r^{-1}  \left[ 15\log d+\log (\degree(v))\right].
\]
Also, recalling that $m=|S_1|$, we have that
\begin{eqnarray}
\ell+1-m &  \geq & \textstyle r^{-1}  \left[ 15m\log d+\sum_{v\in M}\log (\degree(v))\right]
\ \geq \  16r^{-1} m \log d.\hspace{1.5cm} \mbox{[as $\TDeg>d$]} \label{eq:EllVsM}
\end{eqnarray}

\noindent
Plugging \eqref{eq:EllVsM} into \eqref{eq:chio0Organized}  yields
\begin{eqnarray}
\chi(v_0) &\geq &(1+r)^{-2}\frac{ (p_{v_0})^{-1} }{\degree_{in}(v_{\ell})}   \exp\left(  \frac{7\epsilon}{r}m\log d (1-\epsilon/4) \right)
\left( \frac{(1+\epsilon/11)}{d^{15}(1+\epsilon)}\right)^m \hspace{1cm} \mbox{[from \eqref{eq:EllVsM}]} \nonumber \\
&\geq &(1+r)^{-2}\frac{ (p_{v_0})^{-1} }{\degree_{in}(v_{\ell})} d^{60m} \left( \frac{(1+\epsilon/11)}{d^{15}(1+\epsilon)}\right)^m 
\hspace{4.65cm} \mbox{[as $r=\epsilon/10$]} \nonumber \\
&\geq &(1+r)^{-2} \frac{(p_{v_0})^{-1} }{\degree_{in}(v_{\ell})}  \nonumber \\
&\geq &(1+r)^{-2} \frac{k-d}{(1+\epsilon/6)d} \ \ \geq 1/2.  \nonumber 
\end{eqnarray}

\noindent
For the case where 
$\degree(v_{\ell}) > \TDeg$,  in \eqref{eq:Basis4Weight2OfBoudnary} we include
 ${\degree(v_{\ell})}$ in the product of degrees.  We bound the product of degrees in
 the same manner, i.e., using \eqref{eq:HighDegreeLowerBound}. Then, we get the results by
 using  almost identical arguments as above. For this reason we omit the details.
The lemma follows.

\noindent
\section{Proofs of Burn-In Analysis}

\subsection{Proof of Proposition \ref{prop:RapidMixAux1}}\label{sec:prop:RapidMixAux1}

\begin{proof}[Proof of Proposition \ref{prop:RapidMixAux1}.1]

For proving Proposition \ref{prop:RapidMixAux1}.1 we use path coupling and
Proposition \ref{prop:Tail4ZetaIntro}.

For any $t>0$, given $X_t, Y_t$ we  let $W_0=X_t, W_1, W_2, \ldots, W_h=Y_t$  be a sequence of 
colorings where $h= |(X_t\oplus Y_t)\cap \ImpVrtx|$.  
Consider an arbitrary ordering of the vertices in $(X_t\oplus Y_t)\cap \ImpVrtx$, 
e.g., $w_1, w_2, \ldots, w_h$.  We obtain $W_{i+1}$ from $W_i$ by changing the color of  
$w_i$ from $X_t(w_i)$ to $Y_t(w_i)$. It could be that  $w_i $ belongs to the block $B$ such 
that  there is no $j>i$ such that $w_j\in B$,  while there exists  $B'\subset B$  such that
$B'\in X_t\oplus Y_t$. This means that there are disagreements in  $B$ which do not 
belong to $\ImpVrtx$, while $w_i$ is the last vertex in $\ImpVrtx\cap B$ we consider.
If this is the case for $w_i$, then so as to get from $W_{i}$ to $W_{i+1}$  we not only
change $X_t(w_i)$ to $Y_t(w_i)$ but we change $X_t(B')$ to $Y_t(B')$, too.

We couple each pair $W_i, W_{i+1}$, for $i=0, \ldots, h-1$, and we get 
$W'_i, W'_{i+1}$. Recall that $N=|{\cal B}|$.  Proposition \ref{prop:Tail4ZetaIntro} implies that 
there is a coupling such that
\[
\ExpCond{ H(W'_i, W'_{i+1})} { W_i, W_{i+1} }  \leq \left(1+c \TDeg/(Nd) \right) \leq \left(1+2c/N \right),
\]
where $c>0$ is a fixed number, independent of $d$, while we use the fact that $\TDeg\leq 2d$.

Any disagreement which does not belong to $\ImpVrtx$ cannot spread during any update.
Then, path coupling implies that there is a coupling such that
\[
\ExpCond{ H(X_{t+1}, Y_{t+1})}{  H(X_{t}, Y_t) }  \leq \left(1+2c/N \right)H(X_{t}, Y_t).
\]
A simple induction on $t$ yields
$
\Exp{ H(X_{t}, Y_{t}) } \leq \exp\left( 2tc/N\right).$ The result follows by setting   $t=CN/\epsilon$,
in the previous inequality.
\end{proof}

\begin{proof}[Proof of Proposition \ref{prop:RapidMixAux1}.2]
It holds that
\begin{equation}\label{eq:BasisRapidMixAux1ReductionToDt}
\Exp{ |(X_T\oplus Y_T)\cap \ImpVrtx| \   \mathbf{1}\{\ManyDisEvent_T\} } \leq
\Exp{  |D_{\leq T} \cap \ImpVrtx| \   \mathbf{1}\{\ManyDisEvent_T\} }.
\end{equation}
 For small { $\gamma>0$ we specify later},  let $\mathbf{I}=[0, T]$
 and  let $\mathbf{I}_1, \ldots, \mathbf{I}_m$ be a partition of $\mathbf{I}$ into 
 $m$ subintervals each of length $\left \lfloor  T/m \right \rfloor$ (the last interval which maybe smaller),
 where $m=\left \lceil \gamma^{-1}\log d \right \rceil$.

Let $T' $ be the first time such that $|\Dis_{\le T'}|\geq d^{2/3}$.
Using similar arguments to those for  Theorem \ref{thrm:BallResult}, we see that  so as to have $T' \leq T$, at least 
one of the  following two events should happen:

\begin{description}
\item[$J_A:=$] There exists a subinterval $\mathbf{I}_{j}=[t_j, t_{j+1}-1]$ such that 
$|D_{t_j}| <d^{2/3-\gamma}$ and the increase in the number of
disagreements in the set $\ImpVrtx$, during $\mathbf{I}_j$, is at least $Cd^{2/3-\gamma}$,
for  large  $C>0$.

\item[$J_B:=$] There is a subinterval $\mathbf{I}_{j}=[t_j, \ldots, t_{j+1}-1]$ such that 
\[
d^{2/3-\gamma}\leq |D_{\leq t_j} | \leq d^{2/3},
\]
during which  the increase in the number of disagreements in $\ImpVrtx$  is
at least $(1+\gamma/2)|D_{\leq t_j}|$.
\end{description}

\noindent
Let ${\cal J}_T=J_A\cup J_B$.  Noting that  
$\ManyDisEvent_T\subseteq  {\cal J}_T$  we have 
\begin{equation}
\Exp{ |D_{\leq T} | \   \mathbf{1}\{\ManyDisEvent_T\} }
\leq \Exp{  |D_{\leq T} | \   \mathbf{1}\{{ \cal J}_T \} }.
\end{equation}

\noindent
In what follows, we let $\mathbf{I}_j$ be the set that the is involved in the realization of  ${ \cal J}_T$.
Also, let ${\cal L}$ be the event that there is at least one $t'\in \mathbf{I}_j$ such that
\[
\textstyle | D_{\leq t'} |-| D_{\leq t_j} | \in   (1\pm \delta)\frac{\gamma}{2} 
\max\left \{ 
 |D_{\leq t_j} |,  \ d^{2/3-\gamma}  
\right\}
\]
for (any) small fixed $\delta\in (10^{-3}, 10^{-2})$. 
Intuitively, the event ${\cal L}$ requires that $\mathbf{I}_j$ contains a $t'$ during which
the increase in $|D_{\leq t'} |$, compared to $|D_{\leq t_j} |$ falls within
a specific interval. 
It holds that
\begin{equation}\label{eq:EpxJt2JtL-Lc}
\Exp{  |D_{\leq T} | \   \mathbf{1}\{ {\cal J}_T\} } =
\Exp{ |D_{\leq T} | \   \mathbf{1}\{{ \cal J}_T, {\cal L}\} } +
\Exp{  |D_{\leq T}| \   \mathbf{1}\{{ \cal J}_T,  \bar{\cal L}\} }.
\end{equation}
We proceed by bounding  the two expectations on the r.h.s. of the above 
equality.

Consider $\Exp{ |D_{\leq T} | \   \mathbf{1}\{{ \cal J}_T, \bar{\cal L}\} }$.
If  ${\cal J}_T$ and $ \bar{\cal L}$ hold,  then, there should be a moment in $\mathbf{I}_j$ such that
a lot of disagreements are generated, i.e.,
there exist  $t''$ such that $t'', t''+1 \in \mathbf{I}_j$ and
\begin{equation}\label{eq:TPlus1Condition}
| D_{\leq t''+1} |-| D_{\leq t''} | \geq \gamma\ \delta 
\max\left \{ 
 |D_{\leq t_j} |,  \ d^{2/3-\gamma}  
\right\},
\end{equation}
while 
\begin{equation}\label{eq:T''Condition}
| D_{\leq t''} | < |D_{\leq t_j}|+ (1- \delta)({\gamma}/{2})\max\left \{  |D_{\leq t_j}\cap \ImpVrtx|,  \ d^{2/3-\gamma}   \right\}.
\end{equation}

\noindent
For the subinterval $\mathbf{I}_j$,  $t''$ is the latest  moment  that \eqref{eq:T''Condition} is true.
If, subsequently, $t''+1$ does not satisfy \eqref{eq:TPlus1Condition}, then the event must $\cal L$ occur.
The condition in \eqref{eq:TPlus1Condition} implies that   at time $t''+1$   a lot of  disagreements are generated in $\ImpVrtx$.

Let ${\cal R}$ be the  following event: There exists   $\mathbf{I}_{s}$, for some $s=1,2,\ldots, m$,
and  $t'', t''+1 \in \mathbf{I}_s$ which satisfy \eqref{eq:TPlus1Condition},\eqref{eq:T''Condition}, respectively,  while 
$| D_{\leq t''} | \leq 2d^{2/3}$. 
Noting that ${\cal J}_T \cap \bar{\cal L} \subseteq {\cal R}$, we have
\begin{equation}\label{eq:EpxJtLc2R}
\Exp{  |D_{\leq T} | \   \mathbf{1}\{{ \cal J}_T,  \bar{\cal L}\} } \leq
\Exp{  |D_{\leq T} | \   \mathbf{1}\{ {\cal R} \} }.
\end{equation}

\noindent
Let ${\tt Inc}(t)$ be the number of new disagreements in $\ImpVrtx$ generated at the update at time $t$. From path coupling we get that
\begin{eqnarray}
\Exp{  |D_{\leq T} | \   \mathbf{1}\{ {\cal R} \} }  &\leq &
\Exp{  |D_{\leq t''+1}  | \ \mathbf{1}\{ {\cal R} \} }   \   \Exp{ |D_{\leq T} |  }
\nonumber \\
&\leq &
\left( \Exp{ |D_{\leq t''}  | \ \mathbf{1}\{ {\cal R} \}  }  
+ \Exp{   {\tt Inc}(t''+1) \ \mathbf{1}\{ {\cal R} \}  }  \right )  \ \Exp{ |D_{\leq T} |  }
\nonumber \\
&\leq &
\left( 2d^{2/3} \Exp{ \mathbf{1}\{ {\cal R} \}  }
+ \Exp{   {\tt Inc}(t''+1) \ \mathbf{1}\{ {\cal R} \}  }   \right )  \ \Exp{ |D_{\leq T} | },
\label{eq:ExpctDisEventRBasis}
\end{eqnarray}
the last inequality follows from the direct observation that   $|D_{\leq t''} | \leq 2d^{2/3}$.

\begin{claim}\label{claim:IncreaseTail}
Let  $\gamma, \delta$ be as defined above,  $t\in \mathbf{I}$ and  let $s\in [m]$ be such that $t\in \mathbf{I}_s$. 
 Let $\lambda(t)=\max\left \{  |D_{\leq t_s}|,  \ d^{2/3-\gamma}  \right\}$.
There exists $C'>0$ such that  for any $\ell\geq \gamma\delta \lambda(t) $ the following is true:

Let ${\cal A}_t$ be the event that $| D_{\leq t-1}| \leq |D_{\leq t_s}|+ (1- \delta)\frac{\gamma}{2}\lambda(t)$
and $| D_{t-1}| \leq 2d^{2/3}$. Then, 
\begin{equation}\nonumber
\Pr[{\tt Inc}(t)\geq \ell \  |\  {\cal A}_t, D_{\leq t-1} ]  \leq {C'}N^{-1} \textstyle \exp\left( -\ell/C'\right).
\end{equation}
\end{claim}

\noindent
We omit the proof of the above claim, since it  follows by  using very similar arguments to those we use for Lemma \ref{lemma:BigJumpInBurnIn}.

For $t\in \mathbf{I}$,  consider the quantity $\lambda(t)$ and the event ${\cal A}_t$ as defined in Claim \ref{claim:IncreaseTail}.
We  have that
\begin{eqnarray}\label{eq:ExpctDisEventRBasisAComponent}
\Exp{ \mathbf{1}\{ {\cal R} \}   } &= & \Pr[{\cal R}] \ \leq \ \textstyle 
 \sum_{t\in \mathbf{I}} \Pr[{\tt Inc}(t)\geq \gamma\delta\lambda(t), {\cal A}_t   ]   \hspace{5cm} \mbox{[union bound]}\nonumber \\
& = &  \sum_{t\in \mathbf{I}}  \sum^{2d^{2/3}}_{ r=0 } \Pr[{\tt Inc}(t)\geq \gamma\delta \lambda(t) \ | \ {\cal A}_t, \ | D_{t-1}\cap \ImpVrtx| =r  ]
\ \Pr[{\cal A}_t, \ | D_{t-1}\cap \ImpVrtx| =r] \nonumber  \\
& \leq & \sum_{t\in \mathbf{I}} N^{-1}C' \exp\left( -d^{2/3-\gamma}/C'\right ) \sum^{2d^{2/3}}_{ r=0 }
\Pr[{\cal A}_t, \ | D_{t-1}\cap \ImpVrtx| =r] \nonumber \hspace{1.5cm} \mbox{[from Claim \ref{claim:IncreaseTail}]}\\
& \leq & C_0 \textstyle \exp\left( -d^{2/3-\gamma}/C_0\right ),
\end{eqnarray}
where $C_0>0$ is a sufficiently large constant, independent of $d$.
Furthermore, we have that
\begin{eqnarray}
\Exp{  {\tt Inc}(t''+1) \ \mathbf{1}\{ {\cal R} \}   }  &\leq &   \textstyle 
\sum_{t\in \mathbf{I}}  \Exp{ {\tt Inc}(t) \ \mathbf{1}\{ {\tt Inc}(t)\geq \gamma \delta \lambda(t) \}\ \mathbf{1}\{ {\cal A}_t \} }    
\nonumber\\
&\leq & \textstyle  \sum_{t\in \mathbf{I}} 
\Exp{  \ExpCond{ {\tt Inc}(t) \ \mathbf{1}\{ {\tt Inc}(t)\geq \gamma\delta\lambda(t) \}\ \mathbf{1}\{ {\cal A}_t \} }{D_{\leq t_s} }   }, \label{eq:ExpctDisEventRBasisComponentB-A}
\end{eqnarray}
in the above inequality we assume  that $t\in \mathbf{I}_s$, for some $s\in [m]$.

Note that $\lambda(t)$ is fully specified by $D_{\leq t_s}$.  For 
any $D_{\leq t_s}$ such that  $|D_{\leq t_s}\cap \ImpVrtx|\leq 2d^{2/3}$,  we have 
\begin{eqnarray}
\ExpCond{ {\tt Inc}(t) \ \mathbf{1}\{ {\tt Inc}(t)\geq \ell_0(t) \}\ \mathbf{1}\{ {\cal A}_t \} }{  D_{\leq t_s} }
&\leq&  \sum_{j \geq \gamma\delta\lambda(t)} 
j \Pr \left [ {\tt Inc}(t)=j, \ \mathbf{1}\{ {\cal A}_t \}  \ | \ D_{\leq t_s} \right]     \nonumber\\
&\leq& \sum_{j \geq \gamma\delta\lambda(t)} 
j \Pr \left [ {\tt Inc}(t) \geq j |  {\cal A}_t,  D_{\leq t_s} \right]     \nonumber\\
&\leq&  C_1N^{-1} d^{2/3} \textstyle \exp\left(-d^{2/3-\gamma}/C_1 \right), \label{eq:After:eq:ExpctDisEventRBasisComponentB-A}
\end{eqnarray}
for large $C_1>0$.
Due to  the indicator of ${\cal A}_t$, we have
$\ExpCond{ {\tt Inc}(t) \ \mathbf{1}\{ {\tt Inc}(t)\geq \gamma\delta\lambda(t) \}\ \mathbf{1}\{ {\cal A}_t \}}{  D_{\leq t_s} } =0$,
if   $|D_{\leq t_s} | > 2d^{2/3}$.
Combining  this observation with \eqref{eq:After:eq:ExpctDisEventRBasisComponentB-A}  and   \eqref{eq:ExpctDisEventRBasisComponentB-A},  
we get that
\begin{equation}\label{eq:ExpctDisEventRBasisComponentB}
\Exp{ {\tt Inc}(t''+1) \ \mathbf{1}\{ {\cal R} \}  }  \leq C_3d^{2/3}
 \textstyle \exp\left(-d^{2/3-\gamma}/C_3 \right),
\end{equation}
for large constant $C_3>0$.
Finally,  combining  \eqref{eq:ExpctDisEventRBasisAComponent},  \eqref{eq:ExpctDisEventRBasisComponentB} 
and  \eqref{eq:ExpctDisEventRBasis}  we get that
\begin{equation}\label{eq:EpxJtLc}
\Exp{ |D_{\leq T} | \   \mathbf{1}\{{ \cal J}_T,  \bar{\cal L}\} }  \leq 
  \textstyle \exp\left(-d^{3/5} \right).
\end{equation}

\noindent
Now consider  the quantity $\Exp{  |D_{\leq T} | \   \mathbf{1}\{{ \cal J}_T, {\cal L}\} }$.
For some  interval $\mathbf{I}_j$,   such that $|D_{t_j} | <d^{2/3-\gamma}$,
 the probability  that  event ${\cal J}_T, {\cal L}$ happens is  less than $\exp\left(- d^{2/3-2\gamma}\right)$.
This  follows from   Lemma \ref{lemma:BigJumpInBurnIn}.
Similarly, for interval $\mathbf{I}_j$ such that $d^{2/3-\gamma}\leq |D_{t_j}  | \leq d^{2/3}$,
the probability that  event ${\cal J}_T,{\cal L}$ happens is   less than $\exp\left(- d^{2/3-2\gamma}\right)$.

Furthermore, when  ${\cal J}_T$ and ${\cal L}$ occurs,  the expected number of
disagreements is at most 
$$
|D_{\leq t''} | \  \Exp{ |D_{\leq T} | } \leq 
2d^{2/3} \ \Exp{   |D_{\leq T}| }.
$$
The above follows from  path coupling.  Combining all the above  together we have that
\begin{eqnarray}
\Exp{  |D_{\leq T} | \   \mathbf{1}\{{ \cal J}_T,  {\cal L}\} }
 & \leq  &  \textstyle10 \exp\left(-d^{2/3-2\gamma} \right)d^{2/3}  \Exp{  |D_{\leq T}| } \nonumber \\
  & \leq  &  \textstyle10 \exp\left( C'/ \epsilon \right)  \exp\left(-d^{2/3-2\gamma} \right)d^{2/3} 
   \ \leq \  \exp\left(-d^{3/5} \right), \label{eq:EpxJtL}
\end{eqnarray}
where in the last inequality holds for any  $\gamma\in (0, 0.02)$. The  second inequality,  uses the first part of the
proposition to bound $\Exp{ |(X_T\oplus Y_{T} ) \cap \ImpVrtx | }$.

Combining \eqref{eq:EpxJtL}, \eqref{eq:EpxJtLc}, \eqref{eq:EpxJt2JtL-Lc} and 
\eqref{eq:BasisRapidMixAux1ReductionToDt}
we get  Proposition \ref{prop:RapidMixAux1}.2.
\end{proof}

\section{Local Uniformity: Proof of Theorem \ref{thrm:Uniformity1.76}}\label{sec:thrm:Uniformity1.76}

Given the vertex $v$,  let $G^*_v$ denote the graph which is derived from $G$ when we delete all
the edges that are incident to the vertex $v$. Also, let $(X^*_t)_{t\geq 0}$ denote the
corresponding block dynamics on $G^*_v$ where the blocks are identical to those
of $(X_t)_{t\geq 0}$.

In the graph $G^*_v$, the neighborhood of $v$ is empty, since we have deleted all the
incident edges.   However, we follow the convention and call ``neighborhood of $v$" the set of 
vertices which are adjacent to  $v$ in the graph $G$. We denote this set by $N^*(v)$.

\begin{lemma} \label{lemma:Restriction2Cycle}

For all $\epsilon, \Delta, C>0$ there exists positive $d_0$ such that for all $d\geq d_0$
for $k=(\alpha+\epsilon)d$  and every  graph   $G\in \IntrstGraphFam(\epsilon, d, \Delta)$, where 
$\Delta>0$ can depend on $n$,  with  set of blocks ${\cal B}$, for any $v\in \ImpVrtx$ the following is true:

Consider,  $G^*_v$ and the  block dynamics $(X^*_t)_{t\geq 0}$ and $(Y^*_t)_{t\geq 0}$.
Assume  that $X^*_0(w)\neq Y^*_0(w)$ for every $w\in N^*(v)$,
while $X^*_0(w)=Y^*_0(w)$ for every $w\notin N^*(v)$.
There is a coupling for  $ (X^*_t)_{t\geq 0}$ and $(Y^*_t)_{t\geq 0}$  such that 
\[
\Pr\left[ 
D^*_{\leq CN} \not \subseteq \ball \left (N^*(v), d^{4/5} \right) 
 \right]\leq  \textstyle \exp\left( -d^{3/4} \right),
\]
where for any time $s$,  we let  $\WDis^*_{\leq s}=\bigcup_{t'\leq s} (Y^*_{t'}\oplus X^*_{t'})$
\end{lemma}

\begin{proof}
The proof of  Lemma \ref{lemma:Restriction2Cycle} uses the results from Section \ref{sec:DisPerc}.
In particular, the proof is very similar to that of Theorem \ref{thrm:BallResult}.
The main difference between this proof and that of Theorem \ref{thrm:BallResult} is the assumption
that all the vertices in $N^*(v)$ at time $t=0$ are disagreeing. Since $v\in \ImpVrtx$,
it could be that $|N^*(v)|=(1+\epsilon/6)d$, whereas for proving Theorem \ref{thrm:BallResult} we
need the assumption that the number of disagreements cannot get larger than $d^{9/10}  \ll d$.

So as to prove the lemma, first we reduce the problem to studying the spread of disagreements 
for a pair of chains which,  at time $t=0$,  disagree on a single vertex in $N^*(v)$.
More specifically, we work as follows: Consider an arbitrary ordering of the vertices in $N^*(v)$, 
e.g., $w_1, w_2, \ldots, w_h$, where $h=|N^*(v)|$. 
We have a sequence of configurations $\tau_0, \tau_1, \ldots, \tau_h$ such that
$\tau_0=X^*_0$ and $\tau_{h}=Y^*_0$, while $\tau_{i}$ differs from $\tau_{i+1}$
in the assignment of vertex $w_i$.
Furthermore, consider block dynamics $(W^i_t)_{t\geq 0}$, for $i=0, \ldots, h$, 
such $(W^0_t), (W^h_t)$ are identical to $(X^*_t)$ and $(Y^*_t)$, respectively.
Furthermore, we assume that $W^i_0=\tau_i$.

We are  coupling the pairs $(W^i_t)$ and $(W^{i+1}_t)$, for $i=0, \ldots, h-1$, simultaneously.
That is, at each time $t$, we have
a transition for $(W^0_t)$, then, given this transition  the coupling decides a move for 
$(W^1_t)$, then, given the move of $(W^1_t)$, it decides the move for $(W^2_t)$ and so on.
In this setting,  let $B_i$ be the event that $\left( \bigcup_{t\leq CN} (W^i_t \oplus W^{i+1}_{t}) \right)\not \subseteq
\ball \left (N^*(v), d^{4/5} \right)$. It is elementary to verify that
\begin{equation}
\Pr\left[  \WDis^*_{CN} \not \subseteq \ball \left (N^*(v), d^{4/5} \right)  \right] 
\  \leq \ \textstyle \Pr\left[ \bigcup^{h-1}_{i=1} B_i\right] \ \leq \ \sum^{h-1}_{i=1}\Pr[B_i]. \label{eq:basis:lemma:Restriction2Cycle}
\end{equation}
The last inequality follows from the union bound.

For bounding $\Pr[B_i]$ we just work as in the proof of Theorem \ref{thrm:BallResult}. 
Then, for every $i=0,\ldots, h-1$, we have that
$\Pr[B_i]\leq  \textstyle \exp\left( -d^{0.77}\right).$
The lemma follows by plugging the above bound into 
\eqref{eq:basis:lemma:Restriction2Cycle}.
\end{proof}

\begin{lemma}[$G$ versus $G^*_v$]\label{lemma:GVsGStartNeighbors}

In the same setting as Lemma \ref{lemma:Restriction2Cycle} the following is true:

Let $(X_t)_{t\geq 0}$ be  the block dynamics  on $G$.  Also,  consider $G^*_v$ and the corresponding block 
dynamics $(X^*_t)_{t\geq 0}$. Assume that $X_0=X^*_0$. For any time $s$,  let  
$\WDis_{\leq s}=\bigcup_{t'\leq s} (X_{t'}\oplus X^*_{t'})$. There is a coupling of  $(X_t)_{t\geq 0}$ and 
$(X^*_t)_{t\geq 0}$ such that 
\[
\Pr\left[ 
|\WDis_{\leq CN} \cap N^*(v)| \geq \gamma^2 d
 \right]\leq  \textstyle3\exp\left( -d^{3/4}\right).
\]
\end{lemma}

\begin{proof}
We  couple   $(X_t)$ and $(X^*_t)$ such that at each time step we update the same block for both chains.
Then,  it is possible that disagreements are generated because of the fact that in $G^*_v$ the vertices in 
$N^*(v)$ are not connected with $v$. E.g.,  consider some vertex $w\in N^*(v)$ and assume that this is a 
single vertex block. If the coupling updates $w$  at time step $t$,  then,  for  setting $X_t(w)$ we need to 
consider the coloring of vertex $v$. On the other had  the choice of $X^*_t(w)$ is oblivious to the coloring of $v$.
This difference can create disagreement at vertex $w$.

If some vertices in $N^*(v)$ becomes disagreeing, then, subsequently, the disagreements generated 
 propagate to the whole graph.  That is, disagreements in  $N^*(v)$ generate disagreements to vertices at 
 further distances.

Let $t_0=CN$. Assume that we couple the two chains $(X_t)$ and $(X^*_t)$ up to the point  in time  
$T\leq t_0$ such that at least one of the following happens (whatever happens first):
\begin{enumerate}
\item  there are disagreements outside the ball $\ball \left (N^*(v), d^{4/5} \right)$.
\item  $|\WDis_{\leq T} \cap \ImpVrtx|  \geq  d^{3/4}$
\item  $|\WDis_{\leq T} \cap N^*(v)| \geq \gamma^2 d$ 
\item  we have run the coupling for $t_0$ steps.
\end{enumerate}

\noindent
Let ${\cal B}_1$ be the event that $\WDis_{\leq T}\cap \ball \left (N^*(v), d^{4/5} \right)\neq \emptyset$.
Let  ${\cal B}_2$ be the event that $|\WDis_{\leq T} \cap \ImpVrtx| \geq d^{31/40}$.
Finally, let ${\cal B}_3$ be the event that $|\WDis_{\leq T} \cap N^*(v)| \geq \gamma^2 d$. 
Clearly, it holds that
\begin{eqnarray}
\Pr\left[ |\WDis_{\leq t_0} \cap N^*(v)| \geq \gamma^2 d \right]
& \leq &  \Pr\left[ {\cal B}_1\cup {\cal B}_2\cup {\cal B}_3 \right] \nonumber \\
 & \leq &  \Pr\left[ {\cal B}_1 \right]+ \Pr\left[ {\cal B}_2 \right]+  \Pr\left[ {\cal B}^c_1\cap {\cal B}^c_2\cap {\cal B}_3 \right].
 \label{eq:Basis4lemma:GVsGStartNeighbors}
\end{eqnarray}

\noindent
The lemma will follow by bounding appropriately the probability terms on the r.h.s. of 
\eqref{eq:Basis4lemma:GVsGStartNeighbors}. The approach we follow is very similar for 
all the terms. In particular we use  results from Section \ref{sec:DisPerc}.

Working as in  Theorem \ref{thrm:BallResult} we get that
\begin{equation}\label{eq:lemma:GVsGStartNeighbors-B1Bound}
\Pr[{\cal B}_1]\leq   \textstyle \exp\left( -d^{3/4} \right).
\end{equation}

\noindent
Furthermore, using Lemma \ref{lemma:RateOfDisSpreadBurnIn} we get that
\begin{equation}\label{eq:lemma:GVsGStartNeighbors-B2Bound}
\Pr[{\cal B}_2]\leq   \textstyle \exp\left( -d^{3/4} \right).
\end{equation}

\noindent
Assume that the events ${\cal B}^c_1$ and $B^c_2$ hold. Then our girth assumption for $G$ imply
that there is $\hat{N}\subseteq N^*(v)$ which contains all but at most two vertices of $N^*(v)$ such that
the  following is true: Every time a vertex $w\in \hat{N}$ is updated it becomes disagreeing with
probability at most $2/d$, regardless of whether the other vertices in $\hat{N}$ are disagreeing or not.

Let us  be more specific.   If $w$ belongs to a single vertex block, then ${\cal B}^c_1$ and the girth assumptions 
imply that the disagreement at $w$ can only be caused by the lack of edge between $w$ and $v$. 
If, on the other hand, $w$ belongs to a multi-vertex  block, then the disagreement at $B_w\cap \ImpVrtx$ can  influence
$w$ and generate a disagreement. However, when ${\cal B}^c_1$, ${\cal B}^c_2$ hold, then the girth assumption and
Proposition \ref{prop:StochasticDomination} imply that the influence on $w$ by distant disagreements is minor.

We proceed by considering the rate at which disagreements are generated at $N^*(v)$.  If $v$ belongs to a multi-vertex block then
it can be that many vertices in $\hat{N}$ are updated simultaneously.  Still, as long as ${\cal B}^c_1, {\cal B}^c_2$ occur, 
the probability of disagreements at each vertex $\hat{N}$ is  at most $2/d$, regardless of whether the other vertices in 
$\hat{N}$ are disagreeing or not.  On the other hand,  if $v$ belongs to a single-vertex block then only one vertex in 
 $\hat{N}$ are updated at a time.

Let $\tilde{N}\subset N^*(v)$ contain the vertices which belong to the same block as $v$.
Let $S_1=\WDis_{\leq T}\cap \tilde{N}$. Also, let $S_2=\WDis_{\leq T}\cap (N^*(v)\setminus \tilde{N})$.
Let $\tilde{N}$ be such that $|\tilde{N}|=a d$, for some $a\in [0,1]$.  We will get tail bounds for the
cardinalities of $S_1$ and $S_2$, respectively, by considering cases for $a$. 

First, assume that $a>\gamma^2/10$.  Let 
${\cal B}_4$ be the event that $|S_1|\geq  (\gamma^2/5) d$. 
Also, let ${\cal B}_5$ be the event that $|S_2|\geq (\gamma^2/5) d$. 
Then, we have that
\begin{eqnarray}
\Pr\left[ {\cal B}^c_1\cap {\cal B}^c_2\cap {\cal B}_3 \right] &\leq &
\Pr\left[ {\cal B}^c_1\cap {\cal B}^c_2\cap ({\cal B}_4\cup {\cal B}_5)  \right]\nonumber \\
&\leq &\Pr\left[ {\cal B}^c_1\cap {\cal B}^c_2\cap {\cal B}_4 \right]+
\Pr\left[ {\cal B}^c_1\cap {\cal B}^c_2\cap {\cal B}_5  \right], \label{eq:ReductionofBc1Bc2B3ToBc1Bc2B45LargeA}
\end{eqnarray}
where the second inequality follows from the  union bound.

First we consider $\Pr\left[ {\cal B}^c_1\cap {\cal B}^c_2\cap {\cal B}_4  \right]$.
At each step, the update chooses  $N^*(v)$ with probability  $1/N$. 
Let ${\cal Q}$ be the event that the block that $v$ belongs is updated
at least $d^{4/5}$ times, during the time interval $[0,T]$.
Then,  we have that
\begin{equation}\label{eq:ReductionofBc1Bc2B3ToBc1Bc2B5LargeABoundB4-Num1}
\Pr\left[ {\cal B}^c_1\cap {\cal B}^c_2\cap {\cal B}_4 \right] \leq 
\Pr\left[ {\cal B}^c_1\cap {\cal B}^c_2\cap {\cal B}_4 \ |\ {\cal Q}^c\right] +\Pr[{\cal Q}]
\end{equation}
For each block $B\in {\cal B}$,  the number of updates in the interval $[0,T]$ is
dominated by the binomial distribution with parameters $C'N$ and $1/N$.  Taking
large $d$, Chernoff's bound imply that
\begin{equation}\label{eq:ReductionofBc1Bc2B3ToBc1Bc2B5LargeABoundB4-Num2}
\Pr[{\cal Q}]\leq  {\textstyle \exp\left( -d^{4/5}\right)}.
\end{equation}

\noindent
Given ${\cal Q}^c$ and that the events ${\cal B}^c_1\cap {\cal B}^c_2$ hold, 
at time $T$, each  $w\in \tilde{N}$ is disagreeing with probability at most
$2d^{-1/5}$, regardless of the other vertices in $\tilde{N}$. That is, their number
is dominated by ${\tt Binomial }((1+\epsilon/6)d, 2d^{-1/5})$. From Chernoff bounds we get that following:
\begin{equation}\label{eq:ReductionofBc1Bc2B3ToBc1Bc2B5LargeABoundB4-Num3}
\Pr\left[ {\cal B}^c_1\cap {\cal B}^c_2\cap {\cal B}_4 \ |\ {\cal Q}^c\right] \leq   
\textstyle \exp\left( -\gamma^6 d\right).
\end{equation}

\noindent
Plugging  \eqref{eq:ReductionofBc1Bc2B3ToBc1Bc2B5LargeABoundB4-Num2} and
\eqref{eq:ReductionofBc1Bc2B3ToBc1Bc2B5LargeABoundB4-Num3} into
\eqref{eq:ReductionofBc1Bc2B3ToBc1Bc2B5LargeABoundB4-Num1}, we have
\begin{equation}\label{eq:ReductionofBc1Bc2B3ToBc1Bc2B5LargeABoundB4}
\Pr\left[ {\cal B}^c_1\cap {\cal B}^c_2\cap {\cal B}_4 \right] \leq  \textstyle 2\exp\left( -d^{4/5}\right).
\end{equation}

\noindent
Now, we focus on  $\Pr\left[ {\cal B}^c_1\cap {\cal B}^c_2\cap {\cal B}_5 \right]$.
At each step, the update chooses one vertex in $N^*(v)\setminus \tilde{N}$ with probability 
$|N^*(v)\setminus \tilde{N}|/N$. If such a vertex is chosen, then we have a new disagreement
with probability at most $2/d$.  Noting that  $|N^*(v)\setminus \tilde{N}|\leq (1+\epsilon/6)d$
and $T\leq C'N$, we get the following: The cardinality of $S_2$ is dominated by the binomial
distribution with parameters $C'N$ and $2(1+\epsilon/6)/N$.  Then, Chernoff's bound implies that
\begin{equation}\label{eq:ReductionofBc1Bc2B3ToBc1Bc2B5LargeABoundB5}
\Pr\left[ {\cal B}^c_1\cap {\cal B}^c_2\cap {\cal B}_5 \right] \leq  \textstyle \exp\left ( -\gamma^5 d \right).
\end{equation}

\noindent
Plugging into \eqref{eq:ReductionofBc1Bc2B3ToBc1Bc2B5LargeABoundB4} and 
\eqref{eq:ReductionofBc1Bc2B3ToBc1Bc2B5LargeABoundB5} into  
 \eqref{eq:ReductionofBc1Bc2B3ToBc1Bc2B45LargeA} we get the following:
For $a\geq \gamma^2/10$, we have that 
\begin{equation}
\Pr\left[ {\cal B}^c_1\cap {\cal B}^c_2\cap {\cal B}_3 \right] \leq \textstyle 3\exp\left( -d^{4/5}\right).
\label{eq:ReductionofBc1Bc2B3ToBc1Bc2B45LargeA}
\end{equation}

\noindent
For the case where $a<\gamma^2/10$ we work in the same manner. That is, it holds that
\begin{equation}
\Pr\left[ {\cal B}^c_1\cap {\cal B}^c_2\cap {\cal B}_3 \right] \leq 
\Pr\left[ {\cal B}^c_1\cap {\cal B}^c_2\cap {\cal B}_5  \right] 
\end{equation}
We bound $\Pr\left[ {\cal B}^c_1\cap {\cal B}^c_2\cap {\cal B}_5 \right]$ in the same manner as 
for  \eqref{eq:ReductionofBc1Bc2B3ToBc1Bc2B5LargeABoundB5}. That is, 
for $a< \gamma^2/10$, we have that 
\begin{equation}
\Pr\left[ {\cal B}^c_1\cap {\cal B}^c_2\cap {\cal B}_3 \right] \leq  \textstyle \exp\left ( -\gamma^5 d \right).
\label{eq:ReductionofBc1Bc2B3ToBc1Bc2B5SmallA}
\end{equation}
The lemma follows by plugging  \eqref{eq:lemma:GVsGStartNeighbors-B1Bound},
\eqref{eq:lemma:GVsGStartNeighbors-B2Bound},
\eqref{eq:ReductionofBc1Bc2B3ToBc1Bc2B5LargeABoundB4} and
\eqref{eq:ReductionofBc1Bc2B3ToBc1Bc2B5SmallA} into \eqref{eq:Basis4lemma:GVsGStartNeighbors}.
\end{proof}

\noindent
Given  the previous two lemmas, it is immediate to get the following two results.

\begin{corollary}\label{corollary:EscapeFromBall} 

In the same setting as Lemma \ref{lemma:Restriction2Cycle} the following is true:

Consider,  $G^*_v$ and the  block dynamics $(X^*_t)_{t\geq 0}$ and $(Y^*_t)_{t\geq 0}$.
Let   ${\cal T}=\{\tau_1, \tau_2, \ldots \}$ be the random times at which  $B_v$ is updated  in
$(X^*_t)$ during the  time interval  
${\cal I} =\left[\left \lfloor N\log(\gamma^{-3}) \right \rfloor, \left \lfloor CN  \right \rfloor \right ]$.
Assume  that $X^*_{0}$, $Y^*_{0}$  are such that  $X^*_{0}(w)\neq Y^*_{0}(w)$ for every $w\in N^*(v)$,
while $X^*_{0}(w)=Y^*_{0}(w)$ for every $w\notin N^*(v)$. Then, there is a coupling
such that for any $\tau\in {\cal T}$ we have
\begin{equation} 
\Pr\left[ \mathbf{I}\{{ \cal T} \neq \emptyset \}  \wedge \left(
\WDis_{\leq \tau} \not \subset \ball \left (N^*(v), d^{4/5} \right) \right )    \right]\leq \textstyle \exp\left( -d^{3/4} \right).
\end{equation}
\end{corollary}

\noindent 
Furthermore, we have the following:
\begin{corollary}\label{cor:GVsGStartConditioning}

In the same setting as Lemma \ref{lemma:Restriction2Cycle} the following is true:

Let $(X_t)_{t\geq 0}$ be  the block dynamics  on $G$. 
Also,  consider $G^*_v$ and the corresponding block dynamics $(X^*_t)_{t\geq 0}$.
Let   ${\cal T}=\{\tau_1, \tau_2, \ldots \}$ be the random times at which  $B_v$ is updated  in
$(X^*_t)$ during the  time interval  
${\cal I} =\left[\left \lfloor N\log(\gamma^{-3}) \right \rfloor, \left \lfloor CN  \right \rfloor \right ]$.
For any $\tau\in {\cal T}$, conditional that $X_{0}=X^*_{0}$,  there is a coupling such that 
\begin{equation}\label{eq:cor:GVsGStartConditioningA}
\Pr\left[ 
\mathbf{I} \{{\cal T}\neq \emptyset \} \wedge 
\left ( |\WDis_{\tau} \cap N^*(v)| \geq \gamma^2 d  \right)
 \right]\leq  \textstyle \exp\left( -d^{3/4} \right).
\end{equation}
\end{corollary}

\begin{lemma}\label{lemma:ExpctACondBt}

In the same setting as Lemma \ref{lemma:Restriction2Cycle} the following is true:

Let $(X_t)_{t\geq 0}$ be  the block dynamics  on $G$.  Also,  consider $G^*_v$ and the corresponding 
block dynamics $(X^*_t)_{t\geq 0}$. Let ${\cal T}=\{\tau_1, \tau_2, \ldots \}$ be the set of random times 
at which $B_v$ is updated during the time interval  $[N\log(\gamma^{-3}),  CN]$ in $(X^*_t)_{t\geq 0}$.
Given $X^*_{0}$, for any $\tau\in {\cal T}$ it holds that
\[
\Exp{ |\AvialColors_{X^*_{\tau}}(v)|  \  \mathbf{1} \{{\cal T}\neq \emptyset\} }  \geq  k e^{-\degree(v)/k} 
(1-50\gamma^3),
\]
where $\mathbf{1} \{{\cal T}\neq \emptyset\} $ is the indicator  of the event that ${\cal T}\neq \emptyset$.
\end{lemma}
\begin{proof}

Let $R(\tau ,v)\subseteq N^*(v)$ be the set of vertices which are updated at least once 
during the  time interval $[0, \tau]$.  Also, for each $w\in R(\tau ,v)$ let $\tau_w$ 
be the time of the last update of vertex $w$ up to time time $\tau$. Let $S_w$ be the set of available 
colors for vertex $w\in R(\tau ,v)$  when it is updated at time $\tau_w$. For every $j \in [k]$ let 
$\alpha_{w,j}=1$ if $j\in S_w$ and $\alpha_{w,j}=0$, otherwise.

Corollary \ref{corollary:EscapeFromBall}, combined with standard {\em disagreement percolation}
implies the following:
 
 \begin{claim}\label{claim:FactorizationOfProbs1}
 For every $j\in [k]$  let $\mathbf{I}_{\{ j \}}$ be the event that  the color $j$ is not used by any
vertex $w\in R(\tau,v)$  at time $\tau$. Then, given  $X^*_{0}$, for any $j\in [k]$ it holds that
\begin{equation}
\textstyle \left | \Pr\left [\mathbf{I}_{\{j\}} \wedge \mathbf{I}_{\{{\cal T}\neq \emptyset \}}   \right] - 
\Exp{   \mathbf{1}\{{\cal T}\neq \emptyset \} \ \prod_{w\in R(\tau,v)}\left(1-|S_w|^{-1} \right)^{\alpha_{w,j}}  }  \right|\leq  
\textstyle  \exp\left( -d^{3/4} \right),
\end{equation}
where the expectation on the product is w.r.t. $R(\tau,v)$ and $S_w$. 
\end{claim}

\noindent
Let  $Q_\tau$ be the number of colors that are not used by any vertex in $R(\tau,v)$ at time $\tau$. 
Noting that
\[
\Exp{  Q_\tau   \   \mathbf{1} \{{\cal T} \neq \emptyset \} }  = \textstyle 
\sum^k_{j=1}\Pr[\mathbf{I}_{\{ j\} }\wedge \mathbf{I}_{\{ {\cal T}\neq \emptyset \}}    ] \nonumber \\
\]
we have that
\begin{eqnarray}
\Exp {Q_\tau   \   \mathbf{1} \{{\cal T} \neq \emptyset \} }
&\geq & \textstyle \Exp{   \mathbf{1} \{{\cal T}\neq \emptyset \} \   \sum^k_{j=1} \ \prod_{w\in R(\tau,v)}\left(1-|S_w|^{-1} \right)^{\alpha_{w,j}}  }  - 2 d e^{-d^{3/4}}  \qquad \mbox{[From Claim \ref{claim:FactorizationOfProbs1}]}\nonumber \\
&\geq & \textstyle  k\cdot  \Exp{  \mathbf{1} \{{\cal T}\neq \emptyset \}   \prod^k_{j=1} \ \prod_{w\in R(\tau,v)}\left(1-|S_w|^{-1} \right)^{\alpha_{w,j}/k}   }  - 2 d e^{-d^{3/4}},  \nonumber \\
&\geq & \textstyle 
k\cdot  \Exp{  \mathbf{1} \{{\cal T}\neq \emptyset \}  \prod_{w\in R(\tau,v)}  \prod^k_{j=1}  \left(1-|S_w|^{-1} \right)^{\alpha_{w,j}/k}  }  - 2 d e^{-d^{3/4}}
\end{eqnarray}
the second line uses the arithmetic-geometric mean inequality. Since  $\sum_{j}\alpha_{w,j}=|S_w|$ we get that
\begin{eqnarray}
\Exp{ Q_\tau   \   \mathbf{1} \{{\cal T} \neq \emptyset \}  }
&\geq & \textstyle  k\cdot  \Exp{  \mathbf{1} \{{\cal T}\neq \emptyset \}   \prod_{w\in R(\tau,v)}   \left(1-|S_w|^{-1} \right)^{|S_w|/k} }  - 2 d e^{-d^{3/4}}
 \nonumber  \\
&\geq & \textstyle  k\cdot  \Exp{ \mathbf{1}\{{\cal T}\neq \emptyset \}   \prod_{w\in R(\tau,v)}   \left(1-(k-\TDeg )^{-1} \right)^{(k-\TDeg)/k}   }  - 2 d e^{-d^{3/4}}, 
\nonumber 
\end{eqnarray}
where in the last derivation we use the fact that for any $w\in R(\tau,v)$ it holds 
$(1-|S_w|^{-1})^{|S_{w}|}\geq (1-(k-\TDeg))^{k-\TDeg}$. Finally, using the observation 
that $|R(\tau,v)|\leq \degree(v) $, we get that
\begin{eqnarray} 
\Exp{ Q_\tau    \mathbf{1} \{{\cal T} \neq \emptyset \}  } &\geq   &
k \left(1-\frac{1}{k-\TDeg }\right)^{\degree(v)(k-\TDeg)/k} \Exp{  \mathbf{1}\{{\cal T} \neq \emptyset \} }
- 2 d e^{-d^{3/4}} , \nonumber \\
&\geq   & k e^{-\degree(v)/k}(1-2\gamma^3),   \label{eq:LowerBound4ExpctQt}
\end{eqnarray}
where the last inequality follows from the, easy to derive, bound that 
$\Exp{  \mathbf{1} \{{\cal T} \neq \emptyset \} } =1-\gamma^{3}$.

Let $U_1$ be the number of  vertices in $N^*(v)\setminus B_v$ which are not updated in the time interval
$[0,\tau]$. Each vertex in  $N^*(v)\setminus B_v$ is not updated with probability less than $\gamma^{3}$ 
independently  of the other vertices. Since  $|N^*(v)\setminus B_v|\leq \TDeg$,  $U_1$ is dominated by the  binomially 
distribution with parameters parameters $\TDeg$ and $\gamma^{3}$. 

Let ${\cal U}$, ${\cal A}$ be the events,
 $U_1< 15\gamma^3 \TDeg \wedge \mathbf{I} \{{\cal T}\neq \emptyset \}$
and  $U_1 \geq  15 \gamma^3 \TDeg \wedge \mathbf{I} \{{\cal T}\neq \emptyset \}$, respectively.
 From   Chernoff's bounds we get that
\begin{equation} \label{eq:UnUpdatedBound1}
\Pr[{\cal A}] \leq \exp\left( -10\gamma^3 \TDeg \right).
\end{equation}

\noindent
Also, we have that
\begin{eqnarray}
\Exp{ Q_\tau   \  \mathbf{1}\{{\cal T}\neq \emptyset \} } &=& 
\ExpCond{ Q_\tau  \  \mathbf{1} \{{\cal T}\neq \emptyset \} }{ {\cal A}  } \  \Pr[{\cal A} ]  
+ \ExpCond{ Q_\tau  \  \mathbf{1} \{{\cal T}\neq \emptyset \}}{  {\cal U} }  \  \Pr[ {\cal U}  ]  \nonumber \\
&\leq & k  \Pr[ {\cal A}  ]  
+ \ExpCond{ Q_\tau  \  \mathbf{1} \{{\cal T}\neq \emptyset \}}{  {\cal U} }  \  \Pr[{\cal U}  ] 
\hspace{3cm} \mbox{[since $Q_{\tau}\leq k$]} \nonumber\\
&\leq & k \exp\left(-10 \gamma^3\TDeg \right)  + \ExpCond{ Q_\tau  \  \mathbf{1} \{{\cal T}\neq \emptyset \}}{  {\cal U} }, \nonumber
\end{eqnarray}
in the third derivation we use \eqref{eq:UnUpdatedBound1} and the fact that $ \Pr[ {\cal U}  ]\leq 1$.
The above inequality implies that
\begin{eqnarray}
\ExpCond{ Q_\tau  \  \mathbf{1} \{{\cal T}\neq \emptyset \} }{ {\cal U}  }  & \geq  &
\Exp{ Q_\tau  \  \mathbf{1} \{{\cal T}\neq \emptyset \}  } -k\exp\left(-10\gamma^3 \TDeg \right)  \nonumber \\
&\geq & k e^{-\degree(v)/k}(1-2\gamma^3) -k\exp\left(-10\gamma^3 \TDeg \right) 
\qquad \mbox{[from \eqref{eq:LowerBound4ExpctQt}]}
\nonumber \\
&\geq & k e^{-\degree(v) /k}(1-3\gamma^3).
\end{eqnarray}

\noindent
Since  the vertices in $N^*(v)\setminus R(\tau,v)$ can use at most $U_1$ many colors, 
we have that  
\[
\ExpCond{ |\AvialColors_{X^*_{\tau}}(v)| \  \mathbf{1} \{{\cal T}\neq \emptyset \} }{ {\cal U} }  \geq  
k e^{-\degree(v)/k}(1-30\gamma^3).
\]
The lemma follows by noting that since $|\AvialColors_{X^*_\tau}(v)|  \cdot \mathbf{1}\{{\cal T}\neq \emptyset \} \geq 0$, 
we have  that 
\[
\Exp{ |\AvialColors_{X^*_{\tau}}(v)|  \  \mathbf{1} \{{\cal T}\neq \emptyset \} }  \geq 
\Pr[{\cal U}] \cdot \ExpCond{ |\AvialColors_{X^*_t}(v)|   \mathbf{1} \{ {\cal T} \neq \emptyset \}}{ {\cal U} },
\]
while Chernoff's bounds give  $\Pr[{\cal U}]\geq 1-2\exp(-10 \gamma^3 \TDeg )$.
\end{proof}

\begin{lemma}[Uniformity for $G^*_v$]\label{lemma:UniformityGStar176}

In the same setting as Lemma \ref{lemma:Restriction2Cycle} the following is true:

Consider $G^*_v$ and let the block dynamics $(X^*_t)_{t\geq 0}$.
Let ${\cal T}=\{\tau_1, \tau_2, \ldots\}$ be the random times at which $B_v$ is updated
during the time iterval ${\cal I}$.  For any $\tau \in {\cal T}$ and any $X^*_{0}$ the following holds:
\begin{equation}
\Pr\left[ \mathbf{I} \{{\cal T}\neq \emptyset \} \wedge   
\left( 
 \AvialColors_{X^*_\tau}(v)|    \leq (1-100\gamma^3) k \exp\left(-\degree(v)/k \right) 
\right ) 
 \right]\leq  \exp\left( -\gamma^4\TDeg  \right).  \label{Aeq:lemma:UniformityGStar176}
\end{equation}
\end{lemma}

\begin{proof}
First we focus on \eqref{Aeq:lemma:UniformityGStar176}.  Using the fact that
for any two events $A,B$ it holds $\Pr[A\wedge B]\leq \Pr[B|A]$, 
for \eqref{Aeq:lemma:UniformityGStar176} it suffices to show that
\begin{equation}\label{Aeq:lemma:UniformityGStar176WithCond}
 \Pr\left[  
\left( 
{|\AvialColors_{X^*_\tau}(v)|}   \leq (1-100\gamma^3)\  k\exp\left(-\degree(v)/k \right)  
\right )  \ |\  \mathbf{I} \{{\cal T}\neq \emptyset \} 
 \right] \leq  \exp\left( -\gamma^4\Delta \right).
\end{equation}

\noindent
Let $\mu= \ExpCond{ |\AvialColors_{X^*_\tau}(v)| }{  \mathbf{I} \{{\cal T}\neq \emptyset \} }$.
We have that
\begin{eqnarray}
\mu &=& \frac{\Exp{  |\AvialColors_{X^*_{\tau}}(v)|  \  \mathbf{1} \{{\cal T}\neq \emptyset\} }}
{ \Pr[ \mathbf{I} \{{\cal T}\neq \emptyset \} ]}  
\ \geq \ \Exp{   |\AvialColors_{X^*_{\tau}}(v)|  \  \mathbf{1} \{{\cal T}\neq \emptyset\}  }
\hspace{.625cm}\mbox{[since $\Pr[ \mathbf{I} \{{\cal T}\neq \emptyset \} ]\leq 1$]} \nonumber \\
&\geq &   k e^{-\degree(v)/k}  (1-50\gamma^3).
 \hspace{5.675cm} \mbox{[from Lemma \ref{lemma:ExpctACondBt}]} 
 \label{Aeq:lemma:UniformityGStar176MuLB}
\end{eqnarray}

\noindent
Using Hoeffding's inequality we get the following:   for any $\eta>0$ we have that
\[
\Pr \left [|\AvialColors_{X^*_t}(v)| -\mu <\eta\; |\;  \mathbf{I}_{ \{{\cal T}\neq \emptyset \}} \right ]\leq \exp(-\eta^2/(2\degree(v))).
\]
Note that we always have  $|\AvialColors_{X^*_t}(v)|>k-\TDeg\geq (1-\alpha)\TDeg$ and $\degree(v)\leq \TDeg$ since $v\in \ImpVrtx$.
Setting $\eta=\gamma \mu$  we get 
\[
\Pr\left[|\AvialColors_{X^*_t}(v)|  <(1-\gamma)\mu \; |\;  \mathbf{I} \{{\cal T}\neq \emptyset \} \right ]\leq \exp(- \gamma^2(1-\alpha)^2\TDeg/2).
\]
\noindent
Plugging \eqref{Aeq:lemma:UniformityGStar176MuLB} into the above tail bound we
get \eqref{Aeq:lemma:UniformityGStar176WithCond}.
The lemma follows.
\end{proof}

\begin{proof}[Proof of Theorem \ref{thrm:Uniformity1.76}]

We start by assuming that $v\in \ImpVrtx$.
Let ${\cal S}=\{ t_1, t_2, \ldots \} $  be the set of (random) times  when the block $B_v$ is updated 
in  $(X_t)$.
Let ${\cal T}=\{ \tau_{1}, \ldots, \tau_{\ell} \} = {\cal S}\cap {\cal I}$. We follow the convention that 
$\tau_j\leq \tau_{j+1}$.

Let $\mathbf{J}_1, \ldots, \mathbf{J}_{\ell}$ be  such that $\mathbf{J}_j=(\tau_{j}, \tau_{{j+1}})$, 
where $\tau_{\ell+1}=\mathbf{I}_2$.  We let $\mathbf{J}_0=[\mathbf{I}_1 ,\min\{\mathbf{I}_2, \tau_{1}\} )$, 
where we follow the convention that $\tau_1=\infty$ if ${\cal T}=\emptyset$. 

The result  follows by showing the following two inequalities and taking a union bound.

\begin{equation}
\Pr\left[ \mathbf{1} \{{\cal T}\neq \emptyset \} \wedge  
\left(
 \exists t\in \bigcup^{\ell}_{i=1} \mathbf{J}_{i}\;s.t. \;
 {|\AvialColors_{X_t}(v)|}   \leq  (1-\gamma)\ k\exp\left (-\degree(v)/k  \right)
 \right)
 \right]
 \leq  3d^{3}  \exp\left( -d^{3/4} \right) \qquad \qquad \label{eq:thrm:Uniformity176TargetA} 
 \end{equation}
 and
 \begin{equation}
\Pr\left[ \exists t\in \mathbf{J}_{0}\;s.t. \;
\left(  {|\AvialColors_{X_t}(v)|}/{k} \right)^{\mathbf{\Large 1}_{\{v, t\}} }   \leq (1-20\gamma)\exp\left (-\degree(v)/k  \right)
 \right]  \leq   \exp\left( -d^{3/4}\right), \label{eq:thrm:Uniformity176TargetB}
\end{equation}

\noindent
where $\mathbf{I} \{{\cal T}\neq \emptyset \}$ is the event that ${\cal T}$ is non-empty.

Noting that both $\ell$, the cardinality of ${\cal T}$ is a random variable. In particular, $\ell$ is dominated
by the binomial distribution with parameters $CN$ and $1/N$. Applying Chernoff's bounds
we get that
\begin{equation} \label{eq:FromProofThrm:Uniformity176}
\Pr \left [\ell \geq d^2 \right] \leq   \exp\left(- d^2\right).
\end{equation}

\noindent
Let  ${\cal E}_j$ be  the event that at time $\tau_j$,   we have that ${|\AvialColors_{X_{\tau_j}}(v)|}  > (1-12\gamma^2)\ k\exp\left(-\degree(v)/k \right) $.
We are going to show that 
\begin{equation}\label{eq:FirstTargetThrmUniformity176}
\textstyle  \Pr  \left [  \mathbf{I} \{{\cal T}\neq \emptyset \}  \wedge  
\left( 
\bigcup^{\ell}_{i=1}\bar{\cal E}_{i} 
\right)
\right] \leq  2d^2 \exp\left( -d^{3/4} \right).
\end{equation}

\noindent
Consider   $\tau_{j}\in {\cal T}$.  Also, consider $G^*_v$ and the corresponding block dynamics $(X^*_t)$.
Assume that $(X_t)_{t\geq 0}$ and $(X^*_t)_{t\geq 0}$ are such that $X_0=X^*_0$. 
Using  \eqref{Aeq:lemma:UniformityGStar176},  in Lemma \ref{lemma:UniformityGStar176}, we get that 
\[
\Pr\left[ 
\mathbf{1} \{{\cal T}\neq \emptyset \} \wedge
\left( 
{|\AvialColors_{X^*_{\tau_{j}}}(v)|}  \leq (1-100\gamma^3)\ k\exp\left(-\degree(v)/k \right)
\right)
 \right]  \leq   \textstyle  \exp \left( -\gamma^4\TDeg \right).
 \]
Combining the above with \eqref{eq:cor:GVsGStartConditioningA}, in  Corollary \ref{cor:GVsGStartConditioning},
we get that \begin{equation}\label{eq:UniformityForTs+1}
\Pr[\mathbf{1} \{{\cal T}\neq \emptyset \} \wedge \bar{\cal E}_{j}]  \leq \textstyle \exp\left( -d^{3/4}\right).
\end{equation}

\noindent
Using \eqref{eq:UniformityForTs+1} we get the following:
\begin{eqnarray}
\Pr \left[ \mathbf{I} \{{\cal T}\neq \emptyset \}  \wedge  \left( \cup^{\ell}_{i=1}\bar{\cal E}_{i} \right) \right] 
&\leq & 
\Pr \left [\mathbf{I}  \{{\cal T}\neq \emptyset \} \wedge \left( \cup^{\ell}_{i=1}\bar{\cal E}_{i} \right)\ |\   \ell < d^2 \right]
+\Pr\left [ \ell \geq d^2 \right]
\nonumber \\
&\leq & \sum^{ d^2-1}_{i=1}
\Pr\left [ \mathbf{I} \{{\cal T}\neq \emptyset \} \wedge  \bar{\cal E}_i  \ |\  \ell <d^2  \right] 
+\Pr\left [ \ell \geq d^2 \right]
\hspace{1.25cm} \mbox{[union bound]}\nonumber\\
&\leq & \sum^{ d^2-1}_{i=1}
\frac{\Pr[\mathbf{I} \{{\cal T}\neq \emptyset \}  \wedge \bar{\cal E}_{i}]}{1-\exp\left (-d^2\right )} +
\exp\left( -d^2\right)
\hspace{3cm} \mbox{[from  \eqref{eq:FromProofThrm:Uniformity176}]}\nonumber\\
&\leq &\textstyle 2d^2\exp\left( -d^{3/4}\right).
\hspace{6.1cm} \mbox{[from  \eqref{eq:UniformityForTs+1}]} \nonumber
\end{eqnarray}
The above derivations shows that \eqref{eq:FirstTargetThrmUniformity176} is indeed true.

Consider the time interval $\mathbf{J}_i$.  W.l.o.g. assume that  $|\mathbf{J}_i|>\gamma^3 N$.  
Consider a partition of  $\mathbf{J}_i$ into subintervals each of length  (at most) 
$\gamma^3 N$, where the last part can be of smaller length.  Let 
$\mathbf{J}_i(j)=(t_{i,j}, t_{i,j+1})$  be the $j$-th part in this partition, while we have $t_{i,0}=\tau_i$.

Let ${\cal E}_{i}(j)$ be the event that $\frac{|\AvialColors_{X_{t_{i,j}}}(v)|}{k}  > (1-12\gamma^2)\exp\left(-\degree(v)/k \right) $.
For any $0 \leq j\leq \lceil C\gamma^{-3} \rceil$,  we are going to show that 
\begin{equation}\label{eq:SecondTargetThrmUniformity176}
\Pr[\mathbf{I}_{ \{{\cal T}\neq \emptyset \}} \wedge   \bar{\cal E}_{i}(j) ] \leq 
\textstyle \exp\left(-d^{3/4} \right).
\end{equation}

\noindent
Eq. \eqref{eq:UniformityForTs+1} implies that the above  is true for $j=0$.
Consider   $1\leq j \leq \lceil C\gamma^{-3} \rceil$. Consider, also,  $G^*_v$ and the 
corresponding block dynamics $(X^*_t)_{t \geq 0}$. Assume that $(X_t)_{t\geq 0}$ and $(X^*_t)_{t\geq 0}$ 
are such that  $X_0=X^*_0$.   Using  Lemma \ref{lemma:UniformityGStar176} for $(X^*_t)$ we get that 
\[
\textstyle \Pr\left[ \mathbf{1} \{{\cal T}\neq \emptyset \} \wedge 
\left( {|\AvialColors_{X^*_{t_{i,j}}}(v)|}   \leq (1-100\gamma^3) \ k\exp\left(-\degree(v)/k \right)  \right)
 \right]  \leq   \exp\left(-\gamma^4\TDeg \right).
 \]
 Combining the above  Corollary  \ref{cor:GVsGStartConditioning},  we get that
\begin{equation} \label{eq:Tail4Eij} 
\Pr[\mathbf{1} \{{\cal T}\neq \emptyset \} \wedge  \bar{\cal E}_i(j)]  \leq \textstyle \exp\left( -d^{3/4}\right),
\end{equation}
for $1\leq j \leq  C\gamma^{-3}$. 
The above implies that \eqref{eq:SecondTargetThrmUniformity176} is indeed true.

Let ${\cal R}^i_j$ be the event that there is some $s\in \mathbf{J}_i(j)$ such that 
$\frac{|\AvialColors_{X_{s}}(v)|}{k} > (1-14\gamma^2)\exp\left(-\degree(v)/k \right)$.
Some vertex $w \in  N(v)\setminus B_v$ is updated  in a transition of the chain
 with probability  at most   $\TDeg/N$. 
Note that the vertices in $N(v)\setminus B_v$ belong to different blocks. That is, 
an update of vertex in $N(v)\setminus B_v$ updates only a single vertex.

Chernoff's bounds imply that  with  probability at least $1-\exp\left( -\gamma^3 d  \right)$,  
the number of updates of vertices  in $N(v)\setminus B_v$    during  $\mathbf{J}_i(j)$ is at most  
$\TDeg\gamma^2$.  By definition, during $\mathbf{J}_i(j)$ the vertices in  $N(v)\cap B_v$ are 
not updated.

Since  changing any $\TDeg \gamma^2$ vertices in $N(v)$ can only change 
the number of available colors for $v$ by at most $\TDeg\gamma^2$, 
we get the following:  With probability at least $1-\exp\left( -\gamma^3 \TDeg   \right)$, 
during the time period  $\mathbf{J}_i(j)$ the ratio  $|\AvialColors_{X^*_{t_{i}}}(v)|/k$ does not  
change by more than  $\gamma^{2}/1.5$. Then, we get that 
\begin{equation}\label{eq:Tail4Rij}
\Pr[ \bar{R}^i_j \ |\  \mathbf{I} \{{\cal T}\neq \emptyset \}  \wedge  {\cal E}_i(j) ]\leq \exp(-\gamma^3\TDeg).
\end{equation}
We have that 
\begin{eqnarray}
\Pr\left [ \mathbf{I} \{{\cal T}\neq \emptyset \} \wedge \left( \cup_{j}\bar{R}^i_j\right) \right]  
& \leq & 
\sum^{ \lceil C \gamma^{-3} \rceil }_{j=0}
\Pr[\mathbf{I} \{{\cal T}\neq \emptyset \} \wedge  \bar{R}^i_j    ] 
\hspace{.7cm}\mbox{[union bound]} \nonumber\\
& \leq &
\sum^{ \lceil C \gamma^{-3} \rceil  }_{j=0}
\left(
 \Pr[ \mathbf{I} \{{\cal T}\neq \emptyset \} \wedge \bar{\cal E}_i(j) ] +
\Pr[  \bar{R}^i_j \; |\; \mathbf{I} \{{\cal T}\neq \emptyset \} \wedge  {\cal E}_{i}(j)] \right) \nonumber \\
&\leq & \textstyle d \exp\left( -d^{3/4}\right).
 \label{eq:TailOfRiUniformity}
\end{eqnarray}
The last derivation follows from \eqref{eq:Tail4Rij} and \eqref{eq:Tail4Eij}.
Let  ${\cal R}_i=\bigcup_j{\cal R}^i_j$. Note that the event insider the probability term in  
\eqref{eq:thrm:Uniformity176TargetA} is equivalent to the event 
$\mathbf{I}_{ \{{\cal T}\neq \emptyset \}}  \wedge \left( \cup_{i} {\cal R}_i \right)$. 
It holds that
\begin{eqnarray}
\Pr\left [  \mathbf{I}  \{{\cal T}\neq \emptyset \}   \wedge \left( \cup_{i} {\cal R}_i \right) \right] &\leq &
\Pr\left[ \mathbf{I} \{{\cal T}\neq \emptyset \}   \wedge \left( \cup_{i} {\cal R}_i \right)  \ |\  \ell <d^{2} \right]+
\Pr\left[\ell \geq  d^2  \right] \nonumber \\
&\leq & \sum^{d^2-1}_{i=1}
\Pr \left[ \mathbf{I}  \{{\cal T}\neq \emptyset \}  \wedge  {\cal R}_i    \ |\  \ell <d^2 \right ]+
\Pr\left[\ell \geq d^2 \right ] \nonumber \\
&\leq & 2\sum^{d^2-1 }_{i=1}
 \Pr\left [ \mathbf{I}  \{{\cal T}\neq \emptyset \}   \wedge  {\cal R}_i \right ] +
\exp\left( -d^2 \right)  \qquad \mbox{[from  \eqref{eq:FromProofThrm:Uniformity176}]}\nonumber \\
&\leq &\textstyle  2d^3\exp\left( -d^{3/4} \right), \nonumber
\end{eqnarray}
where in the last derivation we used \eqref{eq:TailOfRiUniformity}. 
Eq. \eqref{eq:thrm:Uniformity176TargetA}, follows.

It remains to  show that \eqref{eq:thrm:Uniformity176TargetB} is indeed true.
Recall that $t_1$ is the time the block dynamics updates $B_v$ for first time. 
We consider cases for $t_1$.
The first case is when $t_1>\mathbf{I}_2$. Then, \eqref{eq:thrm:Uniformity176TargetB} is trivially true, i.e.,
there is no update of $B_v$ during the time interval ${\cal I}$. If $t_1\in {\cal I}$, i.e., the first update of 
$B_v$ happened after the beginning of the time interval ${\cal I}$,  then by definition it follows that the 
block $B_v$  is not updated during $\mathbf{J}_0$.   This implies that \eqref{eq:thrm:Uniformity176TargetB} is  true.

The less trivial case is  when $t_1 < \mathbf{I}_1$, i.e., there was an update of block $B_v$ before the time 
period ${\cal I}$ had started. Let $t'=\mathbf{I}_1$. Since we assume that $t_1<t'$,   
Lemma \ref{lemma:UniformityGStar176} and  \eqref{eq:FromProofThrm:Uniformity176} imply  that
\[
\Pr\left[ 
\left(  {|\AvialColors_{X_{t'}}(v)|} \right)   \leq (1-12\gamma^2)\ k\exp\left (-\degree(v)/k  \right)
 \right]\leq  \textstyle \exp\left( -d^{3/4}\right).
\]
Furthermore, using a ``covering argument" very similar to that we used before, we prove that
$\frac{|\AvialColors_{X_{t}}(v)|}{k}  \leq (1-20\gamma^2)\exp\left (-\degree(v)/k  \right)$,
for any $t\in \mathbf{J}_0$, with probability $\leq \exp(-d^{3/4})$,
as promised.  

The case where $v$ is an internal vertex is almost direct. Updating the block of $v$, we have the following: 
conditional on the configuration of $v$ and the vertices at distance $2$ from $v$ the expected number of
available colors is at least $k\exp\left (-\degree(v)/k \right)$. This bound follows by using arguments very similar to those
we have in the proof of Lemma \ref{lemma:ExpctACondBt}. Then, the tail bound on the available colors follows by
using Azuma's inequality, similarly to the proof of Lemma \ref{lemma:UniformityGStar176}. The derivations are
very similar to the aforementioned results for this reason we omit them.
The theorem follows.
\end{proof}

\section{Proof of Theorem \ref{thrm:BlockUpdtCovergent}}\label{sec:thrm:BlockUpdtCovergent}

Let $B_1, B_2, \ldots, B_{s}$  be the  blocks that are  adjacent to $u^*$.
Recall that each of these blocks is a tree with at most one extra edge.  For each $B_j$, 
let $\mathbf{T}^j$ be the maximal sub-block of $B_j$ which   contains all the vertices  that are 
reachable from $v$ through a path inside $B_j$ that does not uses any edges of the cycle
of $B_j$.  Note that $\mathbf{T}^j$  is always a tree. The root of $\mathbf{T}^j$ is the vertex which is
 adjacent to $u^*$.  For each $B_j$, there is only one vertex  such vertex.

If the block $B_j$ is a  tree, then $B_j$ and  $\mathbf{T}^j$ are identical. Otherwise, if $B_j$ is unicyclic then 
 what remains outside $\mathbf{T}^j$ is  the cycle and  the subtrees that hang from the cycle.
For $B_j$ that contains the cycle $C$, and vertex $x$ which is adjacent to a vertex in $C$,  we define the subtree
$\mathbf{T}^j_x$ that contains $x$ and all the vertices in $B_j$ which are reachable from $x$ through a path inside
 $B_j$ which does not uses edges of $C$, e.g., see Figure \ref{fig:DisBlockVsTj}
 
 \begin{figure}
	\centering
		\includegraphics[height=6.3cm]{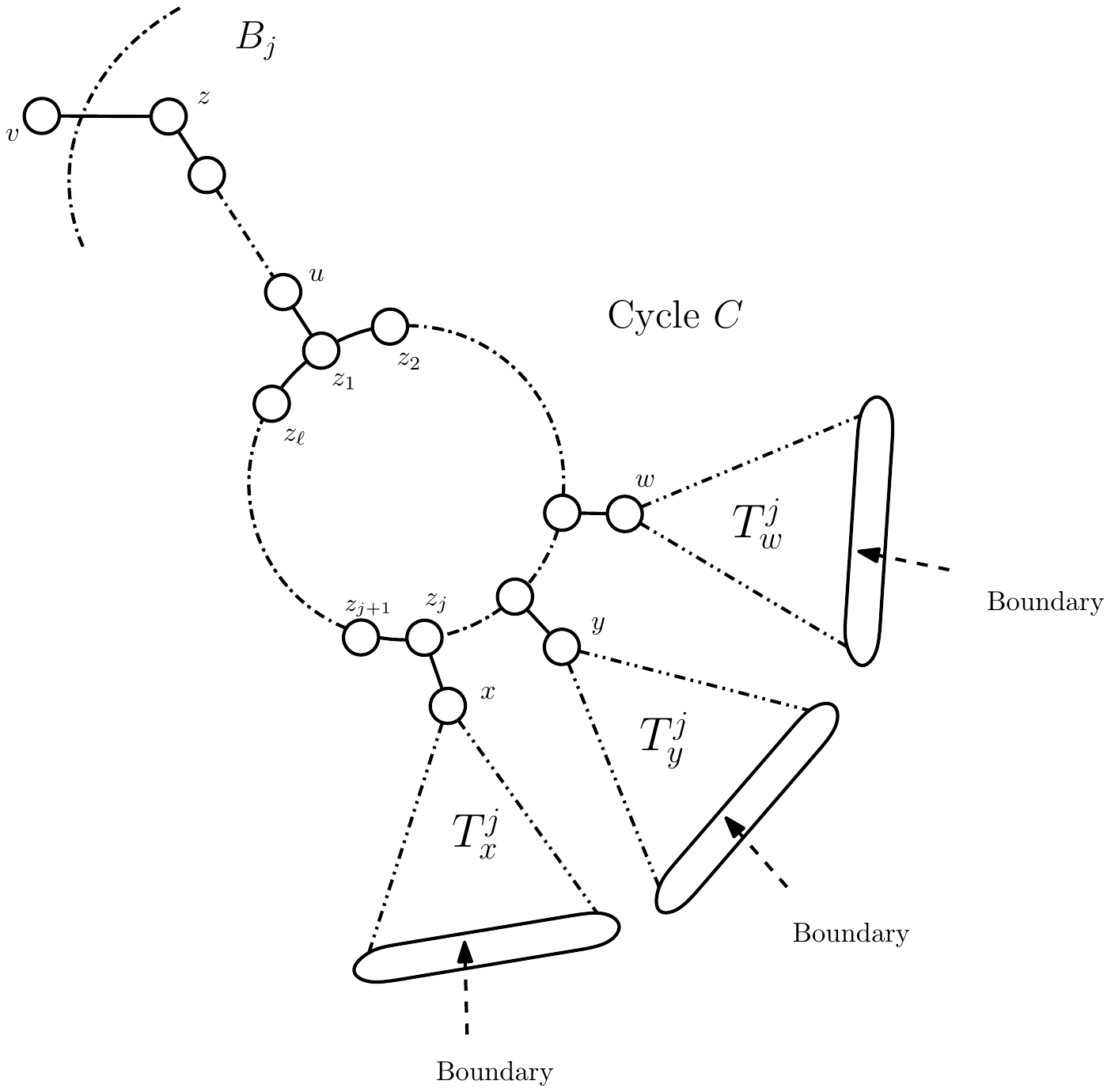}
		\caption{$\mathbf{T}^j$ does not contain $C$ and the subtrees than hang from $z_2, \ldots, z_{\ell}$.}
	\label{fig:DisBlockVsTj}
\end{figure}

For   $\Lambda\subset V$ for which  there exists $B_j$ such that $\Lambda\subseteq B_j$ let
$$
{\cal R}(\Lambda, X_t, Y_t)=n^2\sum_{z\in \Lambda\cap \ImpVrtx}\degree_{out}(z) \indicator{z\in X_t\oplus Y_t}.
$$
For any $w\in \outBound \Lambda$ we let
\[ 
Q_w(\Lambda) =  \ExpCond{  {\cal R}(\Lambda, X_{t+1}, Y_{t+1})  \mathbf{1}\{{\cal E}\}}{  X_{t+1} (w)\neq Y_{t+1}(w), \ 
X_{t},Y_{t}, \   B_j \textrm{ is updated at time $t+1$} }.
\]

For introducing the following concepts, consider the block in  Figure \ref{fig:DisBlockVsTj}.
We let the event $\mathbold{A}_j=$``The block $B_j$ contains  cycle $C$". For each vertex $w\in B$, 
we  let the event  $\mathbold{D}_w=$``From $u^*$,  there is a path of disagreement  in $B_j$ that reaches $w$".  
The linearity of expectation yields
\begin{eqnarray} \label{eq:LinearityOfQvBVsQvT+Rest}
Q_{u^*}(B_j)\leq Q_{u^*}(\mathbf{T}^j)+
\mathbf{1}\{\mathbold{A}_j \} \left ( 
\Pr[\mathbold{D}_{u}] \cdot Q_{u}(C)+  \sum_{z_i\in C\backslash\{z_1\}}\sum_{x\in N(z_i)\backslash C} \Pr[\mathbold{D}_{z_i}]
\  Q_{z_i}(\mathbf{T}^j_x)
 \right),
\end{eqnarray}
where $u$ is the only vertex in $\mathbf{T}^j$ which is adjacent to the cycle $C$ and it is assumed that
$u$ is adjacent to the vertex $z_1\in C$ (see Figure \ref{fig:DisBlockVsTj}). With  \eqref{eq:LinearityOfQvBVsQvT+Rest} 
we break the vertices of $B_j$ which contributed to $Q_v(B_j)$ into groups. That is, the vertices in $\mathbf{T}^j$, the vertices in 
the cycle $C$ and, finally, the trees that hang from $z_2, \ldots, z_{\ell}$, respectively.

The theorem follows by plugging the bounds from Propositions \ref{prop:Bound4QTPlustUniformity} and 
\ref{prop:BoundHangingTreesPlusUniformity} into \eqref{eq:LinearityOfQvBVsQvT+Rest} and 
Note that we have that
\begin{equation}\label{eq:Target4:thrm:BlockUpdtCovergent}
\ExpCond{
\left( 
\EdgeBlockWeight(X_{t+1}, Y_{t+1})
- \EdgeBlockWeight (X_{t}, Y_{t})
\right) 
\mathbf{1}\{ {\cal E} \}}{ X_t, Y_t, \ B\  \textrm{is updated at $t+1$}
}
\leq   Q_{u^*}(B_j)+n.
\end{equation}
The theorem will follow by using \eqref{eq:LinearityOfQvBVsQvT+Rest} to bound $Q_{u^*}(B_j)$.
To that end, we use  the following  results, whose proofs appear in 
Sections \ref{sec:prop:Bound4QTPlustUniformity} and \ref{sec:prop:BoundHangingTreesPlusUniformity}, respectively.

\begin{proposition} \label{prop:Bound4QTPlustUniformity}
Under the assumptions of Theorem \ref{thrm:BlockUpdtCovergent},  for any block $B_j$,  adjacent to $v$,  we have 
$$
Q_{u^*}(\mathbf{T}^j )\leq n^2(1-  2\epsilon/7).
$$
\end{proposition}

\begin{proposition}\label{prop:BoundHangingTreesPlusUniformity}
Under the assumptions of Theorem \ref{thrm:BlockUpdtCovergent}. for any block $B_j$, 
incident to $v$, we have 
$$
\textstyle   \mathbf{1}\{\mathbold{A}_j \} 
\left ( 
\Pr[\mathbold{D}_{u}] \cdot Q_{u}(C)+  \sum_{z_i\in C\backslash\{z_1\}} \sum_{x\in N_{z_i}\backslash C}
 \Pr[\mathbold{D}_{z_i}]\cdot Q_{z_i}(\mathbf{T}^j_x)  
 \right )
 \leq n^2(\log\log d)^{-|C|/10},
$$
(see in Figure \ref{fig:DisBlockVsTj} for the placement of the vertices above).
\end{proposition}

\noindent
Plugging the bounds from Propositions \ref{prop:Bound4QTPlustUniformity} and 
\ref{prop:BoundHangingTreesPlusUniformity} into \eqref{eq:LinearityOfQvBVsQvT+Rest} we get the
desirable bound for $Q_{u^*}(B_j)$. The theorem follows by using \eqref{eq:Target4:thrm:BlockUpdtCovergent}.

\subsection{Proof of Proposition \ref{prop:Bound4QTPlustUniformity}}\label{sec:prop:Bound4QTPlustUniformity}

So as to bound $Q_{u^*}(\mathbf{T}^j)$ we consider the quantities $Q^a$ and 
$Q^b$ defined as follows: Let ${T}_{a}=\mathbf{T}^j\cap \ball(v,r)$,
where $r=(15\log d)/\log(1+\epsilon/10)$. Similarly, let  ${T}_b=\mathbf{T}^j\setminus \ball(v, r)$.
The quantity $Q^a $ includes the  contribution on $Q_{u^*}(\mathbf{T}^j)$ from the 
vertices in the subtree ${T}_a$.   $Q^b$ includes the 
contribution on $Q_v(\mathbf{T}^j)$ from vertices in ${T}_b$.
The linearity of expectation implies that
\begin{equation}\label{eq:basis:prop:Bound4QTPlustUniformity}
Q_{u^*}(\mathbf{T}^j)=Q^a+Q^b.
\end{equation}
The proposition will follow by bounding appropriately $Q^a, Q^b$.
The bound of $Q^a$ is related on the event ${\cal E}$, in the 
statement of Theorem \ref{thrm:BlockUpdtCovergent}.

\begin{lemma} \label{lemma:Bound4QAPlusUniformity}
Under the assumptions of Proposition \ref{prop:Bound4QTPlustUniformity},
we have that   $Q^{a}\leq n^2(1+\epsilon/3)^{-1}$.
\end{lemma}

\begin{proof}
A very useful observation is that  since   $u^*\in \ImpVrtx$,  every vertex in  ${T}_a$ 
is  of degree at most $\TDeg$, i.e., low degree vertex. Clearly we get an overestimate if we
assume that every vertex  $w\in T_a$ contributes to the distance with weight $n^2 \degree_{out}(w)$.
if it becomes disagreeing.

We prove the lemma   using induction. The base case is when $T_a$ is a single 
vertex tree, i.e., it is of height $0$. Let $T_a=\{z\}$.  Recall  that $\degree(z)\leq \TDeg$.  
Recall that $p_z$ is the probability of propagation for vertex $z$.
\[
Q_v(z) \ \leq \ p_z n^2 \degree(z) 
\  \leq \  n^2 (1+\epsilon/2)^{-1}.
\]
The second inequality follows from our assumptions about the event ${\cal E}$ which implies that $p_z\leq [(1+\epsilon/2)\degree(z)]^{-1}$.

Assume that the root of $T_a$ is vertex $z$. Also assume that the induction hypothesis is true
for the subtrees $T_a(y)$s, where $T_a(y)$ is the subtree that contains $y$, child of $z$, and all
its decadents. We are going to show that the induction is also true for $T_a$.
\begin{eqnarray*}
Q_{u^*} \left(T_a \right )  &\leq & 
\textstyle p_z \left( n^2\degree_{out}(z) + \sum_{y\in N(z)\cap {T}_a(y)} Q_z( T_a(y)) \right) \\
&< & p_z \left(n^2\degree_{out}(z) + n^2(\degree(z)-\degree_{out}(z)) \right) 
\hspace{1cm}
\mbox{[induction hypothesis]} \\
&<& n^2(1+\epsilon/2)^{-1}.
\end{eqnarray*}
The lemma follows.
\end{proof}

\begin{lemma} \label{lemma:Bound4QBPlusUniformity}
Under the assumptions of Proposition \ref{prop:Bound4QTPlustUniformity},
we have that  $Q^{b} \leq n^2d^{-10}$.
\end{lemma}

Before proceeding with the proof of Lemma \ref{lemma:Bound4QBPlusUniformity}, we note that the proposition follows by  plugging the 
bounds from Lemmas \ref{lemma:Bound4QAPlusUniformity} and \ref{lemma:Bound4QBPlusUniformity} 
into \eqref{eq:basis:prop:Bound4QTPlustUniformity}.

\begin{proof}[Proof of Lemma \ref{lemma:Bound4QBPlusUniformity}]

So as to the lemma, first note that the following holds for $Q_{u^*}(\mathbf {T}^j)$.
\[ \textstyle
Q_{u^*}(\mathbf{T}^j) \leq p_z \left(n^2\degree_{out}(z) + \sum_{y\in N(z)\cap B_j} Q_u(\mathbf{T}^j(y)) \right),
\]
where $\mathbf{T}^j(y)$ is the subtree of $\mathbf{T}^j$ rooted at $y$, child of $z$. From the above, we get that
\begin{eqnarray}
Q_{u^*}(\mathbold{T}^j)
&< & p_z 
\left( n^2\degree_{out}(z)  + (\degree(z)- \degree_{out}(z))\max_{y\in N(z)\cap \mathbf{T}^j } \left\{Q_z(\mathbold{T}_y) \right\}  \right) 
\nonumber \\
&\leq&  n^2
\max_{\P'=(u_0=z,u_1,\dots,u_\ell)}
\sum_{j=0}^\ell p_{u_j}\cdot \degree_{out}(u_j)
\prod_{i=0}^{j-1} p_{u_i} \times
\left [ \degree(u_i)-\degree_{out}(u_i)\right].\qquad \label{eq:QvTGenBound}
\end{eqnarray}
For  $\ell_0=15\frac{\log d}{(1+\epsilon/10)}$, it is direct that 
\[
Q^{b}= n^2\sum_{j\geq \ell_0+1} p_{u_j}\cdot \degree_{out}(u_j)
\prod_{i=0}^{j-1} p_{u_i} \times
\left [ \degree(u_i)-\degree_{out}(u_i)\right].\qquad \label{eq:QvBB} 
\]

\noindent
Since for every vertex $w\in \mathbf{T}^j$ we have $0\leq \degree_{out}(w) \  \leq \TDeg-1$, it holds that
\begin{eqnarray}\label{eq:QBDef}
Q^b & \leq &n^2\TDeg \ 
\sum_{j\geq \ell_0+1}  p_{u_j}   \prod_{i=0}^{j-1} p_i \times \degree_{in}(u_i).
\end{eqnarray}
The lemma will follow by bounding appropriately the magnitude of each sumad in \eqref{eq:QBDef}, separately. 

For $j\geq 1$,   we let the set $M\subseteq \{u_0,\ldots, u_{\ell_0+j-1}\}$ contain all vertices $u_i$
such that $\degree(u_i)>\TDeg$. Also, let $m=|M|$.
Also, let
\begin{eqnarray}
\mathbold{R}(j) &= & p_{\ell_{0+j}}  \prod_{i=0}^{\ell_0+j-1} p_{u_i} \times  \degree(u_i)
\nonumber \\
&\leq & \left( 1-\epsilon/3 \right)^{\ell_0+j-m}
\left( \mathbf{1}\{\degree(u_{\ell_0+j})\leq  \TDeg\} \frac{1}{k-\TDeg}+\mathbf{1}\{\degree(u_{\ell_0+j})> \TDeg \} \right) \prod_{w\in M}\degree(w).  \qquad 
\label{eq:FirstBound4Rj}
\end{eqnarray}
In the inequalities above, we use the convention that when $M=\emptyset$, then $\prod_{w\in M}\degree(w)=1$.

So as to bound $\mathbold{R}(j)$ we need to argue about  $\prod_{w\in M}\degree(w)$. 
Using Corollary \ref{cor:FromBreakPointProd} we get that
\begin{eqnarray}\label{eq:ProdDegreeBound}
\textstyle \prod_{w\in M}\degree(w) \leq d^{-15 m}\left( 1+\epsilon/10\right)^{\ell_0+j-m+1}.
\end{eqnarray}	
Plugging (\ref{eq:ProdDegreeBound}) into (\ref{eq:FirstBound4Rj}) we get that
\begin{eqnarray}
\mathbold{R}(j) 
&\leq & \left(1-\epsilon/5\right)^{\ell_0+j-m} d^{-15m}(1+\epsilon/10)\nonumber \\
&\leq & \left(1-\epsilon/5\right)^{\ell_0+j} 
\hspace{5cm} \mbox{[since $(1+\epsilon/10)d^{-15m} \left(1-\epsilon/5\right)^{-m}\ll 1$]}\nonumber \\
&\leq & d^{-13}\left(1-\epsilon/5 \right)^{j}  
\hspace{4.685cm}
\mbox{[since $\ell_0 \geq 15\frac{\log d}{\log(1+\epsilon/10)}$].} \label{eq:RJFinallBound} 
\end{eqnarray}
Plugging (\ref{eq:RJFinallBound})  into (\ref{eq:QBDef}) we get
\begin{eqnarray}
Q^b & \leq &  n^2 \TDeg  \ \sum_{j\geq 1}\mathbold{R}(j) \ 
\leq  n^2  \TDeg \ d^{-13}\cdot \sum_{j\geq 0}\left( 1-\epsilon/5 \right)^{j} 
\leq   n^2(10/\epsilon) d^{-12}, \nonumber 
\end{eqnarray}
where in the last inequality we used the fact that $\TDeg < 2d$. The lemma follows.
\end{proof}

\subsection{Proof of Proposition \ref{prop:BoundHangingTreesPlusUniformity}}
\label{sec:prop:BoundHangingTreesPlusUniformity}

To avoid trivialities assume $B_j$ is unicyclic. It is trivial to show that  $Q_{z'}(C)\leq |C|$. Also, it holds that 
 $\Pr[\mathbf{D}_{u}]\geq \Pr[\mathbf{D}_{z_i}]$, for any $z_i \in C$. Thus, it suffices to show that
 \begin{eqnarray}\label{eq:AlternateTargetHanging}
\Pr[\mathbold{D}_{u}] |C| \left(1+ \Delta\  \max_{z_i \in C\setminus\{z_1\}, x\in N_{z_i}\backslash C} 
\left \{ Q_{z_i}(\mathbf{T}^j_x) \right\} \right)\leq (\log \log d)^{(-|C|/10)}.
 \end{eqnarray}

\begin{claim}\label{claim:yVsl}
${\cal P}$ be any  path inside the block $B$ starting from $z$.  Let $\phi$ be the
fraction of vertices $w\in {\cal P}$ such that $\degree(w)>\TDeg$.
If the length of the path is at least 2, then  $\phi \leq \frac{\epsilon}{80 \log d}$.
\end{claim}
\begin{proof}
Let $\ell$ be the length of the path ${\cal P}$. Also let $M$ be the set of high degree vertices in 
${\cal P}$.
Using Corollary \ref{cor:FromBreakPointProd} and noting that for every $w\in M$ it holds $\degree(w)>\TDeg>d$,  
we get that
\[
\left [(1+\epsilon/10)d^{16} \right]^{m} \leq (1+\epsilon/10)^{\ell+1},
\]
where  $m=|M|$. Taking logarithm from both sides, we get that
\[
\frac{m}{\ell+1}\leq \frac{\log(1+\epsilon/10)}{16\log d}\leq \frac{\epsilon}{160\log d}. \qquad \mbox{[since $1+x<e^{x}$]}
\]
Since $\ell\geq 1$,  it elementary to verify that $\phi=m/\ell\leq 2 m/(\ell+1)$. The claim follows.
\end{proof}

\noindent
Let $\ell_0$ be the distance between the vertex $v$ and the cycle $C$. Also,
let $m$ be the number of high degree vertices in the path, in $B_j$, from vertex
$z$ to  vertex $u$,  e.g see Figure \ref{fig:DisBlockVsTj}. 
It holds that
\begin{eqnarray}\label{eq:ProbDis2z'}
\Pr[\mathbold{D}_{u}] &\leq &\left( d/2\right)^{-(\ell_0-m)}
\leq (d/2) ^{-9\ell_0/10}.\qquad \mbox{[from Claim \ref{claim:yVsl}]}.
\end{eqnarray}
In the first  inequality we also use the fact that $p_z>d/2$ for a low degree vertex.
Furthermore, using (\ref{eq:ProbDis2z'}) and the fact that $\ell_0\geq 2\log(\Delta\  |C|)$,
we get that
\begin{eqnarray}\label{eq:PrDisCLength}
\Pr[\mathbold{D}_{u}] \cdot |C|\leq (d/2)^{-\frac35{ \log\Delta^3|C|}}.
\end{eqnarray}
For the following result it helps to consider Figure \ref{fig:DisBlockStudy}.

\begin{lemma}\label{lemma:QHangingAloneDisgreement}
For every $z_i\in C\setminus\{z_1\}$, and any $x\in N_{z_i}\backslash C$ the following is true:
Let ${\cal P}$ be a path from $z$ to $z_i$ (any path). Let $H$ be the set of vertices of high degree
in this path.  In the setting of Proposition \ref{prop:BoundHangingTreesPlusUniformity}, it holds that
\begin{equation}\nonumber
\textstyle Q_{z_i}(\mathbf{T}^j_x)\leq n^23(2\epsilon)^{-1}\left( 1+\epsilon/10 \right)^{l-h}\left( \prod_{w\in H}d^{15}\  \degree(w)\right)^{-1},
\end{equation}
where $h=|H|$ and $l$ is equal to the length of ${\cal P}$.
\end{lemma}
The proof of Lemma \ref{lemma:QHangingAloneDisgreement} appears in Section 
\ref{sec:lemma:QHangingAloneDisgreement}.
Using Lemma \ref{lemma:QHangingAloneDisgreement} and (\ref{eq:ProbDis2z'}), 
we get that
\begin{eqnarray}
\Pr[\mathbold{D}_{u}] \  \Delta\ Q_{z_i}(\mathbf{T}^j_x) &\leq&  \textstyle 
n^2 3(2\epsilon)^{-1}\Delta \  \left( d/2 \right)^{-9\ell_0/10}
\left( 1+\epsilon/10 \right)^{l-h}\left( \prod_{w\in H}d^{15}\cdot \degree(w)\right)^{-1} 
\nonumber \\
&\leq& n^2 
3(2\epsilon)^{-1}\Delta \  \left(d/2 \right)^{-9\ell_0/10} \left( 1+\epsilon/10 \right)^{l-h} d^{-16h} 
\hspace{2cm} \mbox{[since $\forall w\in H\; \degree(w)>d$]} \nonumber\\
&\leq& n^2 
3(2\epsilon)^{-1}\Delta \  \left(d/2 \right)^{-9\ell_0/10} \left( 1+\epsilon/10 \right)^{l} 
\hspace*{3.4cm}
\mbox{[since $(1+\epsilon/10)d^{16}>1$]} \nonumber\\
&\leq& n^2 
3(2\epsilon)^{-1}\Delta \  \left( d/2 \right)^{-9\ell_0/10} \left( 1+\epsilon/10 \right)^{\ell_0+|C|} 
\hspace*{2.6cm} \mbox{[since $\ell<\ell_0+|C|$]} \nonumber\\
&\leq& n^2 
\left( 3\epsilon^{-1}\Delta \left(\frac{2+\epsilon}{d^{9/10}} \right)^{\frac{(\log \log d)\log \Delta}{\log d}} \right)\ 
 \left(\left(\frac{2+\epsilon}{d^{8/10}} \right)^{\frac{(\log \log d)}{\log d}}(1+\epsilon) \right)^{|C|}
 d^{-\frac{ (\log\log d) |C|}{10\log d}},  \nonumber
\end{eqnarray}
where in the last inequality we use  that $\ell_0>\frac{(\log\log d)}{\log d}\left( |C|+\log \Delta\right)$.
It is direct that
\begin{eqnarray}\label{eq:FinalBound4PdisDelaQxTx}
\Pr[\mathbold{D}_{u}] \  \Delta\  Q_{z_i}(\mathbf{T}^j_x)   \leq n^2 (\log d)^{-|C|/10}.
\end{eqnarray}
Combining (\ref{eq:FinalBound4PdisDelaQxTx}) and (\ref{eq:PrDisCLength}) we get that
 (\ref{eq:AlternateTargetHanging}) is true. The proposition follows.

\begin{figure}
	\centering
		\includegraphics[height=6.3cm]{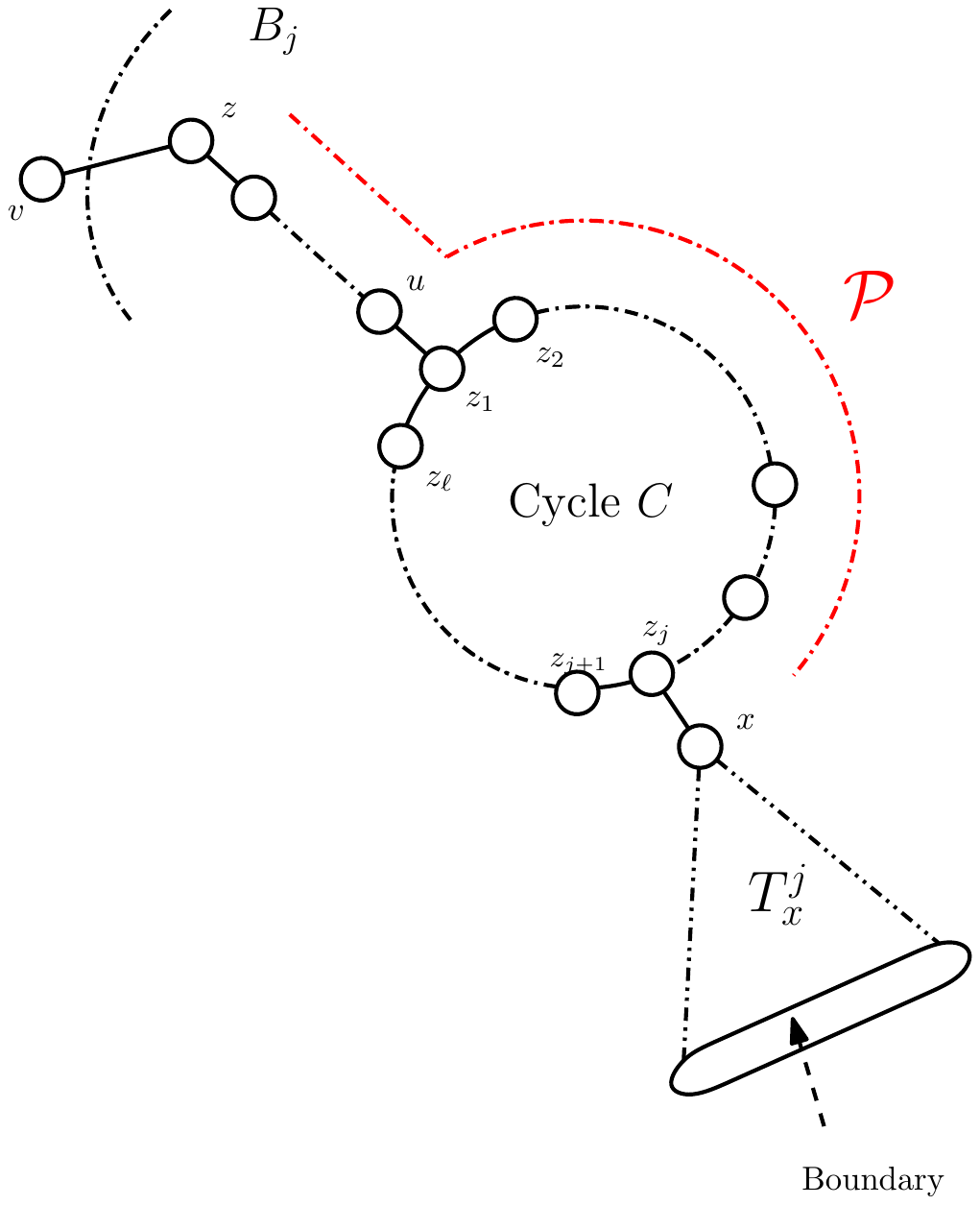}
		\caption{``Unicyclic Block''.}
	\label{fig:DisBlockStudy}
\end{figure}

\subsubsection{Proof of Lemma \ref{lemma:QHangingAloneDisgreement}}\label{sec:lemma:QHangingAloneDisgreement}

Using the same arguments as for (\ref{eq:QvTGenBound}) in the proof of Lemma \ref{lemma:Bound4QBPlusUniformity} , we get that
\begin{eqnarray}
Q_{z_i}(\mathbf{T}^j_x) &\leq & n^2
\max_{\P'=(u_0=z,u_1,\dots,u_\ell)}
\sum_{j=0}^\ell  p_{u_j}\cdot \degree_{out}(u_j)
\prod_{i=0}^{j-1} p_{u_i} \times
\left [ \degree(u_i)-\degree_{out}(u_i)\right]\qquad\nonumber \\ 
&\leq & n^2\  \TDeg \ 
\max_{\P'=(u_0=z,u_1,\dots,u_\ell)}\sum_{j=0}^\ell p_{u_j}
\prod_{i=0}^{j-1} p_{u_i} \times
\left [ \degree(u_i)\right],\qquad \label{eq:Qx'TxBound}
\end{eqnarray}
where the last inequality follows from the fact that $0\leq \degree_{out}(u_j)\leq \TDeg$, for any 
$u_j$.

Let ${\cal P}_z=\{w_0=z,\ldots,w_{\ell} \}$ be the path that maximizes the r.h.s. of (\ref{eq:Qx'TxBound}).
Also let
$$
\textstyle \mathbold{R}(j)=p_{w_j} \prod_{i=0}^{j-1} p_{w_i} \times \left [ \degree(w_i)\right].
$$
That is, $\mathbold{R}(j)$ is the $j$-th sumad in (\ref{eq:Qx'TxBound}).
Let $M$ be the set of high degree vertices in the subpath of ${\cal P}_z$,
$w_0,\ldots, w_j$. Also let  $m=|M|$.   It holds that
\begin{eqnarray}\label{eq:R(j)BoundHangBase}
\mathbold{R}(j) &\leq & \textstyle \left(\frac{1}{1+\epsilon/5}\right)^{j-m}\prod_{w\in M}\degree(w).
\end{eqnarray}
So as to compute $\prod_{w\in M}\degree(w)$ we use Corollary \ref{cor:FromBreakPointProd} 
and get that
\[
\textstyle \prod_{w\in M}\degree(w)\leq d^{-15m}\left(1+\epsilon/10\right)^{\ell+j-(h+m)}\left(\prod_{w\in H}d^{15}\ \degree(w) \right)^{-1}.
\]
Plugging  the above into (\ref{eq:R(j)BoundHangBase}) we get that
\begin{eqnarray}\nonumber
\mathbold{R}(j)  &\leq &\textstyle  \left(1-2\epsilon/3 \right)^{j-m}
d^{-15m}\left(1+\epsilon/10\right)^{\ell-h}\left(\prod_{w\in H}d^{15}\ \degree(w) \right)^{-1} \nonumber\\
&\leq &  \textstyle \left(1-2\epsilon/3\right)^{j} 
\left(1+\epsilon/10\right)^{\ell-h}\left(\prod_{w\in H}d^{15}\ \degree(w) \right)^{-1}. \nonumber
\end{eqnarray}
In the  above bound for $\mathbold{R}(j)$, the only quantity that depends on $j$ is  $(1-2\epsilon/3)^{j}$.
We have that 
\begin{eqnarray}
Q_{z_i}(\mathbf{T}^j_x)&\leq & n^2  \textstyle  \sum_{j\geq 0}\mathbold{R}(j)
\ \leq \  \frac{3}{2\epsilon}\left(1+\epsilon/10\right)^{l-h}\left(\prod_{w\in H}d^{15}\ \degree(w) \right)^{-1}. \nonumber
\end{eqnarray}
The lemma follows.

\section{ Proof of Lemma  \ref{thrm:RapidMixBlockGnp}}\label{sec:thrm:GnpAdmitsPart}

First we show that typical instances of $G(n,d/n)$ admit a block partition ${\cal B}(\epsilon, d,\Delta)$.  
Given $\epsilon>0$, consider the graph $G\sim G(n,d/n)$ for sufficiently large $d>0$.  We use the weighting schema in 
\eqref{eq:DefVrtxWeight}  and \eqref{def:PathWeight} to specify the breakpoints. At this point we introduce the 
notion of ``influence path".

\begin{definition}
The  path $L$ is called ``influence path" only  if none of its vertices is a breakpoint. 

If vertex $w_1$ is a breakpoint, we define that there is only one influence path that starts from $w_1$, this is 
the trivial  path $L={w_1}$.
\end{definition}

\noindent
The following result which implies that 
typically $G(n,d/n)$ does not have long influence paths. 
We call {\em elementary}  every path $L=w_1, \ldots, w_{\ell}$  such that there is no
other path $P$ \footnote{ i.e., $P$ is different than $L$} of length less than $10\frac{ \ln n}{d^{4/5}}$ which connects any two vertices in 
$L$.

\begin{theorem}[Efthymiou \cite{Efthymiou}]\label{thrm:BreakablePaths}
Let  $\epsilon \in (0,3/2)$. For large $d$, consider $G\sim G(n,d/n)$. Let  $\mathbf{U}$ be the set of the 
{\em elementary} paths in $G$ of  length  $\frac{\ln n}{ (\ln d)^{5} }$  that do not have any $r$-breakpoint
for $r=\log n/d^{4/5}$.  It holds that
$
\Pr[\mathbf{U}\neq \emptyset] \leq 4n^{(-\frac{1}{2}\ln d+2)}. 
$
\end{theorem}

\noindent
Furthermore, we use the following result from \cite{Efthymiou}.
\begin{lemma}\label{lemma:BreakPointExtension}
Let  $\epsilon \in (0,3/2)$. For large $d$, consider $G\sim G(n,d/n)$. 
With probability at least $1-2n^{-\frac{d^{2/5}}{2}}$ over the graph instances the following is true:
Every  vertex $v$ which is $r$-breakpoint  for $r=\log n/d^{4/5}$, it is, also, a $r'$-breakpoint for $r'=10\log n$.
\end{lemma}

\noindent
The proof of Lemma \ref{lemma:BreakPointExtension}  is the same as the proof of Lemma 3 in \cite{Efthymiou}.

Let ${\cal C}$ be the set of all cycles of length at most  $4\frac{\ln n}{ (\ln d)^5}$ in $G$.  We need to argue that
any two cycles in ${\cal C}$ are far apart from each other. In particular, we have the following result:

\begin{lemma}\label{lemma:DistantCycle}
  With probability at least $1-10n^{-3/4}$ over the instances of $G(n,d/n)$, 
any two cycles in ${\cal C}$ are at distance greater than $10\frac{\log n}{(\log d)^5}$.
\end{lemma}
\begin{proof}
If there is a pair of cycles in ${\cal C}$ at distance less than  $10\frac{\ln n}{ (\ln  d)^5}$, then the following should hold: There 
is a set of vertices $S$ of cardinality less than $2\frac{\ln n}{(\ln  d)^2}$ such that the number of edges between the vertices 
in $S$ is at least $|S|+1$.  We  show that such a set does not exist in $G(n,d/n)$  with probability at least $1-n^{-3/4}$.

Let $D$ be the event that such a set exists. It holds that
\begin{eqnarray}
\Pr[D]&\leq &\sum_{r=1}^{2\frac{\ln n}{ (\ln  d)^2}}{n \choose r} {{r\choose 2} \choose r+1}\left(\frac{d}{n}\right)^{r+1}
\leq \sum_{r=1}^{2\frac{\ln n}{ (\ln  d)^2}} \left(\frac{ne}{r}\right)^r\left(\frac{r^2e}{2(r+1)}\right)^{r+1}\left(\frac{d}{n}\right)^{r+1}
\hspace{.75cm}\textstyle  \left[\textrm{as }{n \choose r}\leq \left(\frac{ne}{r}\right)^r \right] \nonumber\\
&\leq& \frac{1}{n}\sum_{r=1}^{2\frac{\ln n}{ (\ln  d)^2}}\left(\frac{erd}{2}\right)\left(\frac{e^2d}{2}\right)^{r}
\leq \frac{ed}{ (\ln  d)^2}\ \frac{\ln n}{n}\sum_{r=1}^{2\frac{\ln n}{ (\ln  d)^2}}\left(\frac{e^2d}{2}\right)^{r}
\hspace{3.15cm}\mbox{[as $r\leq 2\ln n/ (\ln  d)^2$]}\nonumber \\
&\leq& n^{-9/10}\left({e^2d}/{2}\right)^{2\frac{\ln n}{ (\ln  d)^2}}
\leq n^{-3/4}. \nonumber
\end{eqnarray}
The lemma follows
\end{proof}

\noindent
Finally we use the following   standard result, for a
proof  see e.g. in  \cite{FriezeKaron}, in Section 3.
\begin{lemma}\label{lemma:MaxDegree}
Let $\Delta$ be the maximum degree in $G(n,d/n)$.  It holds that 
\[
\Pr[\Delta \geq (3/2)\log n/\log \log n] \leq n^{-1/4}.
\]
\end{lemma}

\noindent
We are going to show that $G$ admits the partition ${\cal B}={\cal B}(\epsilon, d, \Delta)$  if  (a) the maximum degree  $\Delta$ is less than
$(3/2)\log n/\log\log n$, (b)  the distance between any two cycles in  ${\cal C}$ is at least $10\log n/(\log d)^5$,
(c)  there are no elementary paths of length $\log n/(\log d)^5$ which do not contain 
a $r$ breakpoint for $r=\log n/d^{4/5}$ and (d) every $r$ breakpoint in $G$ is also an
$r'$ breakpoint for $r=10\log n$. From  Lemmas  \ref{lemma:BreakPointExtension}, \ref{lemma:DistantCycle}, \ref{lemma:MaxDegree} 
and  Theorem \ref{thrm:BreakablePaths},  $G$ satisfies these properties with probability $1-o(1)$.

Let ${\cal H}$ be the set of breakpoints in $G$. Given the sets ${\cal H}$ and ${\cal C}$ we specify the
set of block ${\cal B}$ as follows:
For the cycle $C\in {\cal C}$ we create the block $B_C$.  Let $\partial^r C$ contain all the vertices which are at 
distance $r=\textstyle \max \left \{2  \log(|C|\  \Delta),\   \frac{\log\log d}{\log d}\left( {|C|+\log \Delta} \right) \right\}$. 
Note that we always have $r\leq 5\frac{\log\log d}{(\log d)^6} \log n $.
The block $B_C$ contains all the vertices in the cycle $C$ and $\partial^r C$. 
Additionally the block $B_C$ contains every vertex $w$ for which there is an influence path from $w$ to $\partial^r C$.
We repeat the above process for every cycle in ${\cal C}$.   Note that our assumptions about $G$ imply that the 
blocks created for each cycle are vertex disjoint.

Having specified the blocks which correspond to the cycles in ${\cal C}$,
there are  vertices whose block has not been specified yet. 
For each such vertex $w$ we specify its block $B_w$ by working  as follows:
The block $B_w$ contains $w$ and every $u$ which is reachable from $w$ through an influence
path. The block construction ends once we have specified the block for all vertices in $G$.
Note that if $w$ is a breakpoint then $B_w$ is a single vertex block. This follows from the definition of
influence path.

In the following result we show that  the blocks in  ${\cal B}$ have the structured we promised.

\begin{lemma}\label{lemma:SimpleBlock}
For $\epsilon, d$ as specified in the statement of Lemma \ref{thrm:RapidMixBlockGnp},
consider the graph $G$ which admits the block partition as we described above.
Additionally,  assume that $G$ is such that 
\begin{enumerate}
\item  the distance between any two cycles in  ${\cal C}$ is at least $10\log n/(\log d)^5$ 
\item  there are no elementary paths of length $\log n/(\log d)^5$ which do not contain 
a breakpoint. 
\end{enumerate}
Then,  the set of blocks ${\cal B}$  contains only blocks which are trees with one extra edge. 
\end{lemma}
\begin{proof}
Let ${\cal B}_1$ be the set of blocks created from the cycles in ${\cal C}$ and let ${\cal B}_2={\cal B}\backslash {\cal B}_1$.
It suffices to show that  ${\cal B}_1$ contains only unicyclic blocks and  ${\cal B}_2$ contains only trees.

First we focus on  ${\cal B}_1$. The assumption that there are no elementary paths of length
$\log n/(\log d)^5$ which do not contain a breakpoint implies the following: There is no vertex 
$w$ at distance more than $(3/2)\log n/(\log d)^5$ from a cycle $C\in {\cal C}$ such that both
$w$ and $C$ belong to the same block.  Then, the assumption that for any two cycles in
${\cal C}$ their distance is at least $10\log n/(\log d)^5$  implies that for any two
cycles $C_1, C_2\in {\cal C}$ the corresponding blocks do not intersect.

So as to show that ${\cal B}_2$ consists of tree-like blocks we work as follows: Let some $B\in {\cal B}_2$ and let $w$ be
the vertex we used to created it.  It is direct that every path that connects $w$ to some vertex in any of the blocks in ${\cal B}_1$
should contain at least one breakpoint (otherwise $w$ should belong to a block in ${\cal B}_1$).  
That is, if $B$ contains a cycle $C$,  then $C\notin {\cal C}$. This implies that $|C|> 4\frac{\ln n}{ (\ln  d)^5}$.  
It suffices to show that 
every $B\in {\cal B}_2$ cannot contain a cycle of length $\ell\geq 3\frac{\ln n}{ (\ln  d)^5}$.
But the second assumption about $G$ implies that the maximum cycle in $B$ is $2\log n/(\log d)^5$.
This implies that $B$ cannot contain any cycle.  We conclude that ${\cal B}_2$ contains only blocks which are trees.

The lemma follows. 
\end{proof}

\noindent
So as to show that $G$  admits the block partition ${\cal B}(\epsilon, d, \Delta)$,  it suffices to show 
that  the graph $G$ (and the set of blocks ${\cal B}$) 
has the following properties:
\begin{enumerate}
\item for each multi-vertex block $B\in {\cal B}$,  each $u\in \outBound B$ is $r$-breakpoints
for $r\geq { \max\{{\tt diam}(B), \log\log n\}}$
\item for each multi-vertex block $B$,  each vertex in $\outBound B$ has exactly one neighbor inside $B$
\item  if $B$ contains a cycle $C$, we have that
$\distance(v,C)\geq \max \left \{2  \log(|C|\  \Delta),\   \frac{\log\log d}{\log d}\left( {|C|+\log \Delta} \right) \right\}$,
  for every  $v\in \outBound B$
\item  every  $v\in \ImpVrtx$  does not belong to any cycle of length less than $d^{2}$.
\end{enumerate}

\noindent
 We start by arguing about (1). First, we show that $\outBound B$ consists of $r$-breakpoints for $r=10\log n $.
Assume that  some vertex $u\in \outBound B$ is not a $r$ breakpoint. W.l.o.g. assume that this block is a tree.
Let  $w$ be the vertex that is used to specify the block $B$.  Since we assume that $B$
is multi-vertex $w$ is not a breakpoint. Furthermore, since  $u\in \outBound B$ is not a breakpoint 
there should be an influence path from $w$ to $u$. In turn, this implies that $u$ should be included into
$B$ during the construction of $B$. Clearly, this is a contradiction since  $u$ was assumed to  be in $\outBound B$. 
Then, (1) follows by noting that the diameter of $B$ is always less than $10\log n$.

For  showing (2)  we use proof by contradiction, as well. Assume that we have a multi-vertex block
$B$ and there is  $u\in \outBound B$ which has at least two neighbors inside $B$.  Consider first the
case where $B$ was created by a single vertex $w$. 
The block that is created by a single vertex  cannot intersect with a cycle of length less than $4\log n/(\log d)^5$
However, if there exists such $u\in \outBound B$ then there should be a cycle of length less than 
$(5/2)\log n/ (\log d)^5$ that intersects with the block $B$. Clearly this cannot be
the case since we assumed that the block is created by a single vertex and not a short cycle.
If on the other hand the block $B$ started from a cycle $C\in {\cal C}$, then the fact that there exists
$u\in \outBound B$ implies that there two cycles of length less than $4\log n/(\log d)^5$ whose distance
is much less than $(3/2)\log n/(\log d)^5$.  This is a contradiction, since we assumed that any two cycles in
${\cal C}$ are at greater distance.

As far as (3) is concerned, consider the block construction for   a block which
includes a cycle $C\in {\cal C}$. In such a block we always add  the vertex sets $\partial ^r C$ 
in the block. The assumption that $\Delta=(3/2)\log n/\log\log n$ and the fact that  the length of 
the cycle $C$ is at most $\log n/(\log d)^5$ imply that the addition of  the set of vertices 
$\partial ^r C$  into the block guarantees that  (3) is satisfied.

Finally, for (4) we only need to observe that every $v$ in the outer boundary of a block cannot belong to
a cycle in ${\cal C}$. 

We, also, need to show for $k\geq \alpha d$,
with high probability over the instances of $G(n,d/n)$, the graph  can be colored
using at least $k\geq \alpha d$ colors and the state space is connected. 
As far as  the $k$-colorability of $G(n,d/n)$, for $k\geq \alpha d$, is regarded we use the result from
 \cite{optas-chrno,chromatic2},  i.e., with probability $1-o(1)$ the chromatic number of $G(n,d/n)$ is $d/(2\ln d)$.

 From \cite{DFFV} we have that  the Glauber dynamics (and hence the  block dynamics) is ergodic with probability 
 $1-o(1)$ over the instances  $G(n,d/n)$ when $k\geq d+2$.  For the sake of completeness let us sketch the  proof 
 for ergodicity in \cite{DFFV}. It is shown that if a graph $G$ has no $t$-core\footnote{For some integer $r>0$ and 
 a graph $G$, we say that $G$  has a $r$-core if it has a subgraph with minimum   degree $r$}, then for all $k\geq t+2$ 
 the Glauber dynamics for  $k$-coloring yields an ergodic Markov chain (Lemma 2  in \cite{DFFV}). Then the authors 
 use the result  in \cite{k-core}, which states that w.h.p. $G(n,d/n)$ has no $t$-core for $t\geq d$.

We consider the claim about the size of the  blocks.
 The fact that the vertices in $\inBound B$, for  every $B\in {\cal B}$, are next to a break-point and Lemma \ref{lemma:GrowthFromBP} imply the following:
 for each $w\in \inBound B$ the number of vertices that are at distance $\ell$ from $w$ is less than $[(1+\epsilon)d]^{\ell}$.
Then, the result follows easily one we note that the diameter of each block in ${\cal B}$, given that $\Delta=\Theta(\log n/(\log\log n))$, is less than 
$10\log n/\log^4 d$.

\section{Hard-Core Model - Analysis for Rapid Mixing}\label{sec:HCRapid4Block}

In this section we show the following result:
 \begin{theorem}\label{thrm:Fug1OverD-main}
For all $\eps>0$, there exists $d_0>1$, for all $d>d_0$, for $\lambda \leq (1-\eps)/d$, 
there exists $C=C(d)>0$ such that with probability $1-o(1)$ over the choice of $G\sim G(n,d/n)$, the mixing time of the  Glauber dynamics is 
$O(n^C)$.
 \end{theorem}

\noindent
So as to  get Theorem \ref{thrm:Fug1OverD-main} first we  prove the following result 
that concerns block dynamics.

\begin{theorem}\label{thrm:RapidMixingAvrgDegGraphHC}
For all $\epsilon,\Delta>0$, there exists $C, d_0>0$  such that for all $d\geq d_0$, 
and any  graph  $G$ which admits  block partition ${\cal B}={\cal B}(\epsilon, d)$ 
the following is true: For $\lambda\leq (1-\epsilon)/d$, the 
  block dynamics with set of block ${\cal B}$ has mixing time 
\[
 \Tmix\leq C n \log n.
 \]
\end{theorem}

\noindent
Additionally to Theorem \ref{thrm:RapidMixingAvrgDegGraphHC} we have the following result.

\begin{lemma}\label{thrm:RapidMixBlockGnpHC}
For all $\epsilon>0$ and $\Delta=(3/2) \left( \log n/\log \log n\right)$, there exists $d_0>0$ such that 
for all $d\geq d_0$  $G(n,d/n)$ admits the block partition ${\cal B}={\cal B}(\epsilon, d, \Delta)$.
\end{lemma}

\noindent
The proof of Lemma \ref{thrm:RapidMixBlockGnpHC} is almost identical to that of  Lemma
\ref{thrm:RapidMixBlockGnp}. For this reason we omit it.

In light of Theorem \ref{thrm:RapidMixingAvrgDegGraphHC} and Lemma \ref{thrm:RapidMixBlockGnpHC},
Theorem \ref{thrm:Fug1OverD-main} follows by utilizing a standard comparison argument, see Section \ref{sec:thm:1.76-Glauber}.

\noindent
We proceed with the  proof of Theorem \ref{thrm:RapidMixingAvrgDegGraphHC}.
First we note that for any $\lambda>0$ the dynamics is trivially  ergodic for any  $G$ which admits a block partition 
${\cal B}(\epsilon, d, \Delta)$. This follows from the observation that from every independent set of $G$ there
is a sequence of transitions to the empty independent set, each with positive probability,
and the other way around.

For showing the rapid mixing result for the hard-core model  it suffices to show that
in the block dynamics the blocks are always in a convergent configuration. Then
rapid mixing follows by using  Theorem  \ref{thrm:BlockUpdtCovergent} 
and standard arguments, almost identical to those we use for Theorem \ref{thm:block-regular}.

\begin{corollary} \label{thrm:BlockProbProp}
For all $\epsilon >0$, $\Delta>0$, there exists $d_0>0$ such that for any $d\geq d_0$, 
for every  graph   $G$ which admits block partition  ${\cal B}(\epsilon, d, \Delta)$, 
and any $v\in \ImpVrtx$ the following is true:

Let $(X_t)_{t\geq 0}, (Y_t)_{t\geq 0}$ be  two copies of the  block dynamics  on
the  hard-core model on $G$ such that  for some $t\geq 0$ we  have $X_t\oplus Y_t=\{u^*\}$. 
For any $B$ such that $u^*\in \outBound B$ and any vertex $w\in B$ we have that
the probability of propagation $p_w <(1-\epsilon)/d$.
\end{corollary}

\begin{proof}
Let $(X_t)_{t\geq 0}, (Y_t)_{t\geq 0}$ be  two copies of the  block dynamics  on
the  hard-core model on $G$ such that  for some $t\geq 0$ we
 have $X_t\oplus Y_t=\{u^*\}$.  Consider some block $B$ such that
 $v\in \outBound B$. Then, so as to bound the probability of
 propagation for each vertex $u$ note the following:
Assume that the  vertex $w$ is disagreeing, w.l.o.g.
assume that $X_{t+1}(w)$ is {\em occupied}, i.e., $w$ belongs to the independent set,  and 
$Y_{t+1}(w)$ is {\em unoccupied}.
Clearly $X_{t+1}(u)$ cannot become occupied. The only way we can have disagreement at
$u$, is when all the neighbours of $u$, apart from $w$, in both configurations  are unoccupied. 
Then, $Y_{t+1}(u)$ becomes occupied (disagreeing)  with probability $\frac{\lambda}{1+\lambda}$.

Choosing  $\lambda\leq (1-\epsilon)/d$,  the above remarks implies that the probability of propagation 
is less than $(1-\epsilon)/d$, always. 
\end{proof}

\noindent
In light of Corollary \ref{thrm:BlockProbProp}, Theorem \ref{thrm:RapidMixingAvrgDegGraphHC} follows.

\section{Rapid Mixing for Single Site  Dynamics - The Comparison}\label{sec:thm:1.76-Glauber}

In this section we show that the rapid mixing result we get for the block dynamics for coloring
imply  Theorem \ref{thm:1.76-main}. Similarly, for the hard-core model, i.e., Theorem \ref{thrm:Fug1OverD-main}.
In a lot of our results in this section we need to use continuous time Markov chains, rather than discrete time. 
In the {\em continuous time} block dynamics each block is updated  according to an independent
Poisson clock with rate 1.

We use the following  comparison result from  \cite{Martinelli}, which in our context writes as follows:

\begin{proposition}\label{prop:SingleSiteVsBlock}
Consider some graph $G$. Let $(X_t)_{t\geq 0}$ be the continuous time block dynamics, with set of 
blocks ${\cal B}$, where each vertex $v$ belongs to $Q_v$ different blocks. 
Also, let $(Y_t)_{t\geq 0}$ be the continuous time single site dynamics on $G$.
Let $\tau_{block}, \tau$ be the relaxation times of $(X_t)$ and $(Y_t)$, respectively. 
Furthermore, for each block $B\in {\cal B}$ let $\tau_B$ be the relaxation time of the continuous 
time single site dynamics on $B$, given any arbitrary condition at $\outBound B$.  
Then we have that
\[
\tau \leq \tau_{block} (\max_{B\in {\cal B}} \tau_B) (\max_{v} Q_v).
\]
\end{proposition}

\noindent
For some $G\in \IntrstGraphFam(\epsilon, d, \Delta)$, with block partition ${\cal B}$, 
let ${\cal P}$ be the set of paths which connect either a high degree vertex or
the cycles in the block $B$ (if any)  to $\inBound B$. 

We show rapid mixing for the single site Glauber dynamics of
$G(n,d/n)$ if,  additionally to the condition $G(n, d/n)\in \IntrstGraphFam(\epsilon, d, \Delta)$, for 
$\Delta=(3/2)\log n/(\log\log n)$,  the graph, also,  satisfies the following one:
For each path $P\in {\cal P}$ let 
\[
{\cal J}(P)=450 \sum_{u\in P}\left(\log(\degree(u))+{\degree(u)}/{k} \right).
\]
The additional property  is that every path $P\in {\cal P}$ is such that 
\begin{equation}\label{eq:NewF(P)Condition}
{\cal J}(P)\leq 10^4 \log n/(\log d)^2.
\end{equation}
where $|P|$ is the number of vertices in $P$. 

For some $\epsilon, d, \Delta>0$,  let ${\cal  L}(\epsilon, d, \Delta)$ be the family of graphs $G$ such
that $G\in \IntrstGraphFam(\epsilon, d, \Delta)$ and every $P\in {\cal P}$ satisfies \eqref{eq:NewF(P)Condition}.

\begin{lemma}\label{thrm:RelaxationWeightBound}
For $\epsilon, d$ and $\Delta$ as in Lemma \ref{thrm:RapidMixBlockGnp}, with
probability $1-o(1)$ over the graph instances we have that
$G(n,d/n)\in {\cal L}(\epsilon, d,\Delta)$.
\end{lemma}
\noindent
The proof of Lemma \ref{thrm:RelaxationWeightBound} appears in Section \ref{sec:thrm:RelaxationWeightBound}.

For $\epsilon, d$ and $\Delta$ as in Lemma \ref{thrm:RapidMixBlockGnp}, 
 consider some graph $G\sim G(n,d/n)$ such that $G \in {\cal L}(\epsilon, d, \Delta)$.
Let $(X_t)_{t\geq 0}$ be the continuous time, block dynamics, with set of blocks ${\cal B}$.
Also, let $(Y_t)_{t\geq 0}$ be the continuous time single site dynamics on $G$.
Theorem \ref{thrm:RapidMixingAvrgDegGraph176} and Lemma \ref{thrm:RapidMixBlockGnp} imply that choosing 
$k\geq (\alpha +\epsilon)d$, for   $\tau_{block}$,  the relaxation time  of $(X_t)_{t\geq 0}$, we have that
 \begin{equation}\label{eq:tau_blockBound}
 \tau_{block}=O\left( \log n\right).
 \end{equation}

\begin{lemma}\label{lemma:BlockRelaxColor}
For every $B\in {\cal B}$ consider the continuous time, single site dynamics $(X^B_t)_{t\geq 0}$ over
the $k$-colorings of $B$ with arbitrary boundary condition at $\outBound B$. Let $\tau_{B}$ be the 
relaxation time of $(X^B_t)_{t\geq 0}$. For any $k\geq (\alpha+\epsilon)d$
it holds that
 \begin{equation}
 \tau_{B}\leq n^{2/(\log d )^2}. \nonumber
\end{equation}
\end{lemma}

\noindent
The proof of Lemma \ref{lemma:BlockRelaxColor} appears in Section \ref{sec:lemma:BlockRelaxation}.

Combining \eqref{eq:tau_blockBound}
with Lemma \ref{lemma:BlockRelaxColor} and Proposition \ref{prop:SingleSiteVsBlock} we get the following:
letting $\tau_{cont}$ be the relaxation time $(Y_t)_{t\geq 0}$ for $k\geq (\alpha+\epsilon)d$,   
we have that
\[
\textstyle \tau_{\rm cont} = O\left( n^{2/(\log d )^2}\log n\right) =O\left( n^{3/(\log d )^2}\right ).
\]
Now, let $(Z_t)_{t\geq 0}$ be the {\em discrete time}, single site Glauber dynamics on the $k$-colorings
of $G$ with $k\geq(\alpha+\epsilon)d$.  Let $\tau_{\rm disc}$ and $T_{\rm mix}$ be the relaxation time and the mixing
time of $(Z_t)_{t\geq 0}$, respectively. The above bound for $\tau_{cont}$ implies that
$ \tau_{disc} \leq n^{1+3/(\log d )^2}.$
Then, it is standard that 
$
\textstyle T_{mix} = O\left(n^{2+3/(\log d )^2} \right).
$
Theorem \ref{thm:1.76-main} follows.

As far as the hard-core model is regarded, we  show the following result.

\begin{lemma}\label{lemma:BlockRelaxHardCore}
For every $B\in {\cal B}$ consider the continuous time, single site dynamics $(X^B_t)_{t\geq 0}$ for
the hard-core model of $B$ with arbitrary boundary condition at $\outBound B$. Let $\tau_{B}$ be the 
relaxation time of $(X^B_t)_{t\geq 0}$. For any $\lambda \leq (1-\epsilon)/d $
there exists  $C_1>0$ which depends on $\epsilon, d$, such that 
$ \tau_{B}\leq n^{ C_1}. $
\end{lemma}

\noindent
 In light of Claim \ref{claim:PathDensityGnp},   Lemma \ref{lemma:BlockRelaxHardCore} 
follows directly from Theorem 4.2 and Lemma 4.1 and 4.4 in \cite{MS2}.

Theorem \ref{thrm:RapidMixingAvrgDegGraphHC} follows  by combining Lemma \ref{lemma:BlockRelaxHardCore}
 with arguments which are very similar to those we used for the coloring model.

\subsection{The relaxation time for the blocks - Proof of Lemma \ref{lemma:BlockRelaxColor}}\label{sec:lemma:BlockRelaxation}

We proceed by bounding appropriately the  quantities $\tau_B$ for every $B\in {\cal B}$. 
 As discussed  earlier, the blocks of $G$   are trees with at most one extra edge.

\begin{definition}
For a tree $T$ rooted at $v$,  let the {\em maximal path density} be defined as
$ m(T,v)= \max_{P}{\cal J}(P)$, 
where the maximum is  over all the paths $P$ in $T$ which start from $v$.
\end{definition}

\noindent
For a graph  $G\in {\cal L} (\epsilon, d, \Delta)$, with block partition ${\cal B}$,  
let ${\cal T}={\cal T}(G, {\cal B})$ be the family  which contains the following rooted trees,
 subgraphs of $G$: ${\cal T}$ includes all  the tree-like, multi-vertex blocks in ${\cal B}$.
 The root of each tree is a high degree vertex (any) inside the block.
Also, for each $B\in {\cal B}$ that is  unicyclic  with cycle $C=w_1, \ldots, w_{\ell}$ ,
the set ${\cal T}$ contains every subtree $T_i$ that hang from the cycle. That is,  for $i=1, \ldots, \ell$,
 $T_i$ is  the induced subgraph of $B$ that corresponds to the set of vertices
in the connected component of $B$ that contains vertex $w_i$ once we delete all the edges of 
$C$.  The root for $T_i$ is the vertex $w_i$.
For each $T\in{\cal T}$which belongs to the block $B$,  we  let $\outBound T$ be the set of 
vertices in $\outBound B$ which are incident to $T$.

The following result relates the relaxation times of the trees in ${\cal T}$ and their, corresponding,
maximal path density.

\begin{theorem}\label{thrm:TreeRelaxationBound}
For any $\epsilon, \Delta>0$ and sufficiently large $d>0$ let  $k\geq (\alpha+\epsilon)d$.
Consider $G\in \IntrstGraphFam (\epsilon, d, \Delta)$ and block partition ${\cal B}$. 
For any  $T\in {\cal T}$, with root $v$,  and  boundary condition $\sigma(\outBound T)$,  the relaxation time $\tau_{\rm rel}$ 
of the Glauber dynamics,  we have
$
\tau_{\rm rel}(T)\leq \exp\left( m(T,v)) \right).
$
\end{theorem}

\noindent
The proof of Theorem \ref{thrm:TreeRelaxationBound} appears in Section \ref{sec:thrm:TreeRelaxationBound}.

From Lemma \ref{thrm:RelaxationWeightBound} and Theorem \ref{thrm:TreeRelaxationBound} we get that 
if $G(n,d/n)\in {\cal L} (\epsilon, d, \Delta)$, where $\epsilon$, $d$ and $\Delta$ are as in Lemma \ref{thrm:RapidMixBlockGnp},
then the continuous time Glauber dynamics on $T\in {\cal T}(G(n,d/n), {\cal B})$
exhibits relaxation time
\begin{equation}\label{eq:TinCalTRelaxTime}
\tau_{\rm rel}(T)\leq n^{1/(\log d)^2}.
\end{equation}

\noindent
The above implies that for a tree-like block $B\in {\cal B}$ the lemma is true.

Consider the unicyclic block $B$ with arbitrary boundary condition at $\outBound B$. 
Let $C=w_1, \ldots, w_{\ell}$ be the cycle inside $B$, for some $\ell\leq \log n/(\log d)^5$.
Consider    $(Z^B_t)_{t\geq 0}$  the continuous time,  block dynamics, on $B$ with arbitrary boundary condition at $\outBound B$. 
The set of blocks is the subtrees $T\in {\cal T}$ which intersect with the cycle $C$.
Using path coupling and Proposition  \ref{prop:SpatialMixing4HighDegrees}  it is elementary to show that
the relaxation time of the block dynamics $\tau_B\leq 10 \log |C|= O(\log\log n)$.

Let $(X_t)_{t\geq 0}$ be the Glauber dynamics on $B$ with arbitrary boundary at $\outBound B$.
The bound on relaxation time for $(Z^B_t)_{t\geq 0}$,  combined with \eqref{eq:TinCalTRelaxTime} 
and  Proposition \ref{prop:SingleSiteVsBlock}, imply that the relaxation time for
$(X_t)$ is such that $\tau_{B} = O\left( n^{1/(\log d)^2}\log\log n \right) \leq O\left( n^{2/(\log d)^2} \right)$.
The lemma follows

\section{Proof of Theorem \ref{thrm:TreeRelaxationBound}}\label{sec:thrm:TreeRelaxationBound}

For the tree $T$ and some vertex $u\in T_u$, let $T_u$ denote the subtree of $T$ which contains $u$
and all its descendants. Unless otherwise specified, we assume that the root of $T_u$ is $u$.
Also, for a boundary set $\outBound T$ of  $T$,  we let $\outBound T_u$ contain every $w\in \outBound T$
which is a boundary at $T_u$, as well.

We also have the following result whose proof appears in Section \ref{sec:prop:MixingStar}.

\begin{proposition}\label{prop:MixingStar}
For $\epsilon, d, \Delta, k$ as in Theorem \ref{thrm:TreeRelaxationBound} the following is true:

Let $T\in {\cal T}$ and let $v\in T$. Consider $T_u$  and let $w_1, \ldots, w_{R}$  
be the children of the root, where  $R=\degree(v)$. 
 Consider the block dynamics with set of blocks ${\cal M}=\{\{v\}, T_{w_1}, \ldots, T_{w_R}\}$.
Assume that  for any $\sigma(\outBound T)$, any $v\in \{u, w_1, \ldots, w_{\ell}\}$ 
for the  random coloring $Z$ we have
\begin{equation}\label{eq:SMAssumption}
\left| \Pr[Z(v)\ | \ Z(\outBound T)=\sigma(\outBound T)]-1/k  \right| \leq 100/k^2.
\end{equation}
Then, under any boundary condition at $\outBound T$,  the  block dynamics $(X_t)_{t\geq 0}$  exhibits 
\[
\tau_{\rm rel}(T_u)  \leq  (10 R^2 \log R)^{15} \exp\left( 450 R/k\right).
\]
\end{proposition}

\noindent
For any   $T\in {\cal T}$ and any $u\in T$,   Proposition \ref{prop:UnBiasHighDegree4Trees} implies that if $u$ is a high-degree vertex
 or it is a low-degree vertex which is adjacent to high degree vertices then the spatial mixing assumption 
 \eqref{eq:SMAssumption} is true. 
We also have the following result whose proof appears in Section .\ref{sec:lemma:MixingLowDegreeStar}.
\begin{lemma}\label{lemma:MixingLowDegreeStar}
For $\epsilon, d, \Delta, k$ as in Theorem \ref{thrm:TreeRelaxationBound} the following is true:

Let $T\in {\cal T}$  and let $v\in T$. Consider $T_u$  and let $w_1, \ldots, w_{R}$  
be the children of the root, where  $R=\degree_{in}(v)$. 
Let  $(X_t)_{t\geq 0}$ be  the block dynamics with set of blocks ${\cal M}=\{\{v\}, T_{w_1}, \ldots, T_{w_R}\}$.
Assume that the degrees of $v, w_1, \ldots, w_{R}$ are at most  $\TDeg$.

Under any boundary condition at $\outBound T_u$,    $(X_t)_{t\geq 0}$ exhibits
\[
\tau_{\rm rel}(T_u)\leq 10^4 \exp\left( 5\max\{\log (R/(k-\TDeg)), 5\}\right) \log R
\]
\end{lemma}

\noindent
Note that   Lemma \ref{lemma:MixingLowDegreeStar} includes the case where  $u$ is such that $\degree_{out}(u)>0$, 
i.e., some of the neighbors of $u$ belong to $\outBound T_u$ and have  frozen color assignment.
Unifying  Lemma \ref{lemma:MixingLowDegreeStar} and Proposition \ref{prop:MixingStar}, 
we get the following: for any $T\in {\cal T}$ and any $u\in T$ we have that
 \begin{equation}\label{eq:BlockRelaxUnified}
\tau_{\rm rel}(T_u)\leq \exp\left( 450\left( \log(\degree(v)) +{\degree(u)}/{k}\right)\right).
\end{equation}

\noindent
In light of \eqref{eq:BlockRelaxUnified}, Theorem \ref{thrm:TreeRelaxationBound}
follows by combining  a simple induction and Proposition \ref{prop:SingleSiteVsBlock}.
Consider $T\in {\cal T}$. If $T$ is a single vertex, then $\tau_{\rm rel}(T)=1$. Assume, now, that the root $v$ of $T$ has children
$w_1, \ldots, w_{\ell}$, for some $\ell>0$. Then by the induction hypothesis we have that 
\begin{equation}\label{eq:IndHypoRelax}
\tau_{\rm rel}(T_{w_i}) \leq \exp\left( m(T_{w_i}, w_i)\right) \qquad \textrm{ for $i=1, \ldots, \ell$}. 
\end{equation}
Consider the block dynamics on $T$ where the blocks are, the root $v$ and the subtrees $T_{w_i}$.
The relaxation time for this process is given by \eqref{eq:BlockRelaxUnified}.
The theorem follows from \eqref{eq:BlockRelaxUnified}, \eqref{eq:IndHypoRelax} and Proposition \ref{prop:SingleSiteVsBlock}.

\subsection{Proof of Proposition \ref{prop:MixingStar}}\label{sec:prop:MixingStar}

Let $(X_t), (Y_t)$ be two copies of the {\em discrete time} block dynamics such that $X_0, Y_0$ are arbitrary $k$-colorings of $T$.
We present a coupling such that
after $ R^{5}\exp\left( 100 R/k\right)$ steps the probability of the event $X_t\neq Y_t$ is less than $e^{-1}$.

The coupling is such that we update the same block at each copy of the dynamics. When we update a
block are couple the configurations maximally, i.e., when we update block $B$ at time $t$, we minimize
the probability of the event $X_t(B)\neq Y_t(B)$.

Let $t_1, t_2, \ldots$ be the random times at which $u$ is updated in the coupling. 
For $i \geq 1$, we say that $t_i$ is a ``success" if the following hold:
\begin{enumerate}
\item $ |t_{i+1}-t_i| \geq 3R \max \{ \log(R/k), 5\}$
\item 
we have that
\begin{itemize}
	\item $|A_{X_{t_i}}(u) \oplus A_{Y_{t_i}}(u)| \leq 10$
	\item $\min\{ |A_{X_{t_i}}(u), A_{Y_{t_i}}(u)\} \geq 100$
\end{itemize}
\item the number of vertices $w_j$ such that $X_{t_{i}}(w_j)\neq Y_{t_{i}}(w_j)$ is less than $100R/k$.
\end{enumerate}

\begin{claim}\label{lemma:ConvergenceOnSuccess}

If $t_i$ is a success, for  $i\geq 1$, then there is a coupling such that 
$
\Pr \left [X_{t_{i+1}} \neq Y_{t_{i+1}}  \right] \leq e^{-2}.
$
\end{claim}
\begin{proof}

Consider the time interval ${\cal I}(t_i, t_{i+1})$. Note that if $X_{t_i}(v)=Y_{t_i}(v)$, then 
in the time interval ${\cal I}$ at every update of the block $T_{w_j}$ can be done by using
identical coupling. This means that for every $w_j$ whose block is updated at least once
during ${\cal I}$ we have $X_{t_{i+1}}(w_j)=Y_{t_{i+1}}(w_j)$. Thus, if there exists 
$w_{j}$ such that  $X_{t_{i+1}}(w_j) \neq Y_{t_{i+1}}(w_j)$, then this must have been a
disagreement created at some $t<t_i$ and survived during the time interval ${\cal I}$.

Let $W$ be the number of children $w_i$ which  disagree at time $t_i$ and they  are not updated during the
 interval ${\cal I}$. If there are no such disagreements  we set  $W=0$.
Clearly it holds that
\begin{equation}\label{eq:target4lemma:ConvergenceOnSuccess}
\Pr[X_{t_{i+1}}=Y_{t_{i+1}}] \leq \Pr[X_{t_i}(u)=Y_{t_i}(u), W=0]=\Pr[ W=0\ |\ X_{t_i}(u)=Y_{t_i}(u)]\Pr[X_{t_i}(u)=Y_{t_i}(u)].
\end{equation}

\noindent
Our assumption that $t_i$ is success implies that
\begin{equation}\label{eq:lemma:ConvergenceOnSuccessA}
\Pr[X_{t_i}(u)\neq Y_{t_i}(u)]\leq 1/10.
\end{equation}
Note that each block $T_{w_j}$ such that $X_{t_{i}}(w_j) \neq X_{t_{i}}(w_j)$ is update during the time
interval ${\cal I}$ with probability at least $1-\min\{ (R/k)^{-2}, e^{-15}\}$.  Markov's inequality imply that
\begin{equation}\label{eq:lemma:ConvergenceOnSuccessB}
\Pr[W>0\ |\ X_{t_i}(u)= Y_{t_i}(u)] \leq \min\{ (R/k)^{-1}, e^{-12}\}.
\end{equation}
The result follows by plugging \eqref{eq:lemma:ConvergenceOnSuccessA} and \eqref{eq:lemma:ConvergenceOnSuccessB}
into \eqref{eq:target4lemma:ConvergenceOnSuccess}.
\end{proof}

\noindent
We also have the following result whose proof appears in Section  \ref{sec:lemma:SuccessTime}.
\begin{lemma}\label{lemma:SuccessTime}
For any $t_i\geq 3R \log R$  we have that $ \Pr[\textrm{$t_i$ is success}]\geq \rho$, 
where  $$\rho \geq \exp\left( -15\max\{\log (R/k), 5\} -450R/k \right).$$
\end{lemma}

\noindent
Let $T=10^{4}  \left\lceil \rho^{-1} R\log R \right \rceil$, where $\rho$ is defined in Lemma \ref{lemma:SuccessTime}.
We consider  the time interval ${\cal I}=[0, T]$.  We partition  ${\cal I}$ into subintervals ${\cal I}_0, {\cal I}_1, \ldots$ 
each of length $4R \log R$.  
Lemma \ref{lemma:SuccessTime} implies that that the probability of having a success at  
${\cal I}_{j+2}$ is at least $\rho$, regardless of what happens in ${\cal I}_j$.

Noting that the probability that $v$ is updated during ${\cal I}_{2j}$ is greater than $1/2$, 
 the probability of having  $t_i\in {\cal I}_j$ which is success is at least $\rho/2$.
 
Let ${\cal E}$ be the event that there exists $j\geq 1$ such that ${\cal I}_{2j}$ there exists $t_i$ which is success.  Since there are at least 
$100/\rho$  subintervals to check, it is elementary to verify 
that 
\begin{equation}\label{eq:ProbabilityOfSuccesInCoupling}
\Pr[{\cal E}] \geq 1-e^{-5}.
\end{equation}
Let ${\cal C}$ be the event that in the coupling of $(X_t)$ and $(Y_t)$ there exists $t\in {\cal I}$ such that $X_t=Y_t$. Then, we
have that
\begin{equation}
\Pr[{\cal C}] \geq \Pr[{\cal C}\ |\ {\cal E}]\Pr[{\cal E}] \geq (1-e^{-2})(1-e^{-5})\geq 1-e^{-1}.
\end{equation}
In the above inequalities we substituted $\Pr[{\cal C}\ |\ {\cal E}]$ by using Claim \ref{lemma:ConvergenceOnSuccess} and $\Pr[{\cal E}]$
by using \eqref{eq:ProbabilityOfSuccesInCoupling}.

\subsubsection{Proof of Lemma \ref{lemma:SuccessTime}}\label{sec:lemma:SuccessTime}

Let ${\cal C}$ be the event that $|t_{i+1}-t_i |\geq 3R  \max \{ \log(R/k), 5\}$.
Also, let  ${\cal D}$ be the event that $t_i$ satisfies the requirements 2 and 3  to be ``success".
The lemma follows by noting that
\begin{equation}\label{eq:target4lemma:SuccessTime}
\rho\geq \Pr[{\cal C}] \Pr[{\cal D}].
\end{equation}
At each step the vertex $u$ is updated with probability $\frac{1}{R+1}$. Then  we have
\begin{eqnarray}\label{eq:CondA4Success}
\Pr[{\cal C}] =\left(1-{1}/({R+1}) \right)^{3R \max \{ \log(R/k), 5\} } \geq  \exp\left( -4\max\{\log (R/k), 5\} \right).
\end{eqnarray}

\noindent
For computing $\Pr[{\cal D}]$   we consider  cases regarding $R/k^2$ being larger or at most
$10^{-4}$.

\begin{claim}\label{claim:CondB4SuccessA}
For  $R/k^2 >10^{-4}$ we have that $\Pr[{\cal D}] \geq \exp\left( -450 R/k\right)$.
\end{claim}

\begin{claim}\label{claim:CondB4SuccessB}
For $R/k^2 \leq  10^{-4}$, we have that $\Pr[{\cal D}] \geq \exp\left( -10\max\{\log (R/k), 5\} -400R/k \right)$ 
\end{claim}

\noindent
The lemma follows by plugging  the bounds from 
\eqref{eq:CondA4Success} and Claims \ref{claim:CondB4SuccessA}, \ref{claim:CondB4SuccessB} into \eqref{eq:target4lemma:SuccessTime}.

It remains to show that  Claims \ref{claim:CondB4SuccessA}, \ref{claim:CondB4SuccessB} are indeed true. 

\begin{proof}[Proof of Claim \ref{claim:CondB4SuccessA}]
Let the interval ${\cal I}=[t_i-R, t_i)$ and let the set $W=\{1, 2, \ldots, 100\}$. Consider the following events: let  $\cal G$ be the event that for every
$t\in {\cal I}$ we have $W \subseteq A_{X_t}(u),A_{Y_t}(u) $. That is, the first 100 colors are available for the root
during the whole interval ${\cal I}$. Let ${\cal S}$ be the event that
at time $t_i$ there are $q_1,q_2 \in [k]$ such that $A_{X_{t_i}}(u)=W\cup\{q_1\}$ and
$A_{Y_{t_i}}(u)=W\cup\{q_2\}$. That is, the two sets differ only on at most two colors. 
Let ${\cal N}$ be the event that $u$ is not updated during interval ${\cal I}$.
Finally let ${\cal Z}$ be the event that the number of disagreeing children of $u$ is at most
$100 R/k$.

It is direct that if the events ${\cal G}$,  ${\cal N}$,  ${\cal S}$ and ${\cal Z}$ hold then the event ${\cal D}$ holds.  That is,
$\Pr[{\cal D}] \geq \Pr[{\cal G}, {\cal S}, {\cal N}, {\cal Z}]$.
More specifically, we have
\begin{equation}\label{eq:Target4claim:CondB4SuccessA}
\Pr[{\cal D}] \geq \Pr[{\cal N}] \Pr[{\cal G}\ |\ {\cal N}] \Pr[{\cal S}\ |\ {\cal G}, {\cal N}] \Pr[{\cal Z}\ |\ {\cal G}, {\cal N}, {\cal S}].
\end{equation}
The claim follows by bounding appropriately the probability terms on the r.h.s. of \eqref{eq:Target4claim:CondB4SuccessA}.

We start with $\Pr[{\cal N}]$. Using standard coupon collector argument it is elementary to verify that
\begin{equation}\label{eq:CalNBound}
\Pr[{\cal N}]\geq e^{-2}.
\end{equation}

\noindent
We proceed by considering $\Pr[{\cal G}\ |\ {\cal N}]$. Conditional on ${\cal N}$,  in the coupling of $(X_t)$ and $(Y_t)$ we have that updating block $T_{w_j}$,  each color in $W$ is not used for both 
$X_t(w_j)$ and $Y_t(w_j)$ with probability at least $1-{2}/{k}$. 
Conditioning that all the blocks $T_{w_j}$s are updated prior to time $t_i-R$,  consider the last time that each $T_{w_j}$ is updated
 prior to $t_i-R$.  The probability that non of the colors in $W$ is used for the children $w_1, \ldots, w_{R}$ at time $t_i-R$, is at least 
$\left(1-{200}/{k} \right)^{R}\geq \exp\left( -200R/k\right).$ 
The probability that non of the following $R$ updates  uses any color from $W$ for
$w_1, \ldots, w_{R}$ is at least $\left(1-{200}/{k} \right)^{R}\geq  \exp\left( -200R/k\right).$
From the above we conclude that 
$$\Pr[{\cal G}\ | \ {\cal N}, {\cal Q}]\geq \exp\left( -400R/k \right),$$ where 
${\cal Q}$ is the event that there is no block $T_{w_j}$ which is not updated at least once prior to time $t_i-R$.
Furthermore,  we get that $$\Pr[\bar{\cal Q}\ |\ {\cal N}] \leq \Delta \ \Pr[{T}_{w_1} \textrm{ is not updated by time } t_i-R \ |\ {\cal N}] \leq  2R^{-1},$$
since $t_i-R> 2R\log R$.
Or, $\Pr[{\cal Q}\ |\ {\cal N}]\geq 1/2$.  
We have that
\begin{equation}\label{eq:CalGCondBound}
\Pr[{\cal G} \ | \ {\cal N}]\ \geq\  \Pr[{\cal G} \ | \ {\cal N}, {\cal Q}]\Pr[{\cal Q}\ |\ {\cal N}]  \ \geq \ 2^{-1}\exp\left( -400R/k\right).
\end{equation}
We proceed by considering  $\Pr[{\cal S}\ |\ {\cal G}, {\cal N}]$. 
Conditional on ${\cal N}$ and ${\cal G}$, at each block update $T_{w_j}$ some color $q$ is assigned to vertex $w_j$ with
probability at most $2/k$.  
Assume that at time $t=t_i-\Delta$ we have $X_t(v)=q_1$ and $Y_t(v)=q_2$.
For  $(X_t)$, we call available colors the set of colors $[k]\setminus (W\cup \{q_1\})$.
Similarly, for $(Y_t)$, we call available colors the set of colors $[k]\setminus (W\cup \{q_2\})$.
Let $K$ be the number of colors  which are available in some chain and they are not used from any of 
$w_1, \ldots, w_{R}$ in  the corresponding  chain at time $t_i$.

Assume that at time $t\in {\cal I}$ we update  block $T_{w_j}$.  Recall that for any  $t\in {\cal I}$ we have 
$X_t(v)=q_1$ and $Y_t(v)=q_2$.  We couple $X_t(w_j)$ and $Y_t(w_j)$ such that
for each $q\in [k]\setminus (W\cup \{q_1, q_2\})$ we set $\Pr [X_t(w_j)=Y_t(w_j)=q]$ with probability
$\min\{ \Pr [X_t(w_j)=q], \Pr [Y_t(w_j)=q] \}$, while for the colors $q_1, q_2$ we have
$\Pr [X_t(w_j)=q_2, Y_t(w_j)=q_1]$ with probability $\min\{ \Pr [X_t(w_j)=q_2], \Pr [Y_t(w_j)=q_1] \}$.
The aforementioned coupling is what we call, ``maximal coupling".

The above implies that if at time $t\in {\cal I}$ the coupling updates block $T_{w_j}$, each available
color is  used for $w_j$ with probability at most  $2/k$.  This implies that some available color is not
used at all for coloring any of the vertices in $w_1, \ldots, w_{\Delta}$ during the period ${\cal I}$ with
probability at least $\left(1-2/k \right)^{|{\cal I}|}=\left(1-2/k \right)^{R}$.
The linearity of expectation yields
$$
\ExpCond{K}{{\cal G}, {\cal N}} \leq k \left(1-2/k \right)^{R} \leq k\exp\left( -{3R}/({2k})\right).
$$
Markov's inequality implies that 
\begin{equation}\label{eq:CalSCondBound}
\Pr[{\cal S}\ |\ {\cal G}, {\cal N}]\geq1- \ExpCond{K}{{\cal G}, {\cal N}}\geq 1-k\exp\left( -2R/k\right)\ \geq \ 1-k\exp\left( -10^{-4}k\right)\geq 1/2,
\end{equation}
where the third inequality follows from the assumption that $R/k^2>10^{-4}$.

Letting ${\cal R}$ be the number of disagreements at the vertices $w_1, \ldots, w_{R}$, at time $t_i$, elementary calculations
yield that $\ExpCond{{\cal R}}{{\cal G}, {\cal S}, {\cal N}} \leq 10\left( R /k \right)$. Then, Markov's inequality give
\begin{equation}\label{eq:CalZCondBound}
\Pr[{\cal Z}\ |\ {\cal G}, {\cal S}, {\cal N}] \geq 1-\frac{\ExpCond{{\cal R}}{{\cal G}, {\cal S}, {\cal N}}}{100\left( R/k \right)} \geq 1/2.
\end{equation}
Plugging \eqref{eq:CalNBound}, \eqref{eq:CalGCondBound}, \eqref{eq:CalSCondBound} and \eqref{eq:CalZCondBound} into 
\eqref{eq:Target4claim:CondB4SuccessA} and using  that $\exp(R/k)>10^4$, the claim follows.
\end{proof}

\begin{proof}[Proof of Claim \ref{claim:CondB4SuccessB}]
Let $\hat{t}=t_i-3R \min\{ \log(R/k), 5 \}$. Let the time interval ${\cal I}=(\hat{t}, t_i)$.

We consider the following event. Let ${\cal A}$ be the event that $u$ is not updated during the interval ${\cal I}$.
Let ${\cal R}_1$ be the event that the number of disagreements  on the vertices $w_1, \ldots, w_{R}$, at time $\hat{t}$,
is less than $10^{3}R/k$.
Also, let ${\cal R}_2$ be the event that the number of disagreements  on the vertices $w_1, \ldots, w_{\Delta}$, 
at time $t_i$,
is less than $10^{3}\Delta/k$, while there is $q_1, q_2, \in [k]$ such that for each $w_j$ such that $X_{t_i}(w_j)\neq Y_{t_i}(w_j)$ we have 
$X_{t_i}(w_j), Y_{t_i}(w_j)\in \{q_1, q_2\}.$  Note that this requirement implies that $A_{X_{t_i}}(u)$, $A_{Y_{t_i}}(u)$ differ only in at
most two colors. 
Finally let ${\cal G}$ be the event that non of the colors in  $W=\{1,2,\ldots, 100\}$ is used by any of the
children of $v$ at time $t_i$, in both chains.

It is elementary to show that if the events ${\cal A}, {\cal R}_1, {\cal R}_2, {\cal G}$ occur, then the event ${\cal D}$ also occurs.
That is, we have that
\begin{equation}\label{eq:Target4claim:CondB4SuccessB}
\Pr[{\cal D}] \geq  \Pr[{\cal A}] \Pr[{\cal G}\ |\ {\cal A}] \Pr[{\cal R}_1\ |\ {\cal G}, {\cal A}]\Pr[{\cal R}_2\ |\ {\cal G}, {\cal A}, {\cal R}_1].
\end{equation}

\noindent
Working as for \eqref{eq:CondA4Success} we have that
\begin{equation}\label{eq:CalABoundB}
\Pr[{\cal A}]\geq \exp\left( -4\max\{\log (R/k), 5\} \right).
\end{equation}
Also, working as in \eqref{eq:CalGCondBound} we get that
\begin{equation}\label{eq:CalGBoundCondB}
\Pr[{\cal G}\ |\ {\cal A}] \geq \exp\left( -150R/k\right).
\end{equation}

\noindent
Let $K_1$ be the number of disagreements on the vertices $w_1, \ldots, w_{\Delta}$ at time $\hat{t}$.
Also,  let ${\cal U}$ be the event that there does not exist $w_j$ such that $T_{w_j}$ is not updated prior
to $\hat{t}$.

Conditioning on the events ${\cal A}, {\cal G}$,   since we couple the two copies $(X_t)$ and $(Y_t)$ maximally, we have that each time 
we update a block $T_{w_j}$ the probability of having a disagreement at $w_j$, which is bounded by the probability of the most likely color, 
is less than $2/k$.  Then, 
conditional on the event ${\cal U}$, the time at which  $w_j$ is updated for last time, prior to $\hat{t}$ becomes disagreeing
with probability $2/k$.  From the linearity of expectation we have that
\[
\ExpCond{K_1}{{\cal U}, {\cal A}, {\cal G} }\leq 2R/k.
\] 
Since  $\hat{t}\geq 2R \log R$ we have that $\Pr[{\cal U} \ |\ {\cal A}, {\cal G}]\geq 1/2$. Then we get that
\begin{equation}\label{eq:CalR1CondBoundB}
\Pr[{\cal R}_1\ |\ {\cal A}, {\cal G}]\geq \ \Pr[{\cal R}_1 \ |\ {\cal U}, {\cal A}, {\cal G}]  \Pr[{\cal U} \ |\ {\cal A}, {\cal G}] \ 
\geq \ 2^{-1}\left(1- \frac{\ExpCond{K_1}{{\cal U},{\cal A}, {\cal G} }}{10^3 \Delta/k} \right) \geq 1/3,
\end{equation}
where the second derivation follows from Markov's inequality.

We proceed by bounding $\Pr[{\cal R}_2 \ |\ {\cal R}_1, {\cal A}, {\cal G}]$.
Let $Z$ be the number of blocks such that $X_{\hat{t}}(w_j)\neq Y_{\hat{t}}(w_j)$ and the block $T_{w_j}$ is not updated
during the interval ${\cal I}$. 
Let ${\cal Z}$ be the event $Z=0$.
Conditional on ${\cal R}_1, {\cal A}$ and ${\cal G}$, the choice of the block update at time $t\in {\cal I}$
is uniformly random among all the blocks but $\{u\}$.  Since each block is not updated during ${\cal I}$ with probability at least 
$\exp\left( -4\max\{\log (R/k), 5\} \right)$. We get that
\[
\ExpCond{Z}{{\cal R}_1, {\cal A}, {\cal G}} \leq 100(R /k)\exp\left( -4\max\{\log (R /k), 5\} \right) \leq  \min\{(\Delta/k)^{-2}, e^{-10}\}.
\]
The above  with Markov's inequality imply that
\begin{equation}
\Pr[{\cal Z} \ |\ {\cal R}_1, {\cal A}, {\cal G}]\geq 1-\min\{(\Delta/k)^{-2}, e^{-10}\} \geq 1/2.
\end{equation}

\noindent
Assume that at time $t\in {\cal I}$ we update  block $T_{w_j}$.  Recall that for every  $t\in {\cal I}$ we have 
$X_t(v)=q_1$ and $Y_t(v)=q_2$.  We couple $X_t(w_j)$ and $Y_t(w_j)$ such that
for each $q\in [k]\setminus (W\cup \{q_1, q_2\})$ we set $\Pr [X_t(w_j)=Y_t(w_j)=q]$ with probability
$\min\{ \Pr [X_t(w_j)=q], \Pr [Y_t(w_j)=q] \}$, while for the colors $q_1, q_2$ we have
$\Pr [X_t(w_j)=q_2, Y_t(w_j)=q_1]$ with probability $\min\{ \Pr [X_t(w_j)=q_2], \Pr [Y_t(w_j)=q_1] \}$.
The aforementioned coupling is what we call, ``maximal coupling".

Conditional on the events ${\cal A}, {\cal Z}, {\cal R}_1, {\cal G}$, the above coupling implies that
two kinds of disagreements on some vertex $w_j$  can be generated at time $t\in {\cal I}$.
The first kind involves having $X_t(w_j)=q_2=Y_t(v)$ and $Y_t(w_j)=q_1=X_t(v)$. The second
kind of disagreement involves all the rest. Note that the first kind of disagreement 
occurs with probability at most $2/k$ when we update $T_{w_j}$. Furthermore, 
the second disagreement appears due to the fact that the distributions of $X_t(w_j), Y_t(w_j)$ are not
perfectly uniform over $[k]\setminus \{q_1\}$ and $[k]\setminus \{q_2\}$, respectively. It is elementary
to show that the disagreements of the second kind occur at each update with probability less than $200/k^2$.

Let $F$ be the number of disagreements of the second kind on $w_1, \ldots, w_{R}$ at time $t_i$.
Let ${\cal F}$ be the event that $F=0$.  Since the expected number of such disagreements is $200R/k^2$, Markov's inequality imply that 
$\Pr[{\cal F}\ |\ {\cal A}, {\cal G}, {\cal R}_1, {\cal Z}]\geq 1/2.$ Then we have that
\begin{equation}\label{eq:CallFCondBoundB}
\Pr[{\cal F}\ |\ {\cal A}, {\cal G}, {\cal R}_1] \geq \Pr[{\cal F}\ |\ {\cal A}, {\cal G}, {\cal R}_1, {\cal Z}] \Pr[{\cal Z} \ |\ {\cal A}, {\cal G}, {\cal R}_1] \geq 1/4.
\end{equation}
When the event ${\cal F}$ holds, then we have that $A_{X_{t_i}}(v)\oplus A_{Y_{t_i}}(v)=\{q_1, q_2\}$.

Let $K_2$ be the number of disagreements at vertices $w_1, \ldots, w_{R}$.
Conditional on ${\cal F}$, $K_2$ is equal to the number of vertices $w_j$ such that $X_{t_i}(w_j)= q_2$. 
Then, it is elementary to verify that each time a vertex $w_j$ is updated we have $X_{t_i}(w_j)= q_2$
with probability less than $2/k$, conditional on the events ${\cal A}, {\cal G}, {\cal R}_1, {\cal F}$.
The expected $K_2$ is at most $2R/k$. Then,  Markov's inequality imply that $\Pr[K> 10^3 (R/k)\ |\ {\cal A}, {\cal G}, {\cal R}_1, {\cal F} ]\geq 1/3$.
Since ${\cal R}_2$ occurs only if ${\cal F}$ occurs and $K_2<10^3(R/k)$, we get that
\begin{equation}\label{eq:CalR2CondBoundB}
\Pr[{\cal R}_2\ |\ {\cal A}, {\cal G}, {\cal R}_1] \geq 1/20.
\end{equation}
Plugging \eqref{eq:CalR2CondBoundB}, \eqref{eq:CalR1CondBoundB}, \eqref{eq:CalGBoundCondB} and \eqref{eq:CalABoundB} into 
\eqref{eq:Target4claim:CondB4SuccessB}, the claim follows.
\end{proof}

\subsection{Proof of Lemma \ref{lemma:MixingLowDegreeStar}}\label{sec:lemma:MixingLowDegreeStar}

The  proof of Lemma \ref{lemma:MixingLowDegreeStar} is not too different than that of Proposition \ref{prop:MixingStar}.
The only difference now is the lack of condition \eqref{eq:SMAssumption} which implied a certain kind of symmetry between
the color assignment of each $w_j$. That is, if at time $t$ we update $w_j$ then for any $q_1, q_2 \in A_{X_t}(w_j)$ we have that 
$\Pr[X_t(w_j)=q_1]\approx \Pr[X_t(w_j)=q_2]$.  For this proof  the bounds we assume are  $1/k<\Pr[X_t(w_j)=q_1]\leq 1/(k-\TDeg)$.
Note that we may have that the root $u$ is incident to some vertices in $\outBound T_u$.

The case where $R=\degree_{in}(u)$ is too low, i.e., $R<k/3$ follows directly 
by applying path coupling.  For what follows, we assume that $k/3\leq \degree_{in}(u)\leq \TDeg$.

Consider discrete time block dynamics  $(X_t)$ and $(Y_t)$. Assume that $X_0, Y_0$ are arbitrary $k$-colorings of $T_u$.
We present a coupling such that after $t> 10^4 \exp\left( 5\max\{\log (R/(k-\TDeg)), 5\}\right)R\log R$ steps we have  $\Pr[X_t\neq Y_t] \leq e^{-1}$.
Then the bound for relaxation time of the continuous version follows immediately.

The coupling is such that we update the same block at each copy of the dynamics. When we update a
block are couple the configurations maximally, i.e., when we update block $B$ at time $t$, we minimize
the probability of the event $X_t(B)\neq Y_t(B)$.

Let $t_1, t_2, \ldots$ be the random times at which $v$ is updated in the coupling. 
For $i \geq 1$, we say that $t_i$ is a ``success" if the following hold
\begin{enumerate}
\item $ |t_{i+1}-t_i| \geq 3R \max \{ \log(R/(k-\TDeg) ), 5\}$
\item we have that 
\begin{itemize}
	\item $|A_{X_{t_i}}(u) \oplus A_{Y_{t_i}}(u)| \leq 500$
	\item $\min\{ |A_{X_{t_i}}(u), A_{Y_{t_i}}(u)\} \geq 10^5.$
\end{itemize}
\item The number of vertices $w_j$ such that $X_{t_{i}}(w_j)\neq Y_{t_{i}}(w_j)$ is less than $100R/(k-\TDeg)$.
\end{enumerate}

\noindent
Working as in Claim \ref{lemma:ConvergenceOnSuccess} we get the following:
If for some $i\geq 1$ we have  $t_i$ that is ``success", then  there is a coupling
such that
$
\Pr[X_{t_{i+1}}\neq Y_{t_{i+1}}]\leq e^{-2}.
$

We are going to show that for   any $t_i\geq 3R \log R$  we have that $ \Pr[\textrm{$t_i$ is success}]\geq \rho$, 
where  $$\rho \geq \exp\left( -5\max\{\log (R/(k-\TDeg)), 5\}  \right).$$
Then,  the lemma will follow working as in the proof of Proposition \ref{prop:MixingStar}. 
That is, we show that in the time interval $[0, \hat{T}]$, where
$\hat{T}=10^4 \left\lceil \rho^{-1}R\log R \right \rceil $, the probability of having $t_i$ which is large, i.e., greater than
$1-e^{-5}$

Let ${\cal C}$ be the event that $|t_{i+1}-t_i |\geq 3\Delta \max \{ \log(R/(k-\TDeg)), 5\}$.
Let ${\cal D}$ be the event that $T_i$ satisfies the requirements 2 and 3  to be ``success".
The lemma follows by noting that
\begin{equation}\label{eq:target4lemma:SuccessTime}
\rho\geq \Pr[{\cal C}] \Pr[{\cal D}].
\end{equation}

\noindent
At each step the vertex $v$ is updated with probability $\frac{1}{R+1}$. Then  we have
\begin{eqnarray}\label{eq:CondA4SuccessLowDegree}
\Pr[{\cal C}] =\textstyle \left(1-{1}/({R+1}) \right)^{3R \max \{ \log(R/(k-\TDeg)), 5\} } \geq  \exp\left( -4\max\{\log (R/(k-\TDeg)), 5\} \right).
\end{eqnarray}

\noindent
For computing $\Pr[{\cal D}]$, we  let $Z$ be the number of disagreements in the
set of vertices $w_1, \ldots, w_R$,  at time $t_i$.
The requirement that both $A_{X_{t_i}}(u), A_{Y_{t_i}}(u)$ are sufficiently
large is trivially satisfied since we assume that $R\leq \TDeg$ and $k>(3/2)\TDeg$. 
Furthermore, given  $Z$,  it is elementary to see that the disagreements at the vertices
in $w_{1}, \ldots w_{R}$  involve at most $2Z$ different colors, i.e., 
$|A_{X_{t_i}}(u) \oplus A_{Y_{t_i}}(u)| \leq 2Z+2$.
With the above observations, it is elementary to verify that the event holds once we have 
$Z < 90 R/(k-\TDeg)$. That is,  
$$\Pr[{\cal D}] \geq \Pr[Z< 90 R/(k-\TDeg)].$$
Let ${\cal U}$ be the event that the block of every $w_i$ is updated at least once. 
Each time the block $T_{w_j}$ is updated we have a disagreement with probability less than
$1/(k-\TDeg)$.   Markov's inequality implies 
\[
\Pr[Z\geq 100 R/(k-\TDeg) \ |\ {\cal U}] \leq \frac{\ExpCond{Z}{{\cal U}}}{100\Delta/(k-\TDeg)}\leq 1/50.
\]
Since  $t_i\geq 3R\log R$, we get that $\Pr[{\cal U}]\geq 3/4$. Combining all the above,
we get that
\begin{equation}\label{eq:CondB4SuccessLowDegree}
\Pr[{\cal D}] \geq \Pr[Z< 100 R/(k-\TDeg) \ | \ {\cal U}]\Pr[{\cal U}] \geq 1/2.
\end{equation}
The lemma follows by plugging \eqref{eq:CondB4SuccessLowDegree}, \eqref{eq:CondA4SuccessLowDegree}
to \eqref{eq:target4lemma:SuccessTime}.

\section{Proof of Lemma \ref{thrm:RelaxationWeightBound}}\label{sec:thrm:RelaxationWeightBound}

\noindent
Rewriting ${\cal J}(P)$ we have that 
\begin{equation}\label{eq:Target4thrm:RelaxationWeightBound}
\textstyle {\cal J}(P)=450\left(  \log\prod_{u}\degree(u)+k^{-1}\sum_{u}\degree(u)\right).
\end{equation}
The theorem will follow by bounding appropriately the above sum and the product.

As far as the product of the degree is concerned, let the set $M$ contain every vertex $u\in P$
such that $\degree(u)>\TDeg$.  Note that the choice of $P$ implies that at least one the end vertices of
the path is either a break-point or it is adjacent to one. Then, from Corollary \ref{cor:FromBreakPointProd} we have that
\[
\textstyle \prod_{u\in M}\degree(u)\leq \left(1+\epsilon\right)^{\ell}
\]
Since we trivially have that $\prod_{u\in P\setminus M}\degree(u)\leq (\TDeg)^{\ell}$, we get that
\begin{equation}\label{eq:ProdBound4Relax}
\textstyle \log \left( \prod_{u}\degree(u) \right)\leq 2\ell \log d.
\end{equation}
As far as the sum of degrees over $P$ is concerned, we use the following claim.

\begin{claim}\label{claim:PathDensityGnp}
With probability $1-10n^{-d/(\log d)^2}$, the graph $G(n,d/n)$ has no path $P$ of  length at most
$\log n/(\log d)^4$ such 
\begin{equation}\label{eq:HighDensityProp}
\textstyle k^{-1}\sum_{u\in P}\degree(u) \geq \frac{5 \log n}{(\log d)^2}.
\end{equation}
\end{claim}
\begin{proof}
We are showing the property for paths  of length, exactly, $\log n/(\log d)^4$.
The claim follows by noting that if  a path  $P$ does not satisfy \eqref{eq:HighDensityProp} 
then no subpath of $P$  satisfies \eqref{eq:HighDensityProp}.

Letting $Z$ be the number of paths in $G(n,d/n)$ that  satisfy \eqref{eq:HighDensityProp}, a 
simple derivation gives
\begin{equation} \label{eq:2stMomementDensePaths }
\Exp{ Z} =(1-o(1)) n d^{\ell} p_{\ell},
\end{equation}
where $\ell=\log n/(\log d)^5$ and 
 $p_{\ell}$ is the probability that  a path $P$ in $G(n,d/n)$ of length $\ell$
satisfies \eqref{eq:HighDensityProp}. 
The claim follows by showing that  $p_{\ell}\leq 2n^{-d/(\log d)^2}$.

Given some path $P$, let $A_{ext}$ be the number of edges between a vertex in $P$ and
some vertex outside $P$. Also, let $A_{int}$ be the number of edges between non consecutive
vertices in $P$. Since $k>d$,  we have 
\begin{equation}\label{eq:plVsAextAint}
p_{\ell } \leq \Pr[A_{ext} \geq  (d/(\log d)^2)\log n ]+ \Pr[A_{int} \geq (d/(\log d)^2) \log n ].
\end{equation}
Clearly $A_{ext}$ is dominated by the binomial distribution with parameters
$n(\ell+1)$ and $d/n$.  Similarly, we note that $A_{int}$ is dominated by  the binomial distribution with parameters
$(\ell+1)^2/2$ and $d/n$. From Chernoff's bound we get that
\begin{eqnarray}
\Pr[A_{ext}\geq  (d/(\log d)^2) \log n ] \ \leq \ n^{-d/(\log d)^2} \quad \textrm{and}\quad 
\Pr[A_{int }\geq  (d/(\log d)^2 \log n ]  \leq \ n^{-d/(\log d)^2}. \label{eq:TailAint}
\end{eqnarray}
Plugging   \eqref{eq:TailAint} into \eqref{eq:plVsAextAint} we get
that $p_{\ell}\leq 2n^{-d/(\log d)^2}$. The claim follows.
\end{proof}

The lemma follows by combining Claim \ref{claim:PathDensityGnp}, \eqref{eq:ProdBound4Relax} and
\eqref{eq:Target4thrm:RelaxationWeightBound}.


\end{document}